\documentclass[journal,11pt, onecolumn]{IEEEtran}

\ifCLASSINFOpdf
\else
   \usepackage[dvips]{graphicx}
\fi
\usepackage{url}
\usepackage{algorithmicx}
\usepackage{algpseudocode}
\usepackage{algorithm}
\algdef{SE}[DOWHILE]{Do}{doWhile}{\algorithmicdo}[1]{\algorithmicwhile\ #1}%
\usepackage{amsmath}
\usepackage{dsfont, cuted}
\usepackage{amsfonts}
\usepackage{amsthm}
\usepackage[utf8]{inputenc}
\usepackage{breqn}
\usepackage{xcolor}
\usepackage{subfigure}
\usepackage{multirow}
\usepackage{multicol}
\usepackage{booktabs}
\usepackage{fullwidth}
\usepackage[noadjust]{cite}
\usepackage{setspace}
\usepackage{hyperref}
\usepackage{bmpsize}
\usepackage{graphicx}
\newtheorem{thm}{Theorem}

\newtheorem{rem}{Remark}
\usepackage{multirow,tabularx}
\newcolumntype{C}{>{\centering\arraybackslash}X}
\newcolumntype{L}{>{\raggedright \arraybackslash}X}
\newcolumntype{R}{>{\raggedleft \arraybackslash}X}
\usepackage{tikz}
\usepackage{lipsum}

\newcommand\copyrighttext{%
  \footnotesize \copyright 2021 IEEE.  Personal use of this material is permitted.  Permission from IEEE must be obtained for all other uses, in any current or future media, including reprinting/republishing this material for advertising or promotional purposes, creating new collective works, for resale or redistribution to servers or lists, or reuse of any copyrighted component of this work in other works.}
\newcommand\copyrightnotice{%
\begin{tikzpicture}[remember picture,overlay]
\node[anchor=south,yshift=10pt] at (current page.south) {\fbox{\parbox{\dimexpr\textwidth-\fboxsep-\fboxrule\relax}{\copyrighttext}}};
\end{tikzpicture}%
}

\begin{document}

\title{A Generalized Gaussian Extension to the Rician Distribution for SAR Image Modeling}

\author{Oktay~Karakuş, Ercan~E.~Kuruoglu, Alin~Achim
        \thanks{This work was supported by the UK Engineering and Physical Sciences Research Council (EPSRC) under grant EP/R009260/1 (AssenSAR). Oktay Karakuş and Alin Achim are with the Visual Information Laboratory, University of Bristol, Bristol BS1 5DD, U.K. Ercan E. Kuruoglu is with Data Science and Information Technology Center, Tsinghua-Berkeley Shenzhen Institute, China and is on leave from ISTI-CNR, Pisa, Italy.}
}
\maketitle
\begin{abstract}
We present a novel statistical model, \textit{the generalized-Gaussian-Rician} (GG-Rician) distribution, for the characterization of synthetic aperture radar (SAR) images. Since accurate statistical models lead to better results in applications such as target tracking, classification, or despeckling, characterizing SAR images of various scenes including urban, sea surface, or agricultural, is essential. The proposed statistical model is based on the Rician distribution to model the amplitude of a complex SAR signal, the in-phase and quadrature components of which are assumed to be generalized-Gaussian distributed. The proposed amplitude GG-Rician model is further extended to cover the intensity SAR signals. In the experimental analysis, the GG-Rician model is investigated for amplitude and intensity SAR images of various frequency bands and scenes in comparison to state-of-the-art statistical models that include Weibull, $\mathcal{G}_0$, Generalized gamma, and the lognormal distribution.  The statistical significance analysis and goodness of fit test results demonstrate the superior performance and flexibility of the proposed model for all frequency bands and scenes, and its applicability on both amplitude and intensity SAR images. The Matlab package is available at \url{https://github.com/oktaykarakus/GG-Rician-SAR-Image-Modelling}
\end{abstract}


\IEEEpeerreviewmaketitle
\copyrightnotice
\section{Introduction}
\IEEEPARstart{S}{ynthetic} aperture radar (SAR) imagery is an important source of information in the analysis of various terrains thanks to its capability to capture wider areas under different weather conditions. Statistical modeling of SAR images plays an essential role in characterizing various scenes and underpins applications such as classification \cite{tison2004new, mejail2003classification}, or denoising \cite{achim2006sar, argenti2013tutorial}. The literature abounds with numerous statistical models for different SAR scenes, which are either based on the physics of the imaging process or empirical, and all these models have advantages and disadvantages according to the scene and/or frequency band employed.

In this paper, we address the problem of accurately modeling the SAR amplitude/intensity data within the context of probability density function (pdf) estimation by assuming that the back-scattered SAR signal components possess heavy-tailed non-Gaussian nature. Specifically, we propose a generic and flexible  statistical model, in order to cover various   characteristics of the back-scattered SAR signal which will benefit applications such as despeckling, classification, or segmentation.

The standard SAR model defines the back-scattered SAR signal received by a SAR sensor as a complex signal $R = x + iy$, where $x$ and $y$ are the real and imaginary parts respectively, and follows several assumptions \cite{kuruoglu2004modeling,moser2006sar}:
\begin{enumerate}
    \item The number of scatterers is large,
    \item The scatterers are statistically independent,
    \item The instantaneous scattering phases are statistically independent of the amplitudes,
    \item The phase is uniformly distributed,
    \item The reflectors are relatively small when compared to the illuminated scene,
    \item There is no dominating scatterer in the scene.
\end{enumerate}

In particular, the first two assumptions recall the central limit theorem whereby the real and imaginary parts are jointly Gaussian. Combined with assumption 6), this leads to the case where $x$ and $y$ are independent and identically distributed (i.i.d.) zero-mean Gaussian random variables,
\begin{align}
    x \sim \mathcal{N}(0, \sigma^2)\quad \text{and} \quad y  \sim \mathcal{N}(0, \sigma^2).
\end{align}

Thence, the amplitude distribution becomes the \textit{Rayleigh} distribution, the probability density function (pdf) of which is given by
\begin{align}\label{equ:Rayl}
  f(r|\sigma) = \dfrac{r}{\sigma^2}\exp\left(-\dfrac{r^2}{2\sigma^2}\right)
\end{align}
where $r = \sqrt{x^2 + y^2}$ refers to the amplitude and $\sigma$ is the scale parameter. The scattering mechanism of this kind is depicted in Figure \ref{fig:scatter}-(a).

In various scenes, the illuminated area may include one (or a small number of) dominating scatterer(s) (Figure \ref{fig:scatter}-(b)), and a large number of non-dominant ones \cite{yuesurvey2020}. Hence, the assumption 6) may no longer be valid, and $x$ and $y$ become iid, but non-zero-mean ($\delta$) Gaussian random variables as
\begin{align}
    x \sim \mathcal{N}(\delta, \sigma^2)\quad \text{and} \quad y \sim \mathcal{N}(\delta, \sigma^2).
\end{align}
where $\delta > 0$ is the non-zero mean of the components $x$ and $y$ \cite{devore2000atr}, which also defines the relationship between the dominating scatterer and  the statistical model. Hence, progressing as in the derivation of Rayleigh model and making transformation to polar coordinates we obtain the Rician distribution
\begin{align}\label{equ:Rice}
  f(r|\gamma, \Delta) = \frac{r}{\sigma^2}\exp\left(-\frac{r^2 + \Delta^2}{2\sigma^2}\right) \mathcal{I}_0\left(\frac{r\Delta}{\sigma^2} \right)
\end{align}
where $\Delta = \sqrt{2}\delta$ is the location parameter and $\mathcal{I}_0(\cdot)$ refers to the zeroth-order modified Bessel function of the first kind.

Even though they are theoretically appealing and analytically simple, Gaussian/Rayleigh based statistical models do not reflect the real life phenomena in most cases for SAR reflections. Thus, there are numerous statistical models in the literature which were developed to account for non-Rayleigh cases, and proven to be successful for modeling SAR imagery. Among those, the Gamma distribution is an important statistical model for characterizing multi look SAR intensity images \cite{sun2018dependence, cui2014comparative}. It is the generalization of the exponential distribution via averaging $L$ single-look SAR intensities, each of which are exponentially distributed. The pdf expression for the Gamma distribution is given as
\begin{align}
    f(\nu|L, \gamma) = \dfrac{(\gamma L)^L}{\Gamma(L)} \nu^{L-1} \exp\left( -\gamma L \nu\right),
\end{align}
where $\nu = r^2$ is the intensity SAR signal, $\gamma$ is the scale parameter and $\Gamma(\cdot)$ refers to the Gamma function.

\begin{figure}[t!]
    \centering
    \subfigure[]{\includegraphics[width=.49\linewidth]{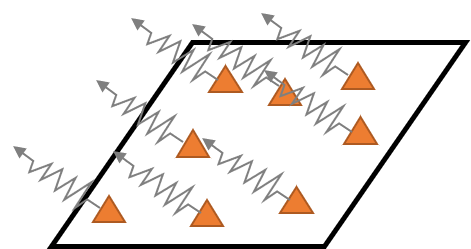}}
    \subfigure[]{\includegraphics[width=.49\linewidth]{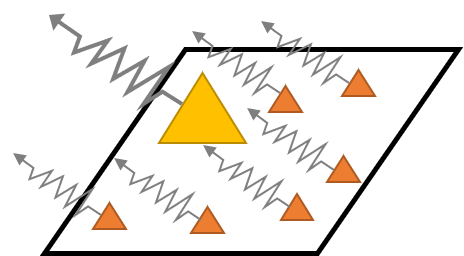}}
    \caption{Surface scattering examples for radar signals. (a) Distributed scattering from multiple scatterers, (b) A scene includes a single dominating scatterer as well as multiple distributed ones.}\vspace{-0.5cm}
    \label{fig:scatter}
\end{figure}

Contrary to the theoretical models discussed above, the Weibull distribution is an empirical statistical model, and has been used in the literature \cite{chitroub2002statistical, ishii2011effect, fernandez2015estimation} to model SAR images in both amplitude and intensity formats. The Weibull pdf is expressed as
\begin{align}
f(r|\alpha, \gamma) = \dfrac{\alpha}{\gamma} \left(\dfrac{r}{\gamma}\right)^{\alpha-1}\exp\left(-\left(\dfrac{r}{\gamma}\right)^{\alpha}\right)
\end{align}
where $\alpha$ refers to the shape parameter, and $\gamma$ is the scale parameter. The Lognormal distribution is another empirical model like Weibull and has generally been used to model SAR amplitude images \cite{chen2013analysis, xing2009statistical}. The lognormal pdf expression is
\begin{align}
f(r|\mu, \gamma) = \dfrac{1}{r\gamma\sqrt{2\pi}}\exp\left(-\dfrac{(\log r - \mu)^2}{2\gamma^2}\right)
\end{align}
where $\gamma$ is the scale, and $\mu$ is the location parameter.

Frery et al. \cite{frery1997model} have proposed a new class of statistical models, the $\mathcal{G}$ distributions, corresponding to a multiplicative speckle model. A special member of this class, which is the $\mathcal{K}$ distribution \cite{sun2018dependence, migliaccio2019sar, chitroub2002statistical} is one of the important statistical models able to model both amplitude and intensity SAR images. The $\mathcal{K}$ distribution pdf for amplitude modeling is expressed as
\begin{align}
 f(r|\alpha, \gamma) = \dfrac{2}{\gamma \Gamma(\alpha+1)} \left( \dfrac{r}{2\gamma} \right)^{\alpha+1} K_{\alpha}\left( \frac{r}{\gamma} \right)
\end{align}
where $K_{\alpha}(\cdot)$ is the modified Bessel function of the second kind of order $\alpha$ and, $\alpha$ and $\gamma$ refer to the shape, and scale parameters, respectively. Another special member of the $\mathcal{G}$ model class, which is the $\mathcal{G}_0$ distribution has shown a considerable performance in modeling extremely heterogeneous clutter such as in urban areas where $\mathcal{K}$ distribution fails. In particular, the $\mathcal{G}_0$ distribution can be obtained under the assumption that the back-scattered SAR amplitude follows a reciprocal of the square root of Gamma distribution \cite{cui2014comparative}. The amplitude pdf of the $\mathcal{G}_0$ distribution is given by
\begin{align}
    f(r|L, \gamma, \alpha) = \dfrac{2L^L\Gamma(L-\alpha)r^{2L-1}}{\gamma^{\alpha}\Gamma(L)\Gamma(-\alpha)(\gamma+Lr^2)^{L-\alpha}}
\end{align}
where $L$ is the number of looks, $\gamma$ and $\alpha$ refer to the scale and shape parameters, respectively. The $\mathcal{G}_0$ distribution has been successfully utilized for various modeling studies in the literature for single/multi looks, heterogeneous regions, and classification applications \cite{frery1997model,gao2010statistical, cui2011coarse, cui2014comparative}.

Another important empirical model for SAR amplitude/intensity modeling is the generalized gamma distribution (G$\Gamma$D) which was first proposed in \cite{stacy1962generalization}. Due to its highly versatile analytical form this statistical model has found various application areas in signal processing \cite{song2008globally}, economics \cite{kleiber2003statistical}, sea clutter modeling \cite{lampropoulos1999high, qin2012statistical,martin2014statistical}, and some other SAR scenes/applications \cite{li2010efficient, li2011empirical, krylov2010generalized, pappas2020river}. The pdf of G$\Gamma$D is given by
\begin{align}
    f(r|\nu, \sigma, \kappa) = \dfrac{\nu}{\sigma\Gamma(\kappa)}\left(\dfrac{r}{\sigma} \right)^{\kappa\nu-1}\exp\left[ -\left(\dfrac{r}{\sigma}\right)^{\nu}\right]
\end{align}
where the parameters $\nu$, $\sigma$ and $\kappa$ are all positive valued, and refer to the power, scale and shape parameters, respectively.

In a previous study, following the observation of non-Gaussian reflections in urban areas, Kuruoglu \& Zerubia \cite{kuruoglu2004modeling} proposed a generalized central limit theorem based  statistical model which extends the standard scattering models discussed above by considering the real and the imaginary parts of the complex back-scattered SAR signal to be jointly symmetric-$\alpha$-Stable random variables. This model, called the generalized Rayleigh distribution (will be denoted as Stable-Rayleigh, or shortly SR for the rest of the paper) for amplitude SAR images, the pdf of which is given as
\begin{align}
f(r|\alpha, \gamma) = r\int_{0}^{\infty} s\exp\left(-\gamma s^{\alpha}\right) \mathcal{J}_0(sr)ds
\end{align}
where $\mathcal{J}_0(\cdot)$ is the zeroth-order Bessel function of the first kind, and $\alpha$ and $\gamma$ refer to the shape, and scale parameters, respectively. SR distribution has been shown to be a good choice for urban SAR image modeling in \cite{kuruoglu2004modeling,karakucs2018generalized} and successfully applied to despeckling problem in \cite{achim2006sar}.

Moser et al., \cite{moser2006sar} have proposed another generalized theoretical statistical model for amplitude SAR modeling, which is similar to SR \cite{kuruoglu2004modeling}, by assuming the real and imaginary parts of the back-scattered signals to be independent zero-mean generalized Gaussian (GG) random variables, which leads to the generalized Gaussian Rayleigh (GGR) distribution, with pdf \cite{moser2006sar}
\begin{dmath}   \label{equ:GGR}
    f(r|\alpha, \gamma) = \dfrac{\alpha^2 r}{4\gamma^2\Gamma^2(\frac{1}{\alpha})} \int_0^{2\pi} \exp\left( -\dfrac{|r\cos\theta|^{\alpha} + |r\sin\theta|^{\alpha}}{\gamma^{\alpha}}\right) d\theta
\end{dmath}
where $\alpha$ and $\gamma$ refer to the shape, and scale parameters, respectively.

In a recent study \cite{karakusICASSP20}, we have proposed a novel statistical model, namely the \textit{Laplace-Rician} distribution for modeling amplitude SAR images of the sea surface. The Laplace-Rician model is based on the Rician distribution, whereby we assume that the real and imaginary parts of the back-scattered SAR signal are non-zero mean Laplace distributed. The Rician distribution is widely used in SAR imaging applications being particularly important in characterizing SAR scenes containing many strong back-scattered echoes. These include natural targets such as forest canopy, mountain tops, sea waves, as well as some man-made structures with dihedral or trihedral configurations such as buildings, or vessels \cite{nicolas2019new,devore2000atr,eltoft2005rician, goodman1975statistical,gao2010statistical, wu2013man, denis2010exact}. Combining the Rician idea with the non-Gaussian case via the Laplace distribution, \cite{karakusICASSP20} addresses both the non-Rayleigh and heavy-tailed characteristics of amplitude SAR images. The Laplace-Rician model, despite being limited to a Laplace distribution as the back-scattered SAR signal components' statistical model, showed superior performance for modeling amplitude SAR images of the sea surface when compared to state-of-the-art statistical models such as Weibull, lognormal, and $\mathcal{K}$ \cite{karakusICASSP20}.

In this paper, we propose a novel statistical model by extending the Laplace-Rician model to a much more general case, where the back-scattered SAR signal components are \textit{non-zero mean Generalized-Gaussian} distributed. We further introduce a Markov chain Monte Carlo (MCMC) based Bayesian parameter estimation method for the proposed statistical model. We demonstrate the modeling capability of the proposed model for amplitude/intensity SAR images from satellite platforms, including TerraSAR-X, ICEYE, COSMO/Sky-Med, Sentinel-1 and ALOS2, and for illuminated scenes of urban, agricultural, land cover, sea surface with and without ships, along with several mixed scenes. We evaluate the performance of the proposed model in comparison to state-of-the-art statistical models including Rician, Weibull, Lognormal, $\mathcal{G}_0$, G$\Gamma$D, SR \cite{kuruoglu2004modeling}, and GGR \cite{moser2006sar}.

The rest of the paper is organized as follows: we present the proposed statistical model in Section \ref{sec:proposed}. The Bayesian parameter estimation method is presented in Section \ref{sec:param}, whilst the experimental analysis is demonstrated in Section \ref{sec:results}. Section \ref{sec:conc} concludes the paper with remarks and future work.

\section{Generalized Gaussian Rician Model}\label{sec:proposed}
In this section, we introduce our main contribution, which is a novel statistical model derived as an extension of the generalized-Gaussian distribution into the Rician scattering idea. Our derivation starts by assuming that the illuminated SAR scene includes one (or more) dominating scatterers, such as vehicles, buildings, sea waves. Following this, the sixth assumption given above for the back-scattered SAR signal will not be valid anymore. Then, as in the Rician case, the real and imaginary components of the back-scattered complex SAR signal will be non-zero mean random variables.
We now recall the generalized Gaussian pdf
\begin{align}\label{equ:GG}
    f(x|\alpha, \gamma, \delta) = \dfrac{\alpha}{2\gamma\Gamma(\frac{1}{\alpha})}\exp\left(-\left|\dfrac{x - \delta}{\gamma}\right|^{\alpha} \right),
\end{align}
where $\delta$ is the location parameter. In order to have a  Rayleigh-type amplitude distribution, the location parameter $\delta$ is assumed to be zero along with the shape parameter $\alpha = 2$, where the assumption 6) is valid. However, for non-zero $\delta$, as long as the complex SAR signal components $x$ and $y$ are independent \cite{moser2006sar}, the joint pdf can be written as
\begin{align}
    f(x, y|\alpha, \gamma, \delta) &= f(x|\alpha, \gamma, \delta)\times f(y|\alpha, \gamma, \delta)\\
    &= \dfrac{\alpha}{2\gamma\Gamma(\frac{1}{\alpha})}\exp\left(-\left|\dfrac{x - \delta}{\gamma}\right|^{\alpha} \right) \dfrac{\alpha}{2\gamma\Gamma(\frac{1}{\alpha})}\exp\left(-\left|\dfrac{y - \delta}{\gamma}\right|^{\alpha} \right)\\
    &=\dfrac{\alpha^2}{4\gamma^2\Gamma^2(\frac{1}{\alpha})}\exp\left(-\dfrac{|x - \delta|^{\alpha} + |y - \delta|^{\alpha}}{\gamma^{\alpha}} \right).
\end{align}
\begin{figure*}[t]
\centering
\subfigure[]{
\includegraphics[width=.32\linewidth]{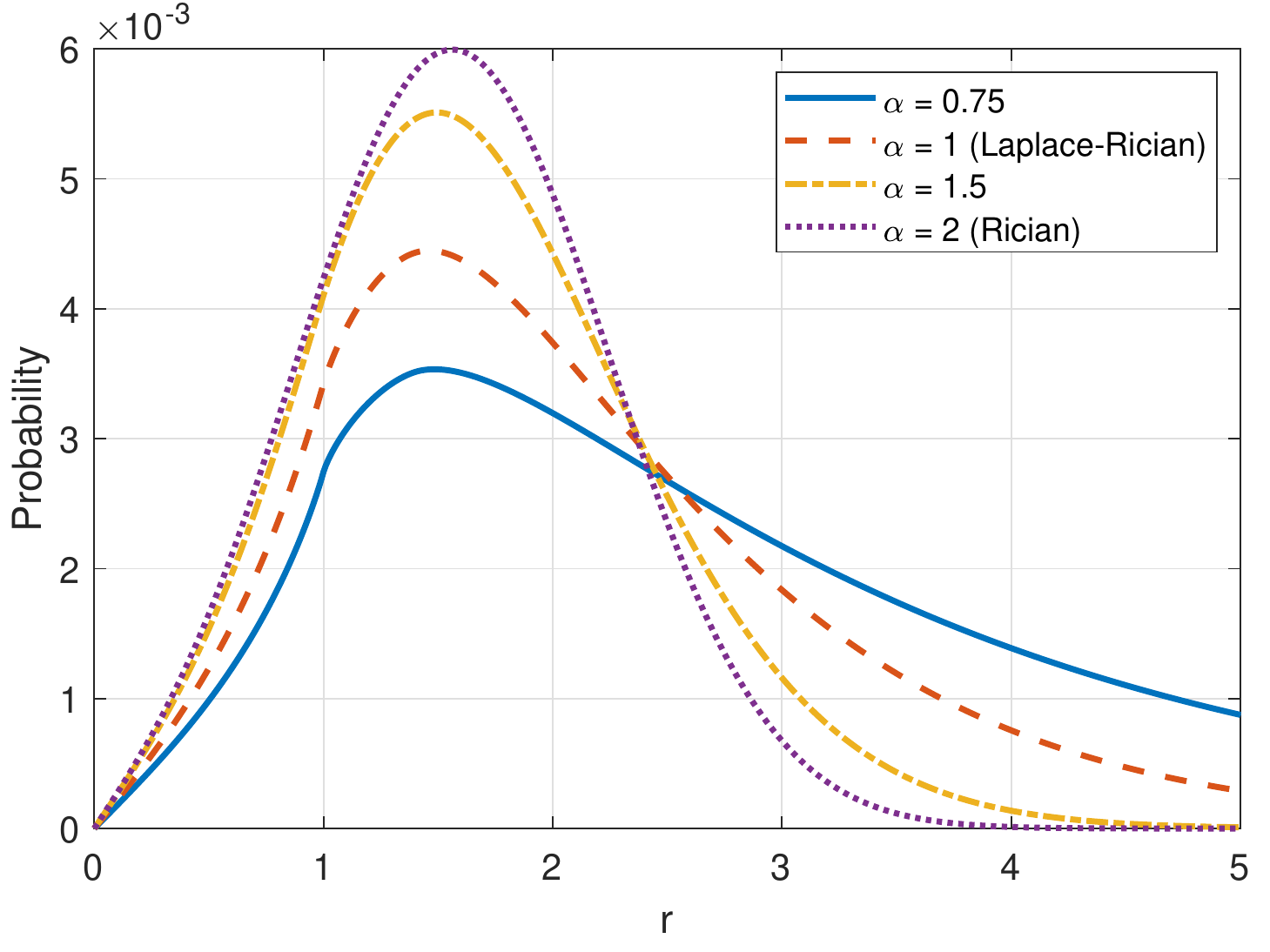}}
\centering
\subfigure[]{
\includegraphics[width=.32\linewidth]{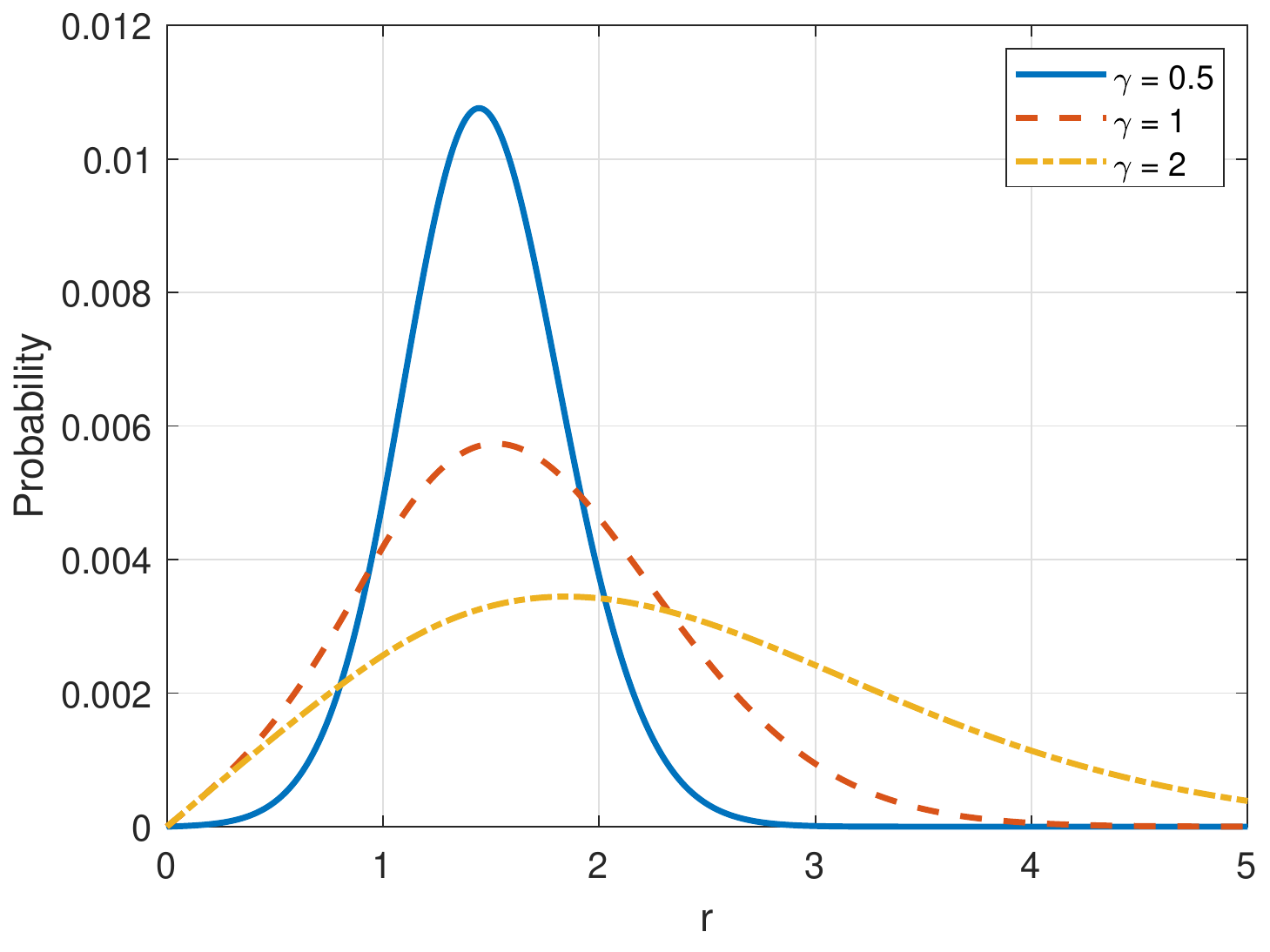}}
\centering
\subfigure[]{
\includegraphics[width=.32\linewidth]{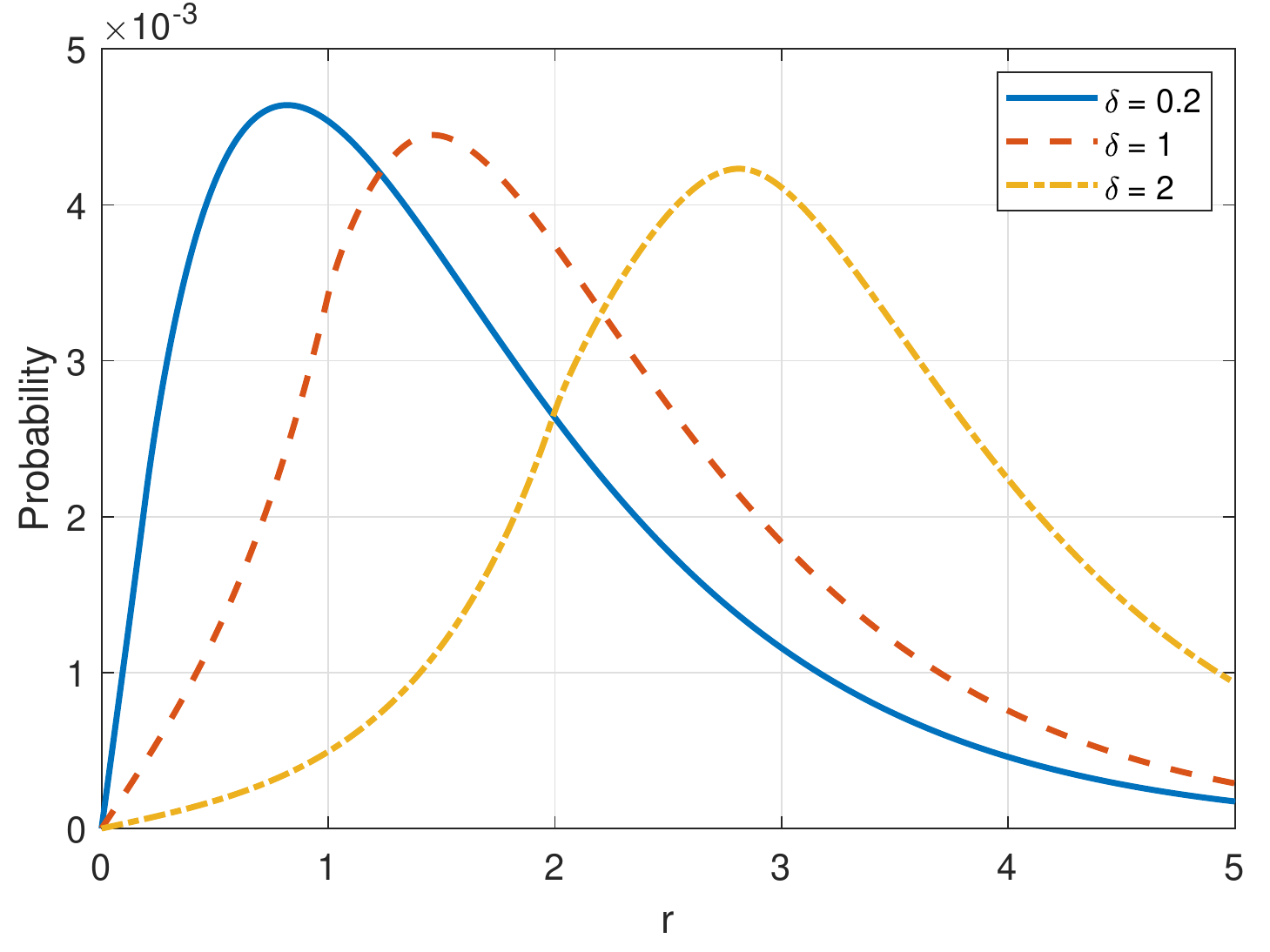}}
\caption{The proposed GG-Rician distribution pdfs for different model parameters of (a) the shape parameter $\alpha$ ($\gamma=1$ and $\delta=1$), (b) the scale parameter $\gamma$ ($\alpha=1.7$ and $\delta=1$), and (c) the location parameter $\delta$ ($\alpha=1$ and $\gamma=1$).} \vspace{-0.25cm}
\label{fig:paramFig}
\end{figure*}
Then, the amplitude distribution can be written by using the identity, $f(r, \theta) = r f(r\cos\theta, r\sin\theta)$, as
\begin{dmath}\label{equ:GGRicProd1}
    f(r, \theta|\alpha, \gamma, \delta) = r\dfrac{ \alpha^2}{4\gamma^2\Gamma^2(\frac{1}{\alpha})}
    \exp\left(-\dfrac{|r\cos\theta - \delta|^{\alpha} + |r\sin\theta - \delta|^{\alpha}}{\gamma^{\alpha}} \right).
\end{dmath}
where $\theta$ is uniformly distributed within $[0,2\pi]$. Hence, the corresponding marginal amplitude pdf can be obtained by averaging (\ref{equ:GGRicProd1}) over $\theta$ and boils down to:
\begin{dmath}\label{equ:GGRic}
    f(r|\alpha, \gamma, \delta) = \dfrac{\alpha^2 r}{4\gamma^2\Gamma^2(\frac{1}{\alpha})} \int_0^{2\pi} \exp\left(-\dfrac{|r\cos\theta - \delta|^{\alpha} + |r\sin\theta - \delta|^{\alpha}}{\gamma^{\alpha}} \right) d\theta,
\end{dmath}
The integral form pdf expression shown in (\ref{equ:GGRic}) refers to the proposed statistical model for the amplitude distribution of a complex back-scattered SAR signal, the components of which are non-zero mean generalized-Gaussian distributed. We now state the following theorems.

\begin{thm}\label{thm:Rici}
The pdf expression given in (\ref{equ:GGRic}) reduces  to the Rician distribution in (\ref{equ:Rice}) for $\alpha = 2$, where the real and imaginary components of back-scattered SAR signal become non-zero mean Gaussian random variables.
\end{thm}

\begin{proof}
Starting from (\ref{equ:GGRic}), and setting the shape parameter $\alpha = 2$, we have
\begin{align}
    f(r|\alpha = 2, \gamma, \delta) &= \dfrac{r}{\gamma^2\pi} \int_0^{2\pi} \exp\left(-\dfrac{(r\cos\theta - \delta)^2 + (r\sin\theta - \delta)^2}{\gamma^2} \right) d\theta,\\
    &= \dfrac{r}{\gamma^2\pi} \int_0^{2\pi} \exp\left(-\dfrac{r^2 - 2r\delta(\cos\theta + \sin\theta) + 2\delta^2}{\gamma^2} \right) d\theta,\\
    &= \dfrac{r}{\gamma^2\pi}  \exp\left(-\dfrac{r^2 + 2\delta^2}{\gamma^2} \right) \int_0^{2\pi} \exp\left(\dfrac{ r\delta(\cos\theta + \sin\theta)}{\gamma^2/2} \right) d\theta,
\end{align}
Using the identity
\begin{align}
(\cos\theta + \sin\theta) = \sqrt{2}\cos(\theta - \pi/4),
\end{align}
we have
\begin{dmath}
    f(r|\gamma, \delta) = \dfrac{r}{\gamma^2\pi}  \exp\left(-\dfrac{r^2 + 2\delta^2}{\gamma^2} \right) \int_0^{2\pi} \exp\left(\dfrac{ r\sqrt{2}\delta\cos(\theta - \pi/4)}{\gamma^2/2} \right) d\theta.
\end{dmath}
Recall that the zeroth order modified Bessel function of the first kind is expressed as
\begin{align}
    \mathcal{I}_0(z) = \dfrac{1}{2\pi} \int_0^{2\pi} \exp\left( z\cos\theta\right) d\theta.
\end{align}
After basic manipulations we can express the pdf as
\begin{dmath}
    f(r|\gamma, \delta) = \dfrac{2\pi r}{\gamma^2\pi}  \exp\left(-\dfrac{r^2 + 2\delta^2}{\gamma^2} \right) \underbrace{\dfrac{1}{2\pi} \int_0^{2\pi} \exp\left(\dfrac{ r\sqrt{2}\delta}{\gamma^2/2}\cos(\theta - \pi/4) \right) d\theta}_{\mathcal{I}_0\left(\frac{r\sqrt{2}\delta}{\gamma^2/2} \right)}.
\end{dmath}
It is straightforward to see that the integral on the right hand side is effectively a zeroth order modified Bessel function of the first kind, and hence we have
\begin{align}
    f(r|\gamma, \delta) = \dfrac{r}{\gamma^2/2}  \exp\left(-\dfrac{r^2 + 2\delta^2}{\gamma^2} \right) \mathcal{I}_0\left(\dfrac{r\sqrt{2}\delta}{\gamma^2/2} \right),
\end{align}
which is the Rician distribution for $\gamma^2/2 = \sigma^2$ and $\sqrt{2}\delta = \Delta$ in (\ref{equ:Rice}).
\end{proof}

\begin{thm}\label{thm:LR}
The pdf expression given in (\ref{equ:GGRic}) is the Laplace-Rician distribution \cite{karakusICASSP20} for $\alpha = 1$,
\begin{align}\label{equ:LR}
     f(r|\gamma, \delta) = \dfrac{r}{4\gamma^2} \int_0^{2\pi} \exp\left(-\dfrac{|r\cos\theta - \delta| + |r\sin\theta - \delta|}{\gamma} \right) d\theta.
\end{align}
where the real and imaginary components of back-scattered SAR signal are distributed according to a non-zero mean Laplace distribution.
\end{thm}

\begin{proof}
For the proof of Theorem \ref{thm:LR}, we refer the reader to \cite{karakusICASSP20}.
\end{proof}

\begin{rem}
We refer to the proposed pdf expression in (\ref{equ:GGRic}) as the Generalized Gaussian-Rician distribution (GG-Rician), since it extends the Rician amplitude model to a heavy-tailed form via the generalized Gaussian distributed complex SAR signal components.
\end{rem}

To give a feel for the characteristics of this class of distributions, they are plotted for various values of parameters in Figure \ref{fig:paramFig}.

\subsection{Extension to Intensity SAR Images} \label{sec:intensity}
The derived proposed statistical model (\ref{equ:GGRic}) characterizes the amplitude SAR signal. However, for some applications, intensity SAR images have been used instead of amplitude images. In this section, we derive the intensity pdf expression for the proposed GG-Rician statistical model.

For an intensity SAR image, the pdf expression can be calculated from the pdf of the amplitude image using the pdf transformation formula:
\begin{align}\label{equ:Amp2Int}
    f_I(\nu) = \dfrac{1}{2\sqrt{\nu}}f_A(\sqrt{\nu})
\end{align}
where $\nu = r^2$ refers to the intensity with the pdf of $f_I(\cdot)$. Then, using the identity in \eqref{equ:Amp2Int}, the GG-Rician intensity pdf can be written as
\begin{dmath}\label{equ:GGRicInt}
    f_I(\nu|\alpha, \gamma, \delta) = \dfrac{\alpha^2 }{8\gamma^2\Gamma^2(\frac{1}{\alpha})} \int_0^{2\pi} \exp\left(-\dfrac{|\sqrt{\nu}\cos\theta - \delta|^{\alpha} + |\sqrt{\nu}\sin\theta - \delta|^{\alpha}}{\gamma^{\alpha}} \right) d\theta.
\end{dmath}
\begin{thm}\label{thm:NakRici}
The intensity pdf expression given in (\ref{equ:GGRicInt}) simplifies to the Nakagami-Rice distribution \cite{dana1986impact,tison2004new} for $\alpha = 2$
\begin{align}\label{equ:NakRice}
  f_I(\nu|R, \Delta) = \dfrac{1}{R}\exp\left(-\dfrac{\nu + \Delta^2}{R}\right) \mathcal{I}_0\left(\dfrac{\sqrt{\nu\Delta^2}}{R/2} \right)
\end{align}
where $R$ is the scale and $\Delta$ is the location parameter.
\end{thm}

\begin{proof}
We start by recalling (\ref{equ:GGRicInt}), and setting the shape parameter $\alpha = 2$. Then, we have
\begin{align}
    f_I(\nu|\alpha=2, \gamma, \delta) &= \dfrac{1 }{2\gamma^2\pi} \int_0^{2\pi} \exp\left(-\dfrac{|\sqrt{\nu}\cos\theta - \delta|^2 + |\sqrt{\nu}\sin\theta - \delta|^{\alpha}}{\gamma^2} \right) d\theta.\\
    &= \dfrac{1}{2\gamma^2\pi} \int_0^{2\pi} \exp\left(-\dfrac{\nu - 2\sqrt{\nu}\delta(\cos\theta + \sin\theta) + 2\delta^2}{\gamma^2} \right) d\theta.\\
    \label{equ:proofNakRici}&= \dfrac{1}{2\gamma^2\pi} \exp\left(-\dfrac{\nu + 2\delta^2}{\gamma^2} \right) \int_0^{2\pi} \exp\left(\dfrac{\sqrt{\nu}\delta(\cos\theta + \sin\theta)}{\gamma^2/2} \right) d\theta.
\end{align}

In a way akin to the proof of Theorem \ref{thm:Rici}, we can easily write the expression in (\ref{equ:proofNakRici}) as
\begin{dmath}
    f_I(\nu|\gamma, \delta) = \dfrac{2\pi}{2\gamma^2\pi}  \exp\left(-\dfrac{\nu + 2\delta^2}{\gamma^2} \right) \underbrace{\dfrac{1}{2\pi} \int_0^{2\pi} \exp\left(\dfrac{ \sqrt{\nu2\delta^2}}{\gamma^2/2}\cos(\theta - \pi/4) \right) d\theta}_{\mathcal{I}_0\left(\frac{\sqrt{\nu2\delta^2}}{\gamma^2/2} \right)}.
\end{dmath}

It is straightforward to see that the integral on the right hand side is a zeroth order modified Bessel function of the first kind. Then, we have
\begin{align}
    f_I(\nu|\gamma, \delta) = \dfrac{1}{\gamma^2}  \exp\left(-\dfrac{\nu + 2\delta^2}{\gamma^2} \right) \mathcal{I}_0\left(\dfrac{\sqrt{\nu2\delta^2}}{\gamma^2/2} \right),
\end{align}
which is the Nakagami-Rician distribution for $\gamma^2 = R$ and $2\delta^2 = \Delta^2$, and completes the proof.
\end{proof}

\begin{rem}
It is straightforward to observe that for $\alpha = 2$ and $\delta = 0$, with a derivation akin to Theorem \ref{thm:NakRici}, the corresponding intensity pdf will boil down to the \textbf{exponential distribution}, and similarly for $L$-look case to the \textbf{Gamma distribution}.
\end{rem}

As mentioned up to this point, the GG-Rician statistical model is a general statistical model, which covers various important amplitude and intensity statistical models as special members. For completeness, we provide GG-Rician pdf expressions, and some special cases in Table \ref{tab:specialCases}.

\begin{table}[htbp]
  \centering
  \caption{GG-Rician family special members for amplitude and intensity.}
    \begin{tabularx}{0.96\linewidth}{@{} RL @{}}
    \toprule
    Distribution & Expression  \\
    \toprule
    GG-Rician (Amplitude) & $f(r|\alpha, \gamma, \delta)$ in (\ref{equ:GGRic})\\
    Rayleigh (\ref{equ:Rayl}) & $f(r|2, \gamma, 0)$ \\
    Rician (\ref{equ:Rice}) & $f(r|2, \gamma, \delta)$\\
    GGR (\ref{equ:GGR}) \cite{moser2006sar} & $f(r|\alpha, \gamma, 0)$ \\
    Nakagami & $L$-look average of $f(r|2, \gamma, 0)$\\
    Laplace-Rician (\ref{equ:LR}) \cite{karakusICASSP20} & $f(r|1, \gamma, \delta)$\\
    \hline
    GG-Rician (Intensity) & $f_I(\nu|\alpha, \gamma, \delta)$ in (\ref{equ:GGRicInt})\\
    Exponential & $f_I(\nu|2, \gamma, 0)$\\
    Nakagami-Rice (\ref{equ:NakRice}) & $f_I(\nu|2, \gamma, \delta)$\\
     Gamma & $L$-look average of $f_I(\nu|2, \gamma, 0)$\\
    \toprule
    \end{tabularx}%
  \label{tab:specialCases}%
\end{table}%

\section{Bayesian Parameter Estimation Method}\label{sec:param}
Since the pdf expression in (\ref{equ:GGRic}) is not in a compact analytical form and it does not seem to be possible to invert it to obtain parameter values, we employ a Bayesian sampling methodology in order to estimate model parameters. In this section, a Markov chain Monte Carlo (MCMC) based method is developed for estimating GG-Rician distribution parameters, namely the shape parameter $\alpha$, the scale parameter $\gamma$, and the location parameter $\delta$. In particular, the method uses the Metropolis-Hastings (MH) algorithm, and in each iteration, it applies one of the three different moves:
\begin{enumerate}
    \item $\mathcal{M}_1$ which updates $\delta$ for fixed $\alpha$ and $\gamma$,
    \item $\mathcal{M}_2$ which updates $\gamma$ for fixed $\alpha$ and $\delta$,
    \item $\mathcal{M}_3$ which updates $\alpha$ for fixed $\gamma$ and $\delta$.
\end{enumerate}
The proposed parameter estimation procedure is given in Algorithm \ref{alg:Param}. Given the observed data $y$, the hierarchical model is expressed by Bayes' theorem as
\begin{align}
    p(\alpha, \delta, \gamma|y) \propto p(y|\alpha, \delta, \gamma) p(\alpha) p(\gamma)p(\delta)
\end{align}
where $p(\alpha, \delta, \gamma|y)$ is the joint posterior distribution, or the MH target distribution, $p(y|\alpha, \delta, \gamma)$ refers to the likelihood distribution, and $p(\alpha)$, $p(\gamma)$ and $p(\delta)$ are the priors.

Due to lack of knowledge on conjugate priors, we choose non-informative priors for the shape, location and scale (Jeffrey's) parameters. In particular, we assume that the location and shape parameters $\alpha$ and $\delta$ are uniformly distributed and that the prior for the scale parameter $\gamma$ is $p(\gamma)=1/\gamma$, which leads to $p(\alpha, \delta, \gamma) \sim 1/\gamma$. The likelihood $p(y|\alpha, \delta, \gamma)$ is the GG-Rician distribution in (\ref{equ:GGRic}) with parameters $\alpha$, $\gamma$ and $\delta$.

Depending on the selected move in iteration $i$, one 
of the proposal distributions given below is used to sample candidate parameters $\delta^*$, $\gamma^*$ or $\alpha^*$
\begin{align}
    \mathcal{M}_1: \quad \delta^* &\propto q\left(\delta^*|\delta^{(i)}\right) =  \mathcal{U}\left(\delta^{(i)} -\epsilon, \delta^{(i)} + \epsilon \right),\\
    \mathcal{M}_2: \quad \gamma^* &\propto q\left(\gamma^*|\gamma^{(i)}\right) = \mathcal{N}\left(\gamma^{(i)}, \xi^2 \right),\\
    \mathcal{M}_3: \quad \alpha^* &\propto q\left(\alpha^*|\alpha^{(i)}\right) =  \mathcal{U}\left(\alpha^{(i)} - \eta, \alpha^{(i)} + \eta \right),
\end{align}
where $\mathcal{U}(\cdot)$ is the uniform, and $\mathcal{N}(\cdot)$ is the Gaussian distributions, both of which are defined in the interval $[0,\infty]$ since $\alpha$, $\delta$ and $\gamma$ are positive parameters. $\eta$, $\epsilon$ and $\xi$ are hyper-parameters of the proposal distributions. Please note that this choice of proposals is not unique and can be replaced with other distributions for faster convergence  in estimating model parameters. 

\begin{algorithm}
\caption{MCMC Parameter Estimation for GG-Rician Distribution}\label{alg:Param}
\begin{algorithmic}[1]\small
\State \textbf{Inputs:} $\text{Given data }y\text{.}$
\State \textbf{Output:} $\text{Joint Posterior }f(\alpha, \delta, \gamma|y)$
\State \textbf{Initialize:} $\alpha^{(1)}\text{, }\delta^{(1)}\text{, }\gamma^{(1)}\text{, }\eta\text{, }\nu\text{ and }\xi.$
  \For {$i=1:N_{iter}$}
    \State \text{Choose Move, }$m^{(i)}$ equally likely between $\mathcal{M}_1$, $\mathcal{M}_2$ or $\mathcal{M}_3$
    \If {$m^{(i)} \rightarrow \mathcal{M}_1$}
    \State Sample $\delta^* \sim q\left(\delta^*|\delta^{(i)}\right)$
    \State Set $\alpha^* = \alpha^{(i)}$ and $\gamma^* = \gamma^{(i)}$ and $A = A_{\mathcal{M}_1}$.
    \Else {\text{\textbf{if} $m^{(i)} \rightarrow \mathcal{M}_2$ \textbf{then}}}
    \State Sample $\gamma^* \sim q\left(\gamma^*|\gamma^{(i)}\right)$
    \State Set $\alpha^* = \alpha^{(i)}$ and $\delta^* = \delta^{(i)}$ and $A = A_{\mathcal{M}_2}$.\\
    \hspace{0.4cm} \text{\textbf{elseif} $m^{(i)} \rightarrow \mathcal{M}_3$ \textbf{then}}
    \State Sample $\alpha^* \sim q\left(\alpha^*|\alpha^{(i)}\right)$
     \State Set $\delta^* = \delta^{(i)}$ and $\gamma^* = \gamma^{(i)}$ and $A = A_{\mathcal{M}_3}$.
    \EndIf
    \State Sample random variable $u \sim \mathcal{U}(0, 1)$
    \If {$u \leq A$}
        \State $\alpha^{(i+1)} = \alpha^*$ and $\delta^{(i+1)} = \delta^*$ and $\gamma^{(i+1)} = \gamma^*$
    \Else
        \State $\alpha^{(i+1)} = \alpha^{(i)}$ and $\delta^{(i+1)} = \delta^{(i)}$ and $\gamma^{(i+1)} = \gamma^{(i)}$
    \EndIf
  \EndFor
\end{algorithmic}
\end{algorithm}

Consequently, the acceptance probability expressions for each move can be expressed as
\begin{align}
    A_{\mathcal{M}_1} &= \min\left(1, \dfrac{p(y|\alpha^*, \delta^*, \gamma^*)q\left(\delta^{(i)}|\delta^*\right)}{p(y|\alpha^{(i)}, \delta^{(i)}, \gamma^{(i)})q\left(\delta^*|\delta^{(i)}\right)}\right),\\
    A_{\mathcal{M}_2} &= \min\left(1, \dfrac{p(y|\alpha^*, \delta^*, \gamma^*)p(\gamma^*)q\left(\gamma^{(i)}|\gamma^*\right)}{p(y|\alpha^{(i)}, \delta^{(i)},\gamma^{(i)})p(\gamma^{(i)})q\left(\gamma^*|\gamma^{(i)}\right)}\right),\\
    A_{\mathcal{M}_3} &= \min\left(1, \dfrac{p(y|\alpha^*, \delta^*, \gamma^*)q\left(\alpha^{(i)}|\alpha^*\right)}{p(y|\alpha^{(i)}, \delta^{(i)},\gamma^{(i)})q\left(\alpha^*|\alpha^{(i)}\right)}\right).
\end{align}

\section{Experimental Analysis}\label{sec:results}

The proposed method was tested in four different perspectives using both simulated and real data.
\begin{enumerate}
    \item In the first simulation case, we used synthetically generated GG-Rician data for various parameters and tested the parameter estimation performance of the proposed MCMC method in terms of the normalized mean-square-error (NMSE), and statistical significance measures including KL, KS and $p$-value.
    \item Second, we subsequently conducted experiments to determine the best fitting amplitude distribution for given real SAR images of various scenes.
    \item Third, we evaluated the estimated GG-Rician model parameters on a large SAR scene, which was then decomposed into several image patches of 250$\times$250. For each patch, estimated model parameters are combined to create a parameter map, which potentially gives ideas on how the different parts of a large image affects the estimated parameters of the GG-Rician model.
    \item For the fourth and the last set of simulations, we performed a modeling study on intensity SAR images using the GG-Rician distribution.
\end{enumerate}

We used the statistical significance measures of \textit{Kullback-Leibler} (KL) divergence, \textit{Kolmogorov-Smirnov} (KS) score and $p$-value in order to assess the performance of fitting distributions. Smaller KL and KS values (higher $p$-values) indicate a better modeling performance. KL divergence is used to test the performance by considering the estimated pdfs and data histograms, whereas KS score is calculated by evaluating the estimated and the empirical cumulative distribution functions (CDFs). In addition, we also used error metrics of root-mean square error (RMSE), mean absolute error (MAE) and Bhattacharyya distance (BD) in order to test the performance. Lastly, in order to provide a performance measure that also considers the effect of the number of model parameters, we used the corrected Akaike information criterion (AICc) \cite{cavanaugh2019akaike}. 

The number of iterations, $N_{iter}$ in the MCMC parameter estimation method was set to 1000 and the first 500 iterations were discarded as burn-in period. Initial values for the parameters $\alpha^{(1)}$, $\delta^{(1)}$ and $\gamma^{(1)}$ were set to 2, 10 and 10, respectively. For proposal hyper-parameters, we chose $\epsilon = 2.5$, $\xi = 3$, and $\eta = 0.5$ after a trial-error procedure. All three model moves $\mathcal{M}_1$, $\mathcal{M}_2$ and $\mathcal{M}_3$ are equiprobable whilst satisfying $p(\mathcal{M}_1) + p(\mathcal{M}_2) + p(\mathcal{M}_3) = 1$. For all state-of-the-art statistical models, we utilized an MCMC based maximum likelihood (ML) methodology to estimate the model parameters. The number of histogram bins for analysis was calculated for each image by using Sturge's method \cite{sturges1926choice}. The Matlab package associated with the GG-Rician pdf generation and MCMC based parameter estimation have been made publicly available via the University of Bristol Research Data Repository\footnote{\url{https://doi.org/10.5523/bris.3nqhdd4qvorwx28hjh8bh6g3r8}} and gitHub\footnote{\url{https://github.com/oktaykarakus/GG-Rician-SAR-Image-Modelling}} for reproducibility.

\subsection{Synthetically Generated Data}
In the first set of simulations, eight synthetically generated GG-Rician data sets were obtained and the proposed parameter estimation method was used to estimate $\alpha$, $\delta$ and $\gamma$ for each data set. The corresponding data sets were generated for $(\alpha, \delta, \gamma)$ which are given in Table \ref{tab:compSyn2}. Each data set has 1500 samples, and the results are presented in
Table \ref{tab:compSyn2} and Figure \ref{fig:paramEstim}.

By examining estimated values in Table \ref{tab:compSyn2}, we can see  that all the model parameters $\alpha, \gamma$ and $\delta$ are estimated with relatively small NMSE values and are very close to their true values. For all eight example data sets, KL and KS values are low (with relatively high $p$-values) which certifies that the model parameters are successfully estimated.

\begin{table*}[htbp]
  \centering
  \caption{Modeling and statistical significance results for synthetically generated GG-Rician data sets.}\label{tab:compSyn2}
\begin{tabularx}{0.99\linewidth}{@{} CCCCCCCC @{}}
\hline
        \multicolumn{1}{c}{($\alpha$, $\delta$, $\gamma$)} &\multicolumn{1}{c}{\textbf{Est. Shape$^*$ ($\hat{\alpha}$)}} & \multicolumn{1}{c}{\textbf{Est. Location$^*$ ($\hat{\delta}$)}} & \multicolumn{1}{c}{\textbf{Est. Scale$^*$ ($\hat{\gamma}$)}}  &   \multicolumn{1}{c}{\textbf{NMSE}} &    \multicolumn{1}{c}{\textbf{KL Div.}} & \multicolumn{1}{c}{\textbf{KS Score}} & \multicolumn{1}{c}{\textbf{$p$-value}}\\
          \hline
    \multicolumn{1}{l}{(1.7. 2.9, 2.3)} & 1.58$\pm$0.111 & 2.74$\pm$0.040 & 2.26$\pm$0.129 & 0.0034 &0.0043 & 0.0246 & 0.9992 \\
    \multicolumn{1}{l}{(1.45, 1, 5)} & 1.42$\pm$0.081 & 1.09$\pm$0.391 & 4.97$\pm$0.342 & 0.0006 &0.0049 & 0.0186 & 1.0000 \\
    \multicolumn{1}{l}{(1.1, 10, 2)} & 1.04$\pm$0.043 & 10.09$\pm$0.051 & 1.86$\pm$0.118 & 0.0010 &0.0069 & 0.0227 & 0.9998 \\
    \multicolumn{1}{l}{(0.7, 5, 1.5)} & 0.79$\pm$0.022 & 4.86$\pm$0.091 & 1.98$\pm$0.133 & 0.0163 & 0.0137 & 0.0183 & 1.0000 \\
    \multicolumn{1}{l}{(1.2, 47, 32)} & 1.31$\pm$0.078 & 46.74$\pm$0.787 & 34.92$\pm$2.207 & 0.0504& 0.0042 & 0.0262 & 0.9979 \\
     \multicolumn{1}{l}{(0.5, 2, 0.5)} & 0.59$\pm$0.032 & 2.00$\pm$0.086 & 0.94$\pm$0.186 & 0.0317& 0.0153 & 0.0335 & 0.9632 \\
      \multicolumn{1}{l}{(1, 1.7, 1.3)} & 1.04$\pm$0.038 & 1.71$\pm$0.035 & 1.39$\pm$0.081 & 0.0014& 0.0018 & 0.0147 & 1.0000 \\
       \multicolumn{1}{l}{(2, 2, 4)} & 1.85$\pm$0.151 & 2.05$\pm$0.163 & 3.73$\pm$0.242 & 0.0067&0.0032 & 0.0136 & 1.0000 \\
          \hline
    \end{tabularx}\\ \vspace{0.1cm} \raggedright
    $^*$ Estimated values are given in a format of: (\textit{posterior mean})$\pm$(\textit{posterior standard deviation}).
\end{table*}%
\begin{figure}[ht]
\centering
\subfigure[]{
\includegraphics[width=.32\linewidth]{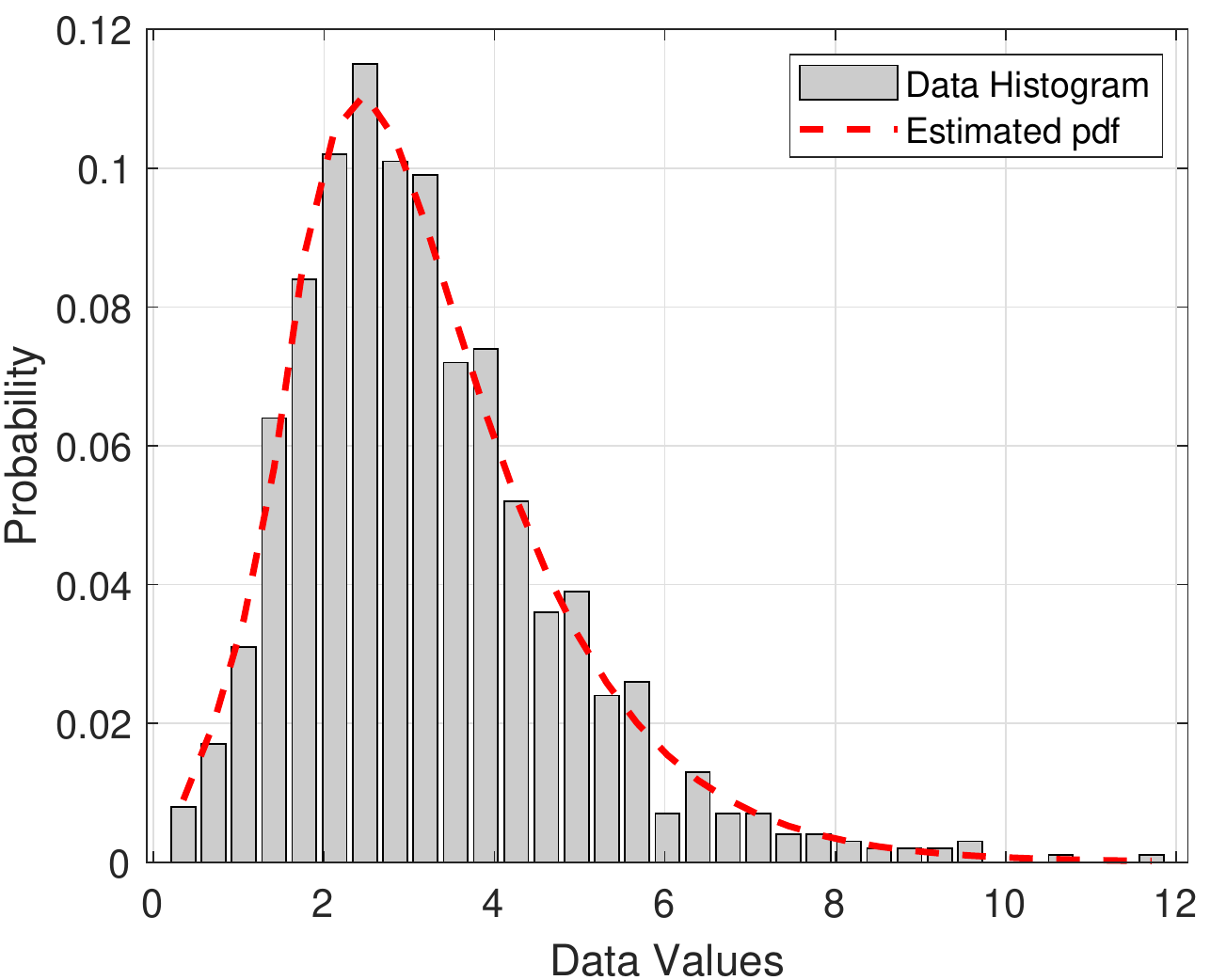}}
\centering
\subfigure[]{
\includegraphics[width=.32\linewidth]{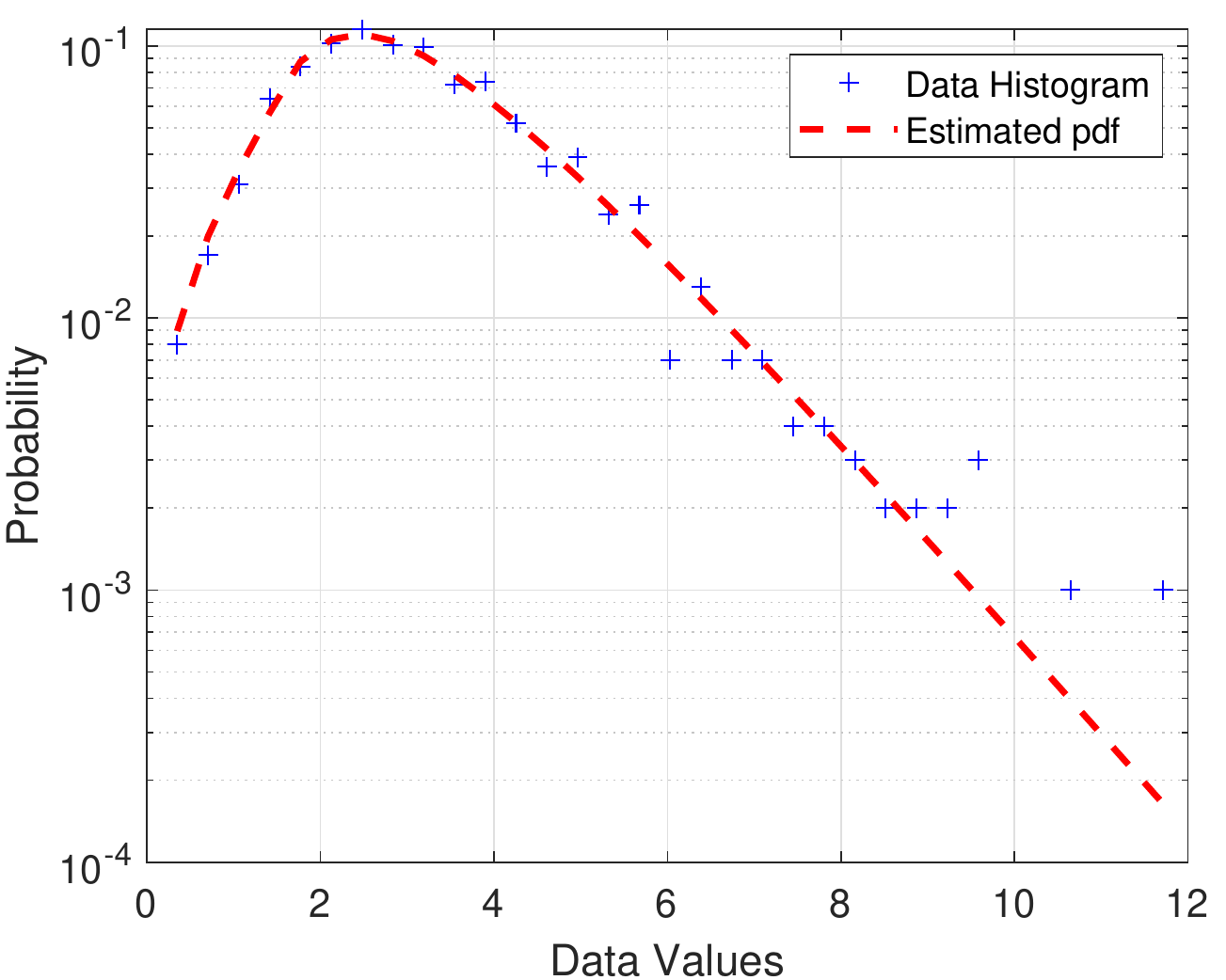}}\\
\centering
\subfigure[]{
\includegraphics[width=.32\linewidth]{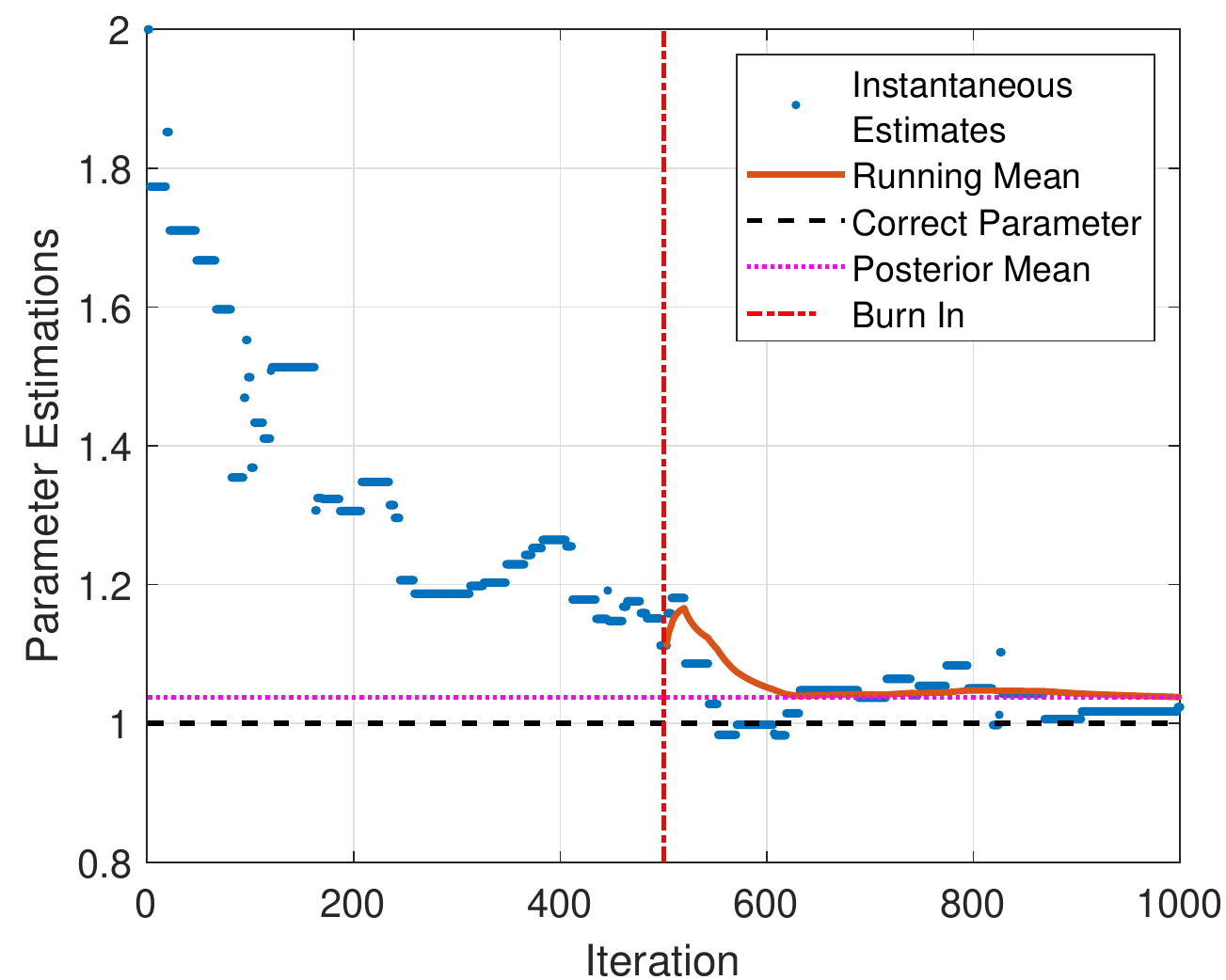}}
\subfigure[]{
\includegraphics[width=.32\linewidth]{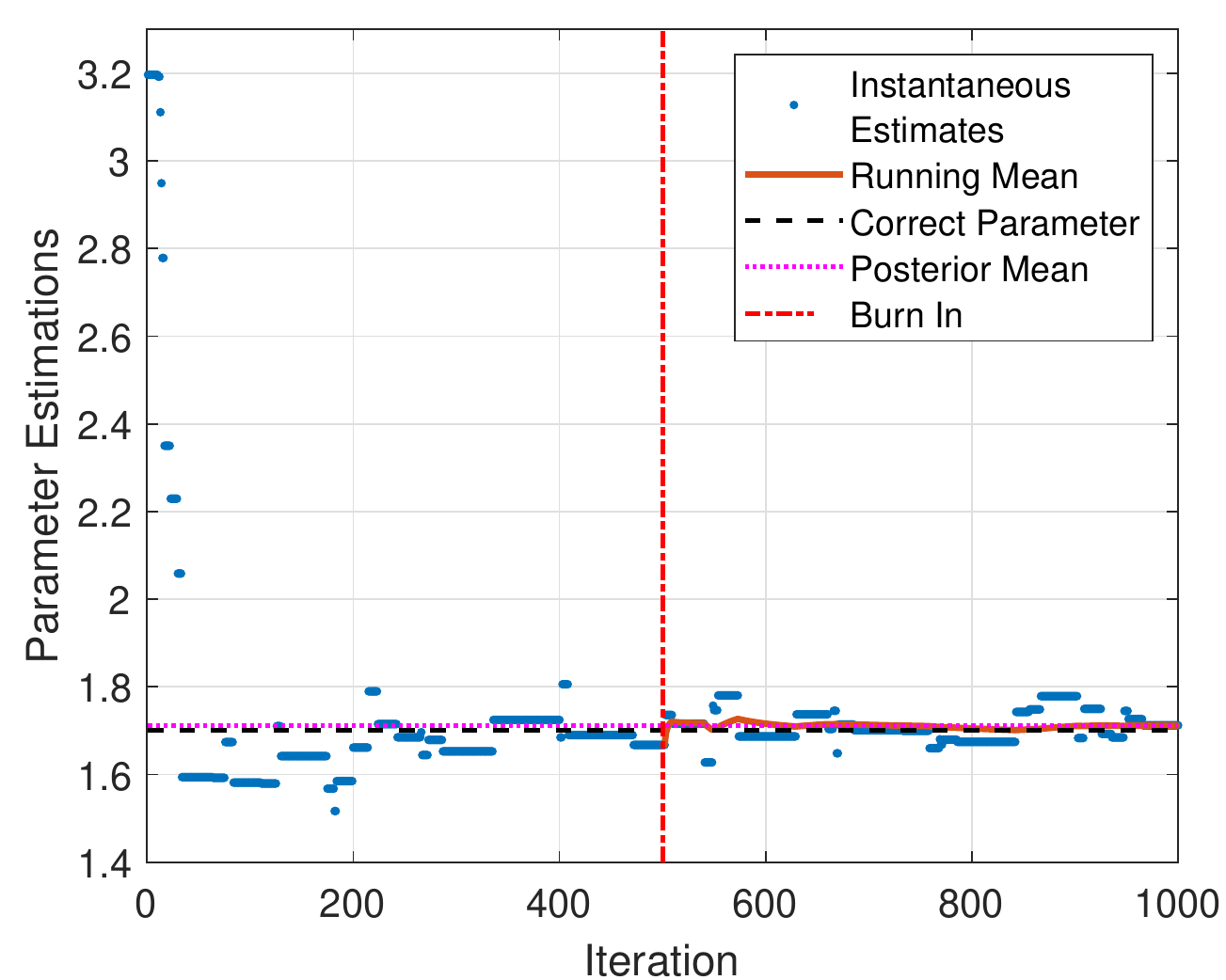}}
\subfigure[]{
\includegraphics[width=.32\linewidth]{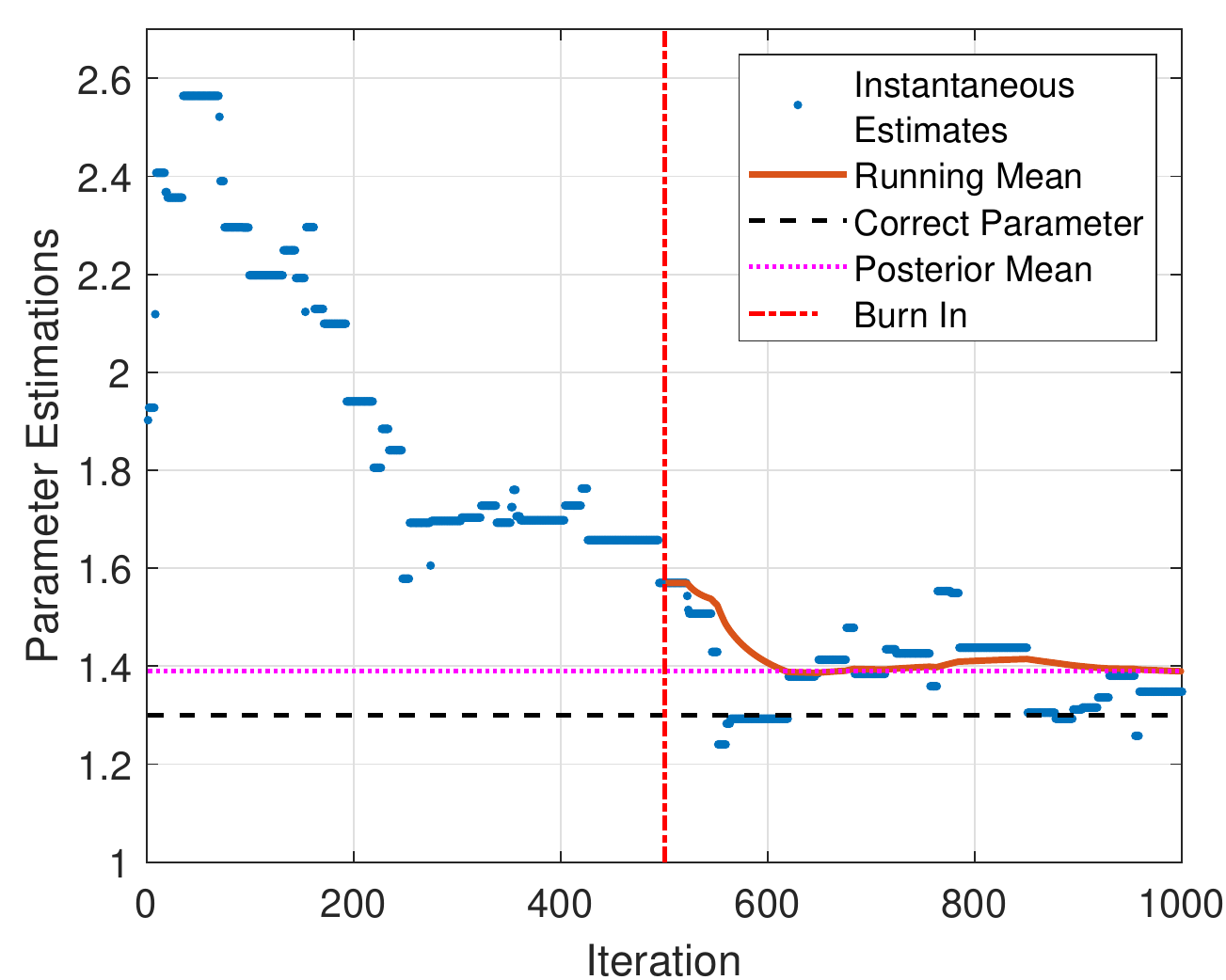}}
\caption{Modeling and parameter estimation results for synthetically generated data for the GG-Rician model of (1, 1.7, 1.3). (a) pdf Fitting. (b) Log-pdf fitting. Instantaneous estimates are presented for (c) The shape parameter $\alpha$, (d) The location parameter $\delta$, and (e) The scale parameter $\gamma$.}
\label{fig:paramEstim}\vspace{-0.5cm}
\end{figure}

Figure \ref{fig:paramEstim} shows modeling and parameter estimation results for the synthetically generated data from GG-Rician model of (1, 1.7, 1.3). When examining sub-figures in Figure \ref{fig:paramEstim}-(a) and (b), we can state that the fitted distribution follows the generated data histogram well for both numerical and logarithmic scales. Sub-figures in Figure \ref{fig:paramEstim}-(c)-(e) show instantaneous estimates for the parameters $\alpha$, $\delta$ and $\gamma$, respectively. The vertical line in all sub-figures represent the burn-in period, whilst the black and pink lines refer to the true and posterior mean values of the model parameters. When examining sub-figures in Figure \ref{fig:paramEstim}-(c)-(e), we can state that the parameter estimation method converges to the true parameter values within $N_{iter}$ iterations. Furthermore, 500 iterations of burn-in period looks like a good choice since all instantaneous estimates are scattered near the true model parameters after the burn-in period.

\subsection{Real Amplitude SAR Data}
In the second set of simulations, the proposed method was tested for 43 different SAR images coming from various platforms with frequency bands of \textit{X} (TerraSAR-X, COSMO-SkyMed and ICEYE), \textit{L} (ALOS-2) and \textit{C} (Sentinel-1). Each SAR image corresponds to one type of scene, i.e. \textit{urban, agricultural, mountain, land cover, mixed} and \textit{sea surface with and without ships and their wakes}. Since we have three sources for X band SAR imagery as mentioned above, we have more X band example images in this study, the exact distribution of images in terms of scenes and frequency bands is given in Table \ref{tab:dataset}. The performance of the GG-Rician model was compared to state-of-the-art models including the Rician, Weibull, Lognormal, Stable-Rayleigh (SR), GG-Rayleigh (GGR), $\mathcal{G}_0$, and G$\Gamma$D distributions. It is worth noting that other common models such as the Rayleigh, Gamma, Nakagami and $\mathcal{K}$ distributions have been left aside from our simulations, since these are all special members of already included statistical models.

\begin{table}[htbp]
  \centering
  \caption{Distribution of images in terms of scene and frequency bands.}
    \begin{tabularx}{0.76\linewidth}{@{} CCCCC @{}}
    \toprule
    Scene & \multicolumn{3}{c}{Frequency band} & \multicolumn{1}{c}{Total} \\
          & \multicolumn{1}{c}{X} & \multicolumn{1}{c}{C} & \multicolumn{1}{c}{L} &  \\
          \toprule
    Urban & 4     & 3     & 0     & 7 \\
    Agricultural & 3     & 3     & 2     & 8 \\
    Land  & 4     & 1     & 1     & 6 \\
    Mountain & 1     & 2     & 3     & 6 \\
    wSea & 3     & 1     & 1     & 5 \\
    woSea & 2     & 2     & 1     & 5 \\
    Mixed & 2     & 3     & 1     & 6 \\
    \hline
    Total & 19 & 15 & 9 & 43\\
    \toprule
    \end{tabularx}%
  \label{tab:dataset}%
\end{table}%

\begin{figure}\centering
    \includegraphics[keepaspectratio=true, width=0.5\linewidth]{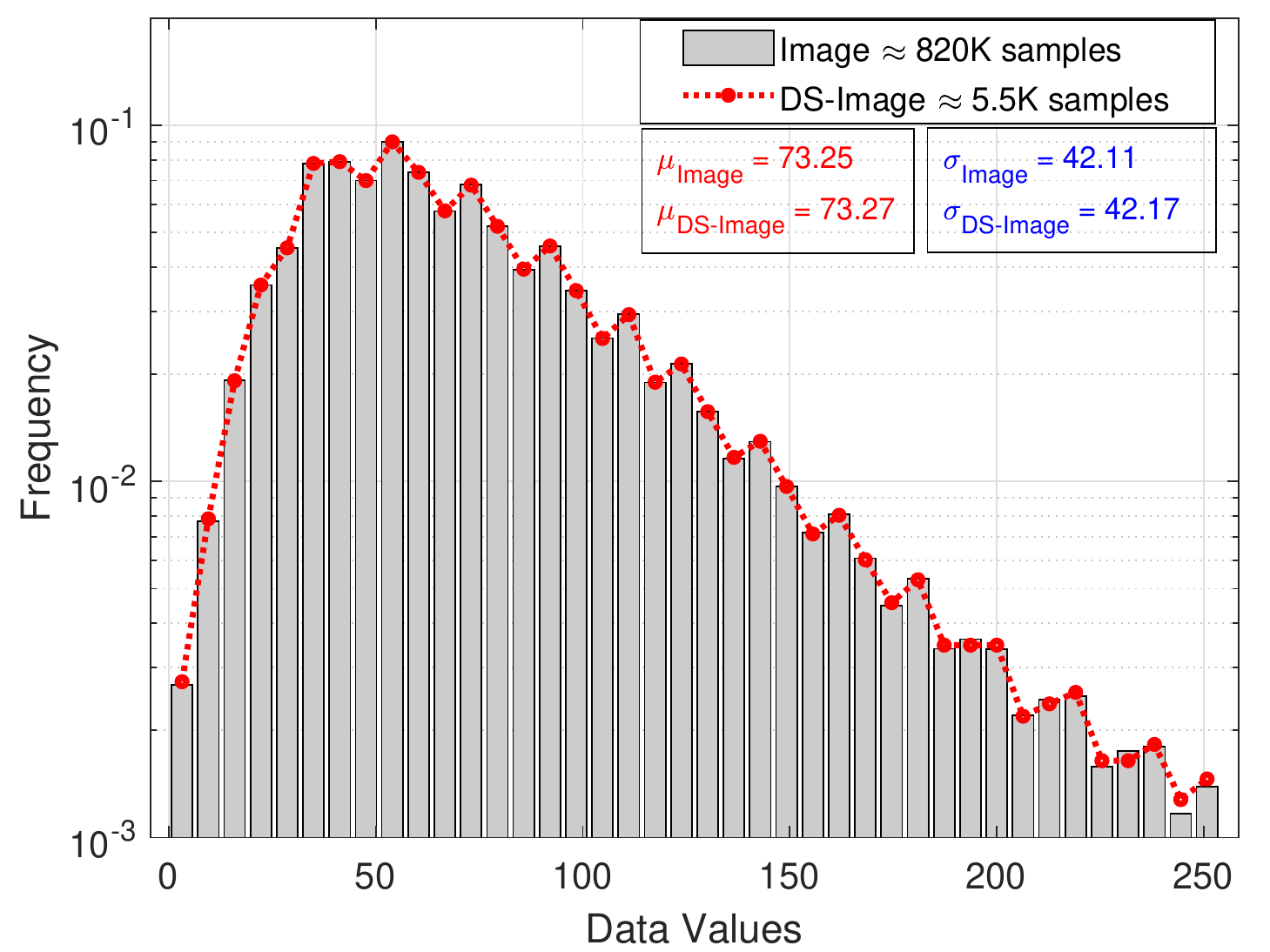}
    \caption{An illustration of the down-sampling operation. The original scene has the size of 908$\times$895, whilst the DS image only has 5500 samples.}\label{fig:histComp}
\end{figure}
\begin{table*}[htbp]
  \centering
  \caption{Statistical significance for modeling performance}
    \resizebox{\linewidth}{!}{\begin{tabular}{cll||cccccccc||cccccccc}
    \toprule
          &       &       & \multicolumn{8}{c||}{Kullback - Leibler (KL) Divergence}        & \multicolumn{8}{c}{Kolmogorov - Smirnov (KS) $p$-value} \\
    \#    & Source$^{\dagger}$  & Scene & \multicolumn{8}{c||}{Distribution Families}                     & \multicolumn{8}{c}{Distribution Families} \\
          &       &       & Rician & Weibull & Lognormal & SR    & GGR   & $\mathcal{G}_0$    & G$\Gamma$D  & GG-Rician & Rician & Weibull & Lognormal & SR    & GGR   & $\mathcal{G}_0$    & G$\Gamma$D    & GG-Rician \\
    \toprule
        1     & T (X) & Agri. & 0.0424 & 0.0791 & \textbf{0.0054} & 0.4595 & 0.2511 & 0.0468 & 0.0207 & 0.0205 & 0.2454 & 0.1161 & \textbf{0.9975} & 0.0000 & 0.0000 & 0.0165 & 0.5575 & 0.9796 \\
    2     & CSM (X) & Urban & 0.0798 & 0.0706 & 0.0208 & 0.3213 & 0.2156 & \textbf{0.0071} & 0.0303 & 0.0160 & 0.0001 & 0.0001 & 0.6234 & 0.0000 & 0.0000 & 0.6055 & 0.0085 & \textbf{0.7841} \\
    3     & T (X) & woSea & 0.0215 & 0.0287 & 0.0306 & 0.2851 & 0.1231 & 0.0152 & \textbf{0.0048} & 0.0158 & 0.5886 & 0.4769 & 0.4153 & 0.0000 & 0.0002 & 0.8958 & 0.9807 & \textbf{0.9943} \\
    \midrule
    4     & T (X) & Agri. & 0.0972 & 0.1096 & 0.0054 & 0.2980 & 0.4446 & \textbf{0.0051} & 0.0411 & 0.0229 & 0.0031 & 0.0039 & 0.9682 & 0.0000 & 0.0000 & 0.6064 & 0.0538 & \textbf{0.9969} \\
    5     & I (X) & Agri. & 0.0124 & 0.0091 & 0.0395 & 0.0270 & 0.0965 & 0.0082 & 0.0058 & \textbf{0.0055} & 0.4461 & 0.8889 & 0.0441 & 0.8360 & 0.0000 & 0.5730 & \textbf{0.9714} & 0.9173 \\
    6     & S1 (C) & Mount & 0.0808 & 0.0637 & 0.0058 & 0.1525 & 0.2279 & \textbf{0.0022} & 0.0190 & 0.0089 & 0.0187 & 0.0280 & 0.9999 & 0.0037 & 0.0082 & \textbf{1.0000} & 0.3187 & 0.9999 \\
    \midrule
    7     & A2 (L) & Mount & 0.0596 & 0.0523 & 0.0097 & 0.0801 & 0.1267 & 0.0201 & 0.0213 & \textbf{0.0033} & 0.0005 & 0.0421 & \textbf{1.0000} & 0.5315 & 0.0000 & 0.2927 & 0.2140 & \textbf{1.0000} \\
    8     & A2 (L) & Land  & 0.0149 & 0.0647 & \textbf{0.0070} & 0.6737 & 0.3283 & 0.1039 & 0.0071 & 0.0079 & 0.6384 & 0.1885 & \textbf{0.9995} & 0.0000 & 0.0000 & 0.0061 & 0.9007 & 0.9986 \\
    9     & CSM (X) & Mixed & 0.0367 & 0.0279 & 0.0125 & 0.1316 & 0.0331 & \textbf{0.0018} & 0.0070 & 0.0124 & 0.1220 & 0.0587 & 0.8340 & 0.0237 & 0.5200 & \textbf{0.9996} & 0.6461 & 0.9677 \\
    \midrule
    10    & T (X) & Urban & 0.0791 & 0.0636 & 0.0248 & 0.1587 & 0.0771 & \textbf{0.0101} & 0.0260 & 0.0105 & 0.0194 & 0.0059 & 0.7721 & 0.0854 & 0.0262 & 0.8454 & 0.1006 & \textbf{0.9543} \\
    11    & T (X) & Land  & 0.0310 & 0.0397 & 0.0251 & 0.2397 & 0.1219 & 0.0224 & \textbf{0.0093} & 0.0148 & 0.4036 & 0.4380 & 0.3415 & 0.0000 & 0.0013 & 0.3444 & \textbf{0.9642} & 0.8399 \\
    12    & CSM (X) & woSea & 0.0299 & 0.0328 & 0.0199 & 0.2106 & 0.1075 & 0.0074 & \textbf{0.0020} & 0.0140 & 0.4907 & 0.5579 & 0.4516 & 0.0000 & 0.0039 & 0.9656 & \textbf{0.9969} & 0.9892 \\
    \midrule
    13    & CSM (X) & wSea  & 0.0362 & 0.0384 & 0.0230 & 0.1876 & 0.1055 & \textbf{0.0035} & 0.0061 & 0.0176 & 0.3369 & 0.3888 & 0.8084 & 0.0000 & 0.0019 & \textbf{1.0000} & 0.8767 & 0.9855 \\
    14    & S1 (C) & Mixed & 0.0540 & 0.0374 & 0.0113 & 0.1050 & 0.0695 & \textbf{0.0065} & 0.0087 & 0.0228 & 0.4248 & 0.4499 & 0.9459 & 0.0216 & 0.1732 & 0.9783 & 0.9887 & \textbf{0.9961} \\
    15    & S1 (C) & Urban & 0.0596 & 0.0447 & 0.0116 & 0.1050 & 0.0707 & 0.0068 & 0.0147 & \textbf{0.0061} & 0.0263 & 0.0500 & 0.9957 & 0.0391 & 0.0813 & 0.9832 & 0.4588 & \textbf{1.0000} \\
    \midrule
    16    & S1 (C) & Land  & 0.0143 & 0.0199 & 0.0274 & 0.2485 & 0.1064 & 0.0376 & \textbf{0.0017} & 0.0082 & 0.7134 & 0.6231 & 0.5312 & 0.0000 & 0.0021 & 0.0580 & \textbf{1.0000} & 1.0000 \\
    17    & A2 (L) & woSea & 0.0383 & 0.1249 & \textbf{0.0106} & 0.6929 & 0.3781 & 0.1050 & 0.0368 & 0.0137 & 0.2005 & 0.0189 & 0.1833 & 0.0000 & 0.0000 & 0.0085 & 0.1565 & \textbf{0.9889} \\
    18    & A2 (L) & wSea  & 0.0721 & 0.1409 & 0.0044 & 0.5520 & 0.3267 & \textbf{0.0028} & 0.0487 & 0.0259 & 0.0361 & 0.0086 & 0.7945 & 0.0000 & 0.0000 & \textbf{1.0000} & 0.0852 & 0.8925 \\
    \midrule
    19    & A2 (L) & Agri. & 0.0440 & 0.0502 & 0.1224 & 0.1763 & 0.0898 & 0.0988 & 0.0519 & \textbf{0.0416} & 0.4339 & 0.3418 & 0.0154 & 0.0000 & 0.0534 & 0.0028 & 0.5163 & \textbf{0.8672} \\
    20    & S1 (C) & woSea & 0.0191 & 0.0126 & 0.0649 & 0.0753 & 0.4786 & 0.0090 & 0.0027 & \textbf{0.0017} & 0.1283 & 0.2062 & 0.4767 & 0.0022 & 0.2018 & 0.9600 & 0.9999 & \textbf{1.0000} \\
    21    & S1 (C) & wSea  & 0.0129 & 0.0130 & 0.0618 & 0.1336 & 0.0597 & 0.0095 & 0.0058 & \textbf{0.0042} & 0.5448 & 0.6510 & 0.0664 & 0.0000 & 0.0088 & 0.9352 & 0.9638 & \textbf{0.9997} \\
    \midrule
    22    & A2 (L) & Mixed & 0.1591 & 0.0154 & 0.0170 & 0.0486 & 0.0309 & 0.0179 & \textbf{0.0062} & 0.0071 & 0.0000 & 0.5340 & 0.1248 & 0.5702 & 0.0901 & 0.4181 & \textbf{0.9439} & 0.7931 \\
    23    & T (X) & Mixed & 0.0943 & 0.1036 & 0.0168 & 0.2968 & 0.2011 & 0.0244 & 0.0432 & \textbf{0.0137} & 0.0003 & 0.0003 & 0.6915 & 0.0000 & 0.0001 & 0.1475 & 0.0151 & \textbf{0.9949} \\
    24    & A2 (L) & Mount & 0.1204 & 0.0981 & 0.0099 & 0.2298 & 0.1240 & \textbf{0.0094} & 0.0432 & 0.0213 & 0.0010 & 0.0005 & 0.7917 & 0.0074 & 0.0000 & \textbf{0.9013} & 0.0247 & 0.7290 \\
    \midrule
    25    & A2 (L) & Mount & 0.0770 & 0.0634 & 0.0067 & 0.1707 & 0.1329 & \textbf{0.0030} & 0.0196 & 0.0151 & 0.0095 & 0.0146 & 0.9985 & 0.0062 & 0.0050 & 0.9991 & 0.2365 & \textbf{0.9999} \\
    26    & S1 (C) & Mixed & 0.0348 & 0.0305 & 0.0227 & 0.1365 & 0.0710 & 0.0186 & 0.0096 & \textbf{0.0063} & 0.1010 & 0.1641 & 0.4146 & 0.0056 & 0.1001 & 0.1961 & 0.8499 & \textbf{0.9829} \\
    27    & S1 (C) & Agri. & 0.0231 & 0.0210 & 0.0204 & 0.1601 & 0.0805 & 0.0229 & \textbf{0.0020} & 0.0104 & 0.5594 & 0.5390 & 0.3248 & 0.0002 & 0.0384 & 0.0937 & 0.9998 & \textbf{1.0000} \\
    \midrule
    28    & S1 (C) & Mixed & 0.0534 & 0.0558 & 0.0211 & 0.0424 & 0.0562 & 0.0216 & 0.0169 & \textbf{0.0054} & 0.0016 & 0.1325 & 0.9974 & 0.4237 & 0.0446 & 0.1227 & 0.3803 & \textbf{1.0000} \\
    29    & S1 (C) & Urban & 0.0314 & 0.0249 & 0.0635 & 0.1205 & 0.0351 & 0.0279 & 0.0204 & \textbf{0.0158} & 0.0769 & 0.1362 & 0.0492 & 0.0390 & 0.0806 & 0.2231 & \textbf{0.8628} & 0.5925 \\
    30    & I (X) & Mount & 0.0102 & 0.0106 & 0.0292 & 0.0607 & 0.0107 & \textbf{0.0031} & 0.0032 & 0.0033 & 0.8417 & 0.6343 & 0.1662 & 0.3683 & 0.7393 & 0.9892 & \textbf{1.0000} & 0.9979 \\
    \midrule
    31    & CSM (X) & wSea  & 0.0496 & 0.0457 & 0.0187 & 0.1649 & 0.1045 & \textbf{0.0025} & 0.0081 & 0.0134 & 0.2191 & 0.1971 & 0.8377 & 0.0000 & 0.0028 & \textbf{1.0000} & 0.7945 & \textbf{1.0000} \\
    32    & A2 (L) & Agri. & 0.0964 & 0.0164 & 0.0218 & 0.0449 & 0.0099 & 0.0144 & 0.0063 & \textbf{0.0029} & 0.0000 & 0.4935 & 0.1978 & 0.5876 & 0.8635 & 0.7306 & \textbf{0.9404} & 0.9186 \\
    33    & S1 (C) & Agri. & 0.0291 & 0.0242 & 0.0438 & 0.1029 & 0.2392 & 0.0120 & 0.0067 & \textbf{0.0042} & 0.3447 & 0.5170 & 0.1514 & 0.0010 & 0.0000 & 0.9644 & 0.9906 & \textbf{1.0000} \\
    \midrule
    34    & T (X) & wSea  & 0.0595 & 0.1049 & 0.0070 & 0.3260 & 0.2282 & \textbf{0.0026} & 0.0227 & 0.0134 & 0.0260 & 0.0034 & 0.8203 & 0.0000 & 0.0000 & 0.2893 & 0.1299 & \textbf{0.9952} \\
    35    & T (X) & Land  & 0.0974 & 0.0095 & 0.0580 & 0.1915 & 0.0098 & \textbf{0.0052} & 0.0098 & 0.0089 & 0.0000 & 0.5289 & 0.0077 & 0.0001 & 0.4488 & \textbf{0.9750} & 0.9729 & 0.6159 \\
    36    & S1 (C) & Mount & 0.0300 & 0.0206 & 0.0304 & 0.1203 & 0.0444 & 0.0141 & 0.0069 & \textbf{0.0020} & 0.0784 & 0.1255 & 0.1368 & 0.0774 & 0.0120 & 0.2721 & 0.9321 & \textbf{0.9997} \\
    \midrule
    37    & I (X) & Urban & 0.0891 & 0.0896 & 0.0079 & 0.2340 & 0.1672 & \textbf{0.0055} & 0.0293 & 0.0184 & 0.0163 & 0.0193 & 0.9887 & 0.0000 & 0.0007 & 0.9957 & 0.2029 & \textbf{0.9999} \\
    38    & S1 (C) & woSea & 0.0096 & 0.0067 & 0.0943 & 0.0754 & 0.0336 & 0.1692 & 0.0051 & \textbf{0.0046} & 0.6923 & 0.8015 & 0.0253 & 0.0004 & 0.0598 & 0.0000 & 0.9716 & \textbf{0.9997} \\
    39    & T (X) & Land  & 0.1026 & 0.0059 & 0.0534 & 0.0754 & 0.0990 & \textbf{0.0042} & 0.0071 & 0.0077 & 0.0000 & \textbf{0.9079} & 0.0091 & 0.0094 & 0.0598 & 0.9975 & 0.9961 & 0.7104 \\
    \midrule
    40    & T (X) & Land  & 0.0440 & 0.0490 & \textbf{0.0086} & 0.2134 & 0.1840 & 0.0330 & 0.0137 & 0.0143 & 0.1834 & 0.0667 & 0.7374 & 0.0000 & 0.0000 & 0.0206 & 0.6126 & \textbf{0.8774} \\
    41    & S1 (C) & Urban & 0.0676 & 0.0522 & 0.0115 & 0.2016 & 0.4171 & \textbf{0.0062} & 0.0181 & 0.0143 & 0.0033 & 0.0046 & \textbf{1.0000} & 0.0009 & 0.1365 & 0.9873 & 0.1692 & 0.9641 \\
    42    & S1 (C) & Agri. & 0.0606 & 0.0583 & \textbf{0.0025} & 0.1962 & 0.2130 & 0.0043 & 0.0179 & 0.0235 & 0.1503 & 0.1973 & \textbf{1.0000} & 0.0000 & 0.0004 & 0.9509 & 0.6262 & 0.9737 \\
    43    & I (X) & Urban & 0.1043 & \textbf{0.0014} & 0.0581 & 0.0313 & 0.0303 & 0.0029 & 0.0059 & 0.0114 & 0.0003 & \textbf{1.0000} & 0.0014 & 0.9939 & 0.0049 & 0.9902 & 0.9574 & 0.5331 \\
    \bottomrule

    \multicolumn{19}{l}{$^{\dagger}$ T (X): TerraSAR-X (X-band), CSM (X): COSMO/SkyMed (X-band), I (X): ICEYE (X-band), S1 (C): Sentinel-1 (C-band), A2 (L): ALOS2 (L-band), w(o)Sea: Sea surface with(out) ships.}
    \end{tabular}}%
  \label{tab:KLKSAll}%
\end{table*}%

\begin{figure}[ht]
\centering
\subfigure[]{
\includegraphics[width=.3\linewidth]{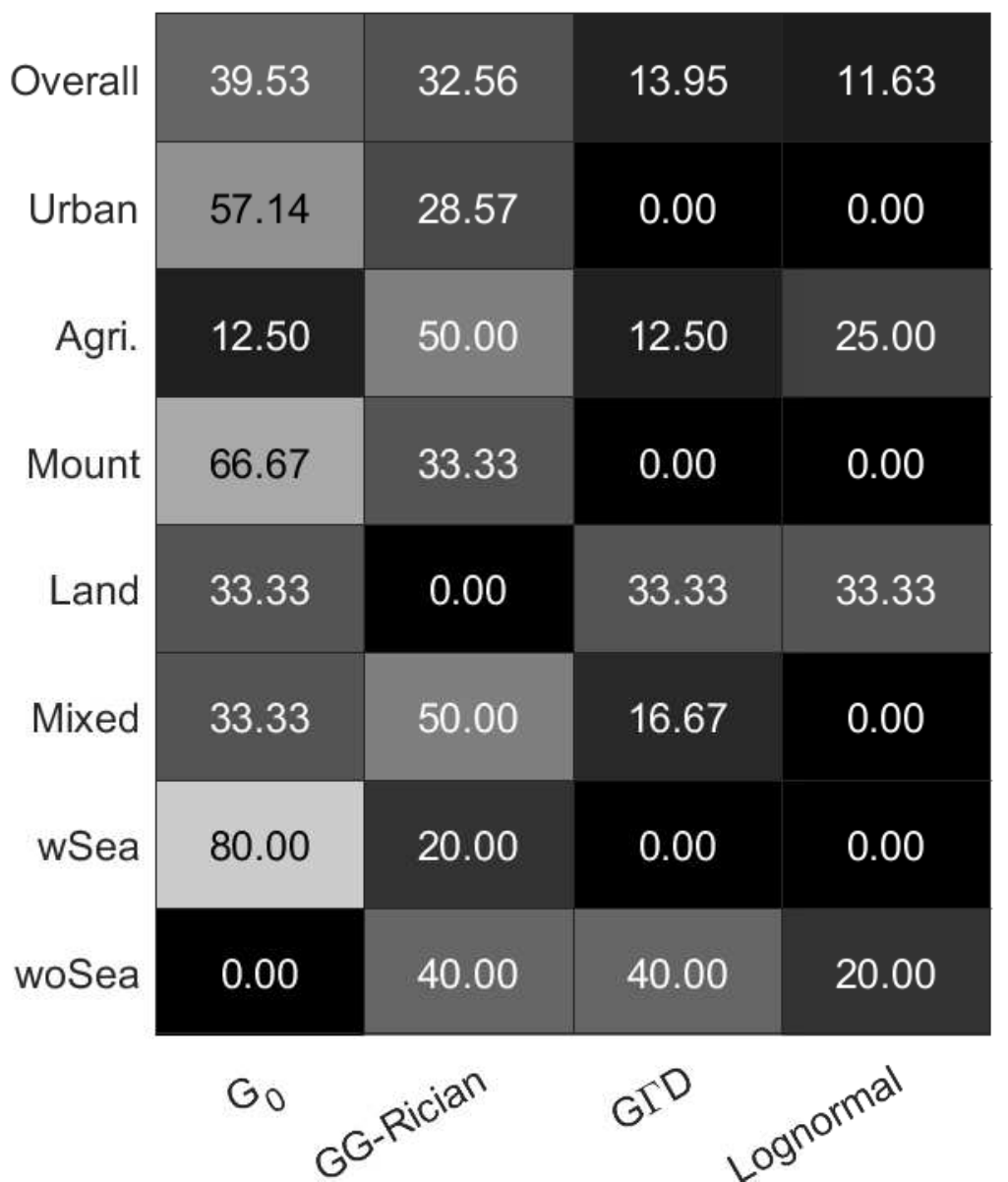}}
\centering
\subfigure[]{
\includegraphics[width=.3\linewidth]{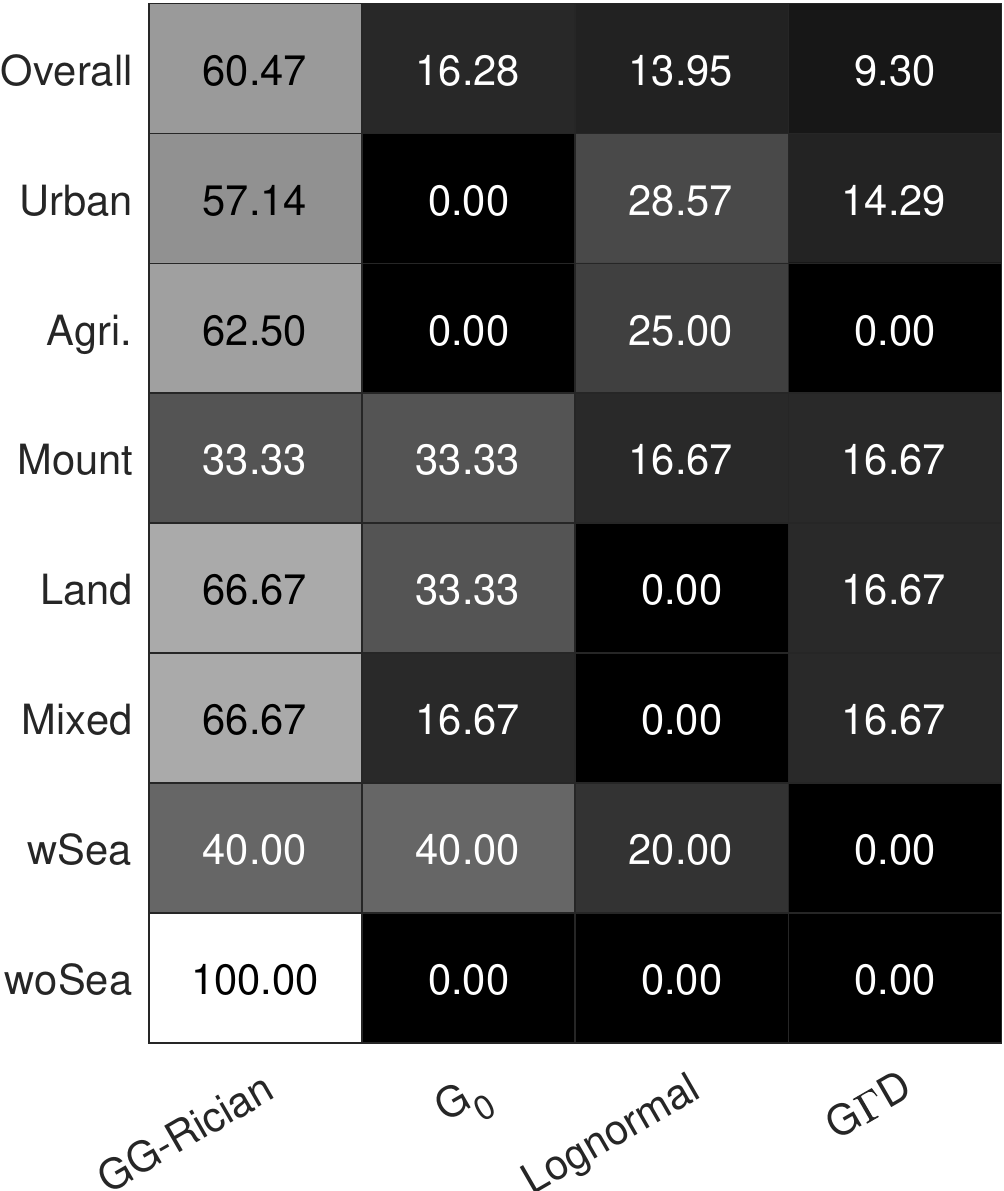}}
\centering
\subfigure[]{
\includegraphics[width=.3\linewidth]{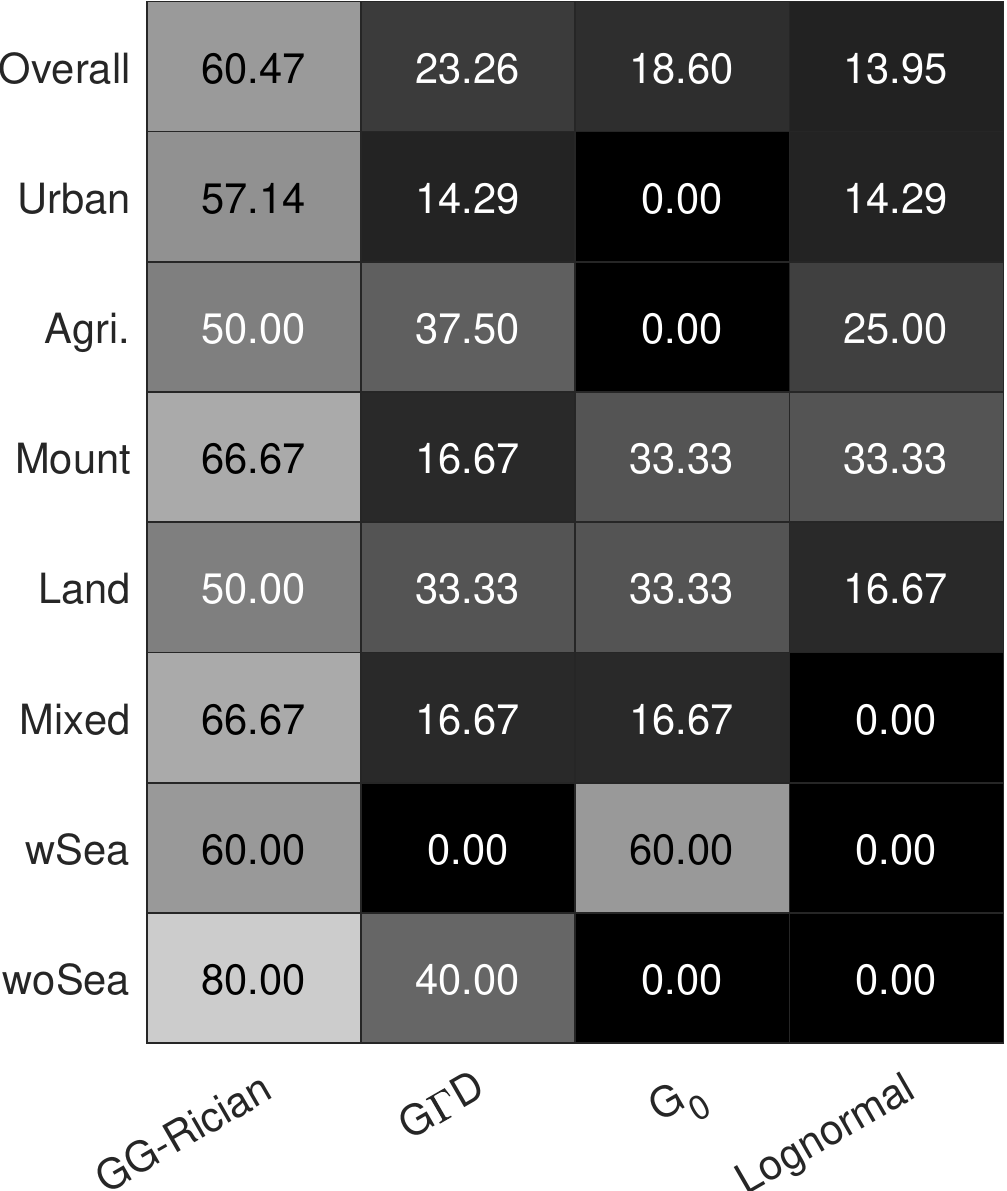}}
\centering
\subfigure[]{
\includegraphics[width=.3\linewidth]{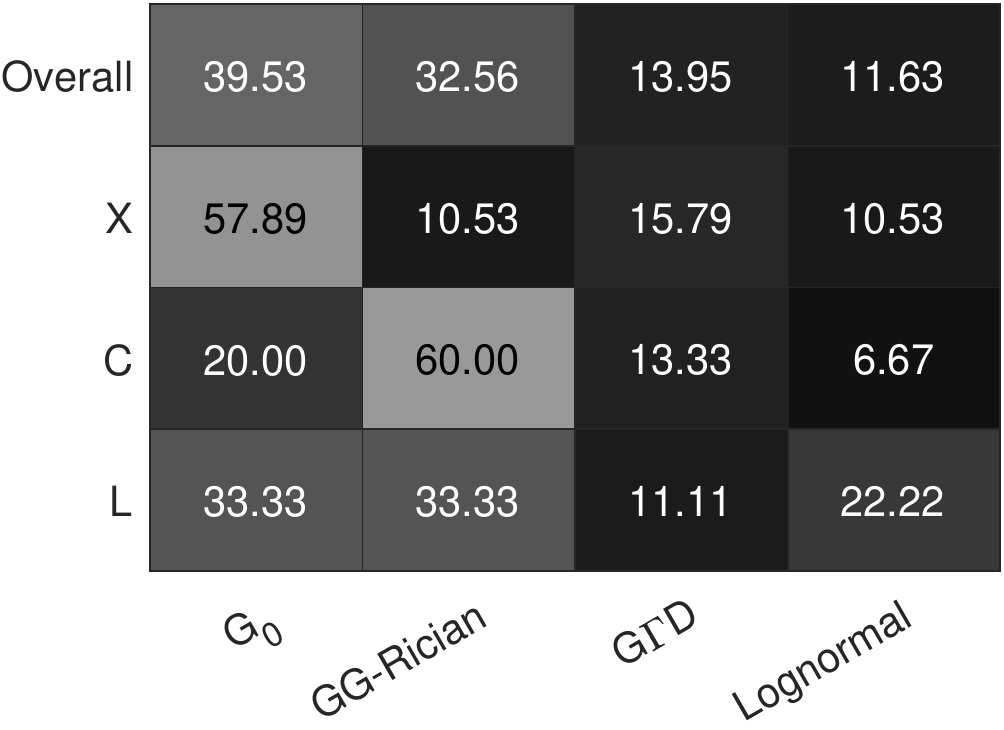}}
\centering
\subfigure[]{
\includegraphics[width=.3\linewidth]{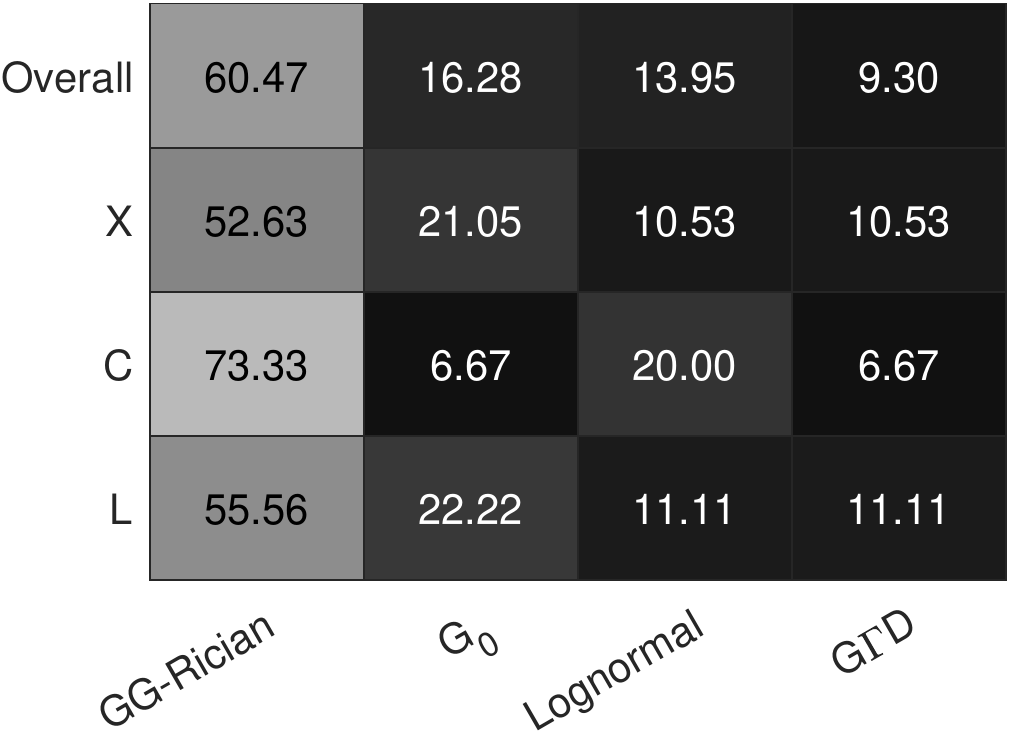}}
\centering
\subfigure[]{
\includegraphics[width=.3\linewidth]{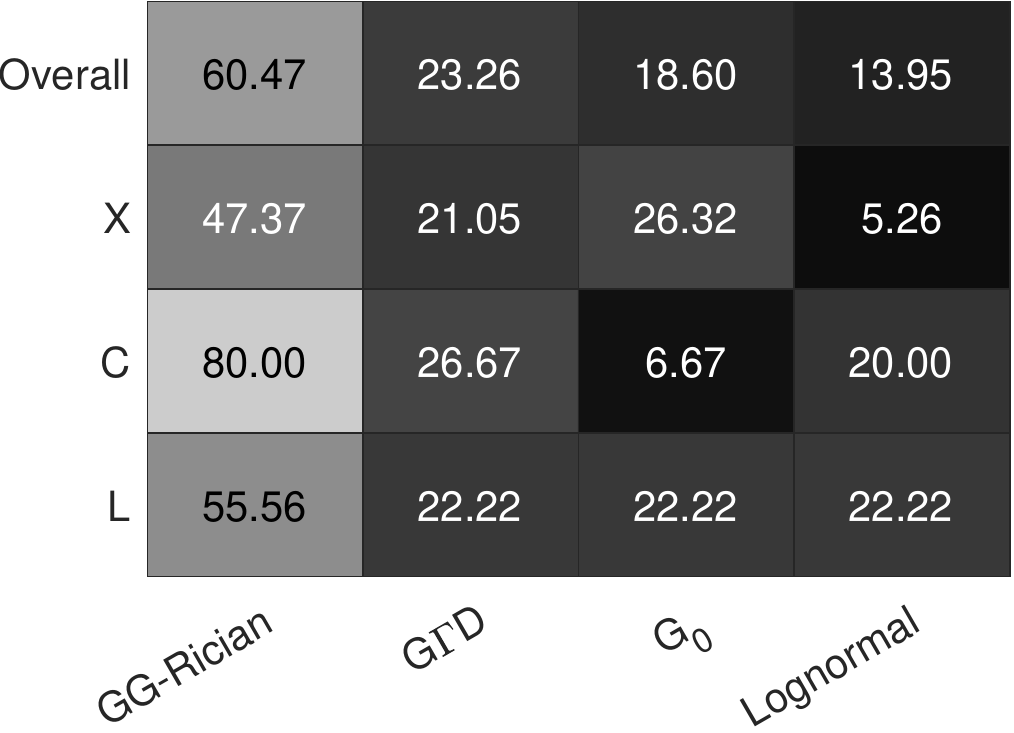}}
\caption{Heatmap representations of the modeling performance for the first four best performing statistical models. (a)-(c) refer to SAR scene comparisons in terms of the KL divergence, KS Score and $p$-value, respectively. (d)-(f) refer to SAR frequency band comparison in terms of the KL divergence, KS Score and $p$-value, respectively. For each sub-figure, models are sorted depending on their overall performance values, in which the first statistical model from the left becomes the overall best.}
\label{fig:klks}\vspace{-0.45cm}
\end{figure}
Initially, each utilized SAR image was down-sampled to have a sample size of around 5000-10000. The down-sampling factor was different for each image since the images had various sizes. We followed a sorted-value down-sampling mechanism, in which we first sort all the pixel intensities/amplitudes and then perform down-sampling. This gives us the opportunity to preserve the correct density shape of the whole scene, and exploit the highly correlated statistical characteristics with lower number of samples. An example is shown in Figure \ref{fig:histComp}. The modeling performance of the proposed statistical model was compared to state-of-the-art statistical models including Rician, Weibull, lognormal, $\mathcal{G}_0$, G$\Gamma$D, SR \cite{kuruoglu2004modeling} and GGR \cite{moser2006sar}. Finally, the corresponding modeling results are presented in Table \ref{tab:KLKSAll}, and in Figures from \ref{fig:klks} to \ref{fig:score}. In Figure \ref{fig:klks}-(a)-(c), we shared the percentages of images for which the models achieved the best performance in terms of KL, KS and $p$-value for various SAR scenes. Figure \ref{fig:klks}-(d)-(f) present the same performance analysis for different SAR frequency bands. 
\begin{figure}[t]
\centering
\subfigure[]{
\includegraphics[width=.3\linewidth]{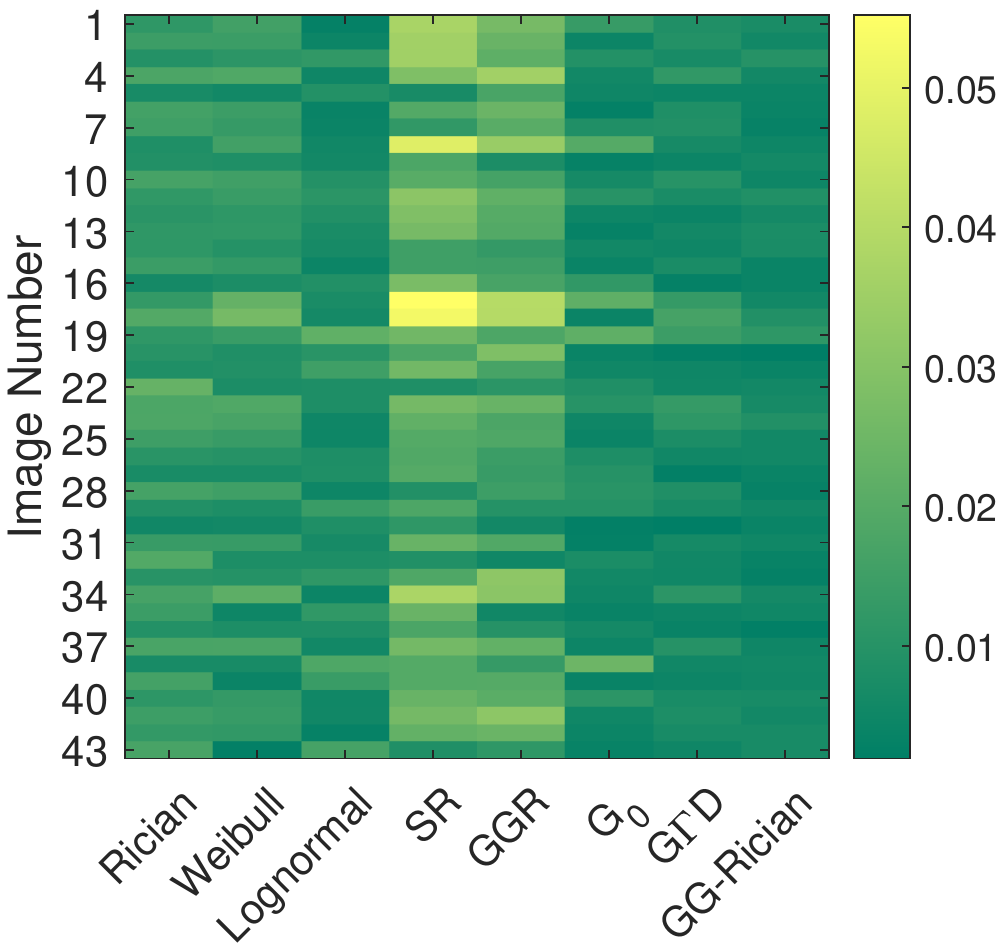}}
\centering
\subfigure[]{
\includegraphics[width=.3\linewidth]{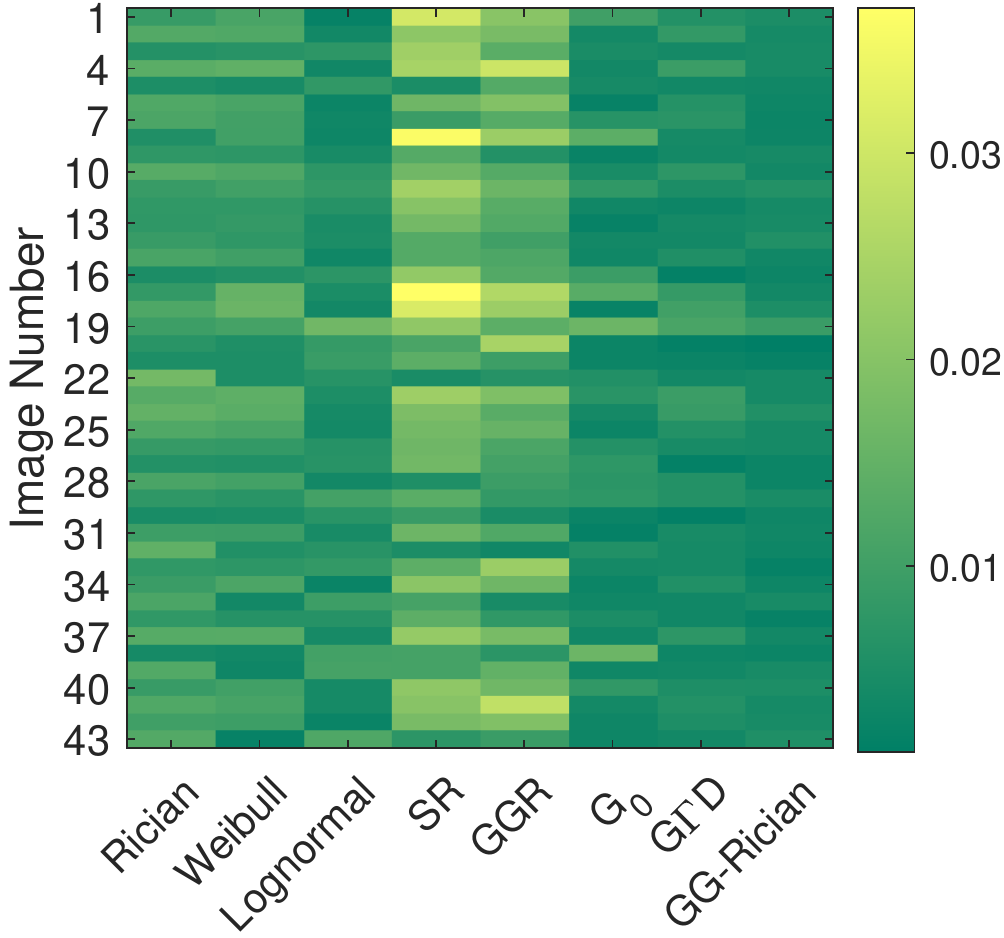}}\\
\centering
\subfigure[]{
\includegraphics[width=.3\linewidth]{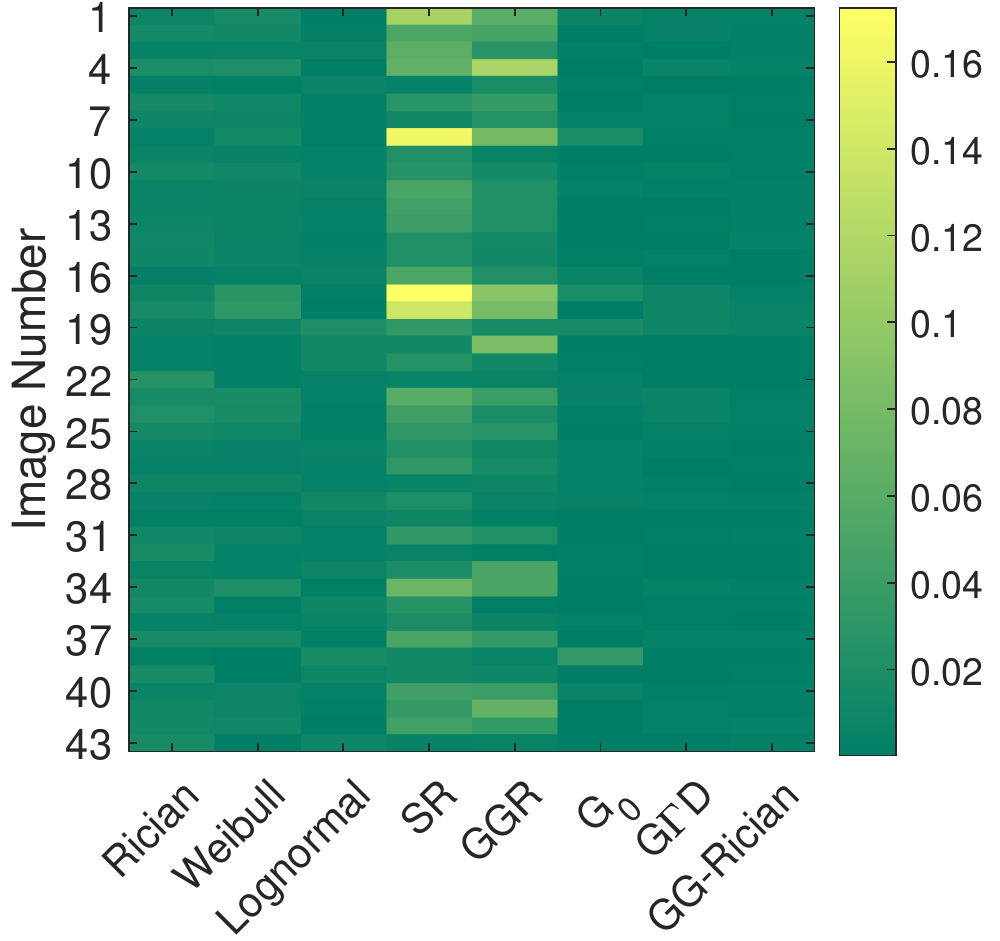}}
\centering
\subfigure[]{
\includegraphics[width=.3\linewidth]{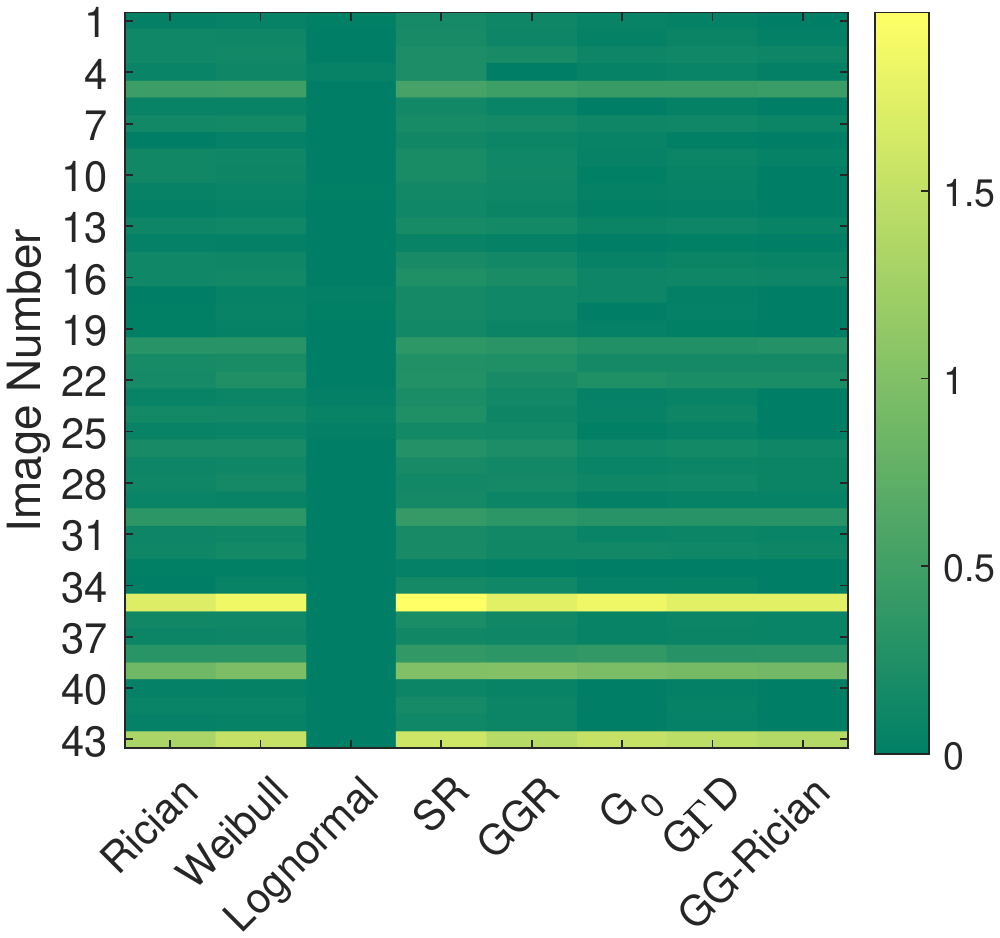}}
\caption{Modeling performance analysis in terms of the (a) RMSE, (b) MAE, (c) Bhattacharyya distance and (d) AICc for all 43 SAR images. For all sub-figures dark regions refer to better fitting performance.}
\label{fig:rmse}\vspace{-0.4cm}
\end{figure}

\begin{figure*}[htbp]
\centering
\subfigure[]{
\includegraphics[width=.3\linewidth]{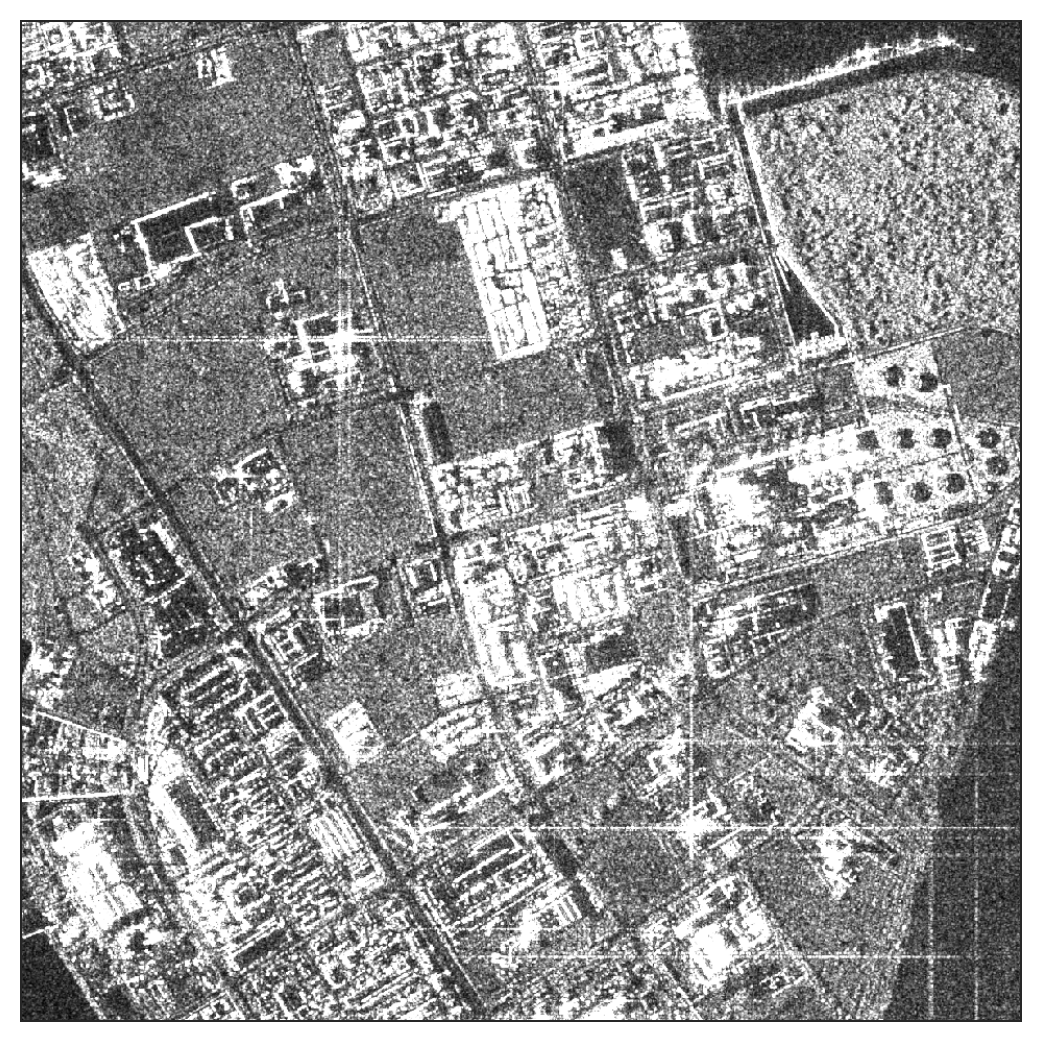}}
\centering
\subfigure[]{
\includegraphics[width=.3\linewidth]{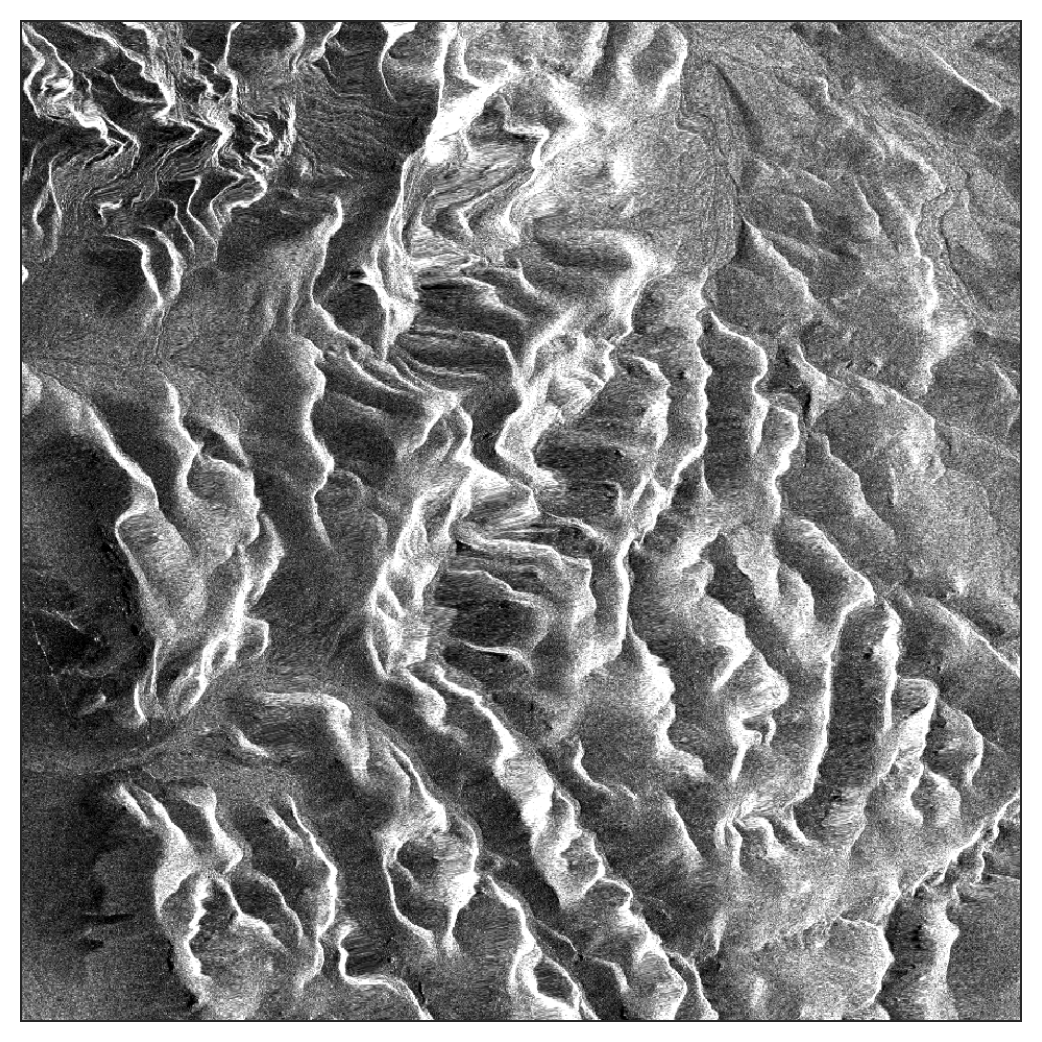}}
\centering
\subfigure[]{
\includegraphics[width=.3\linewidth]{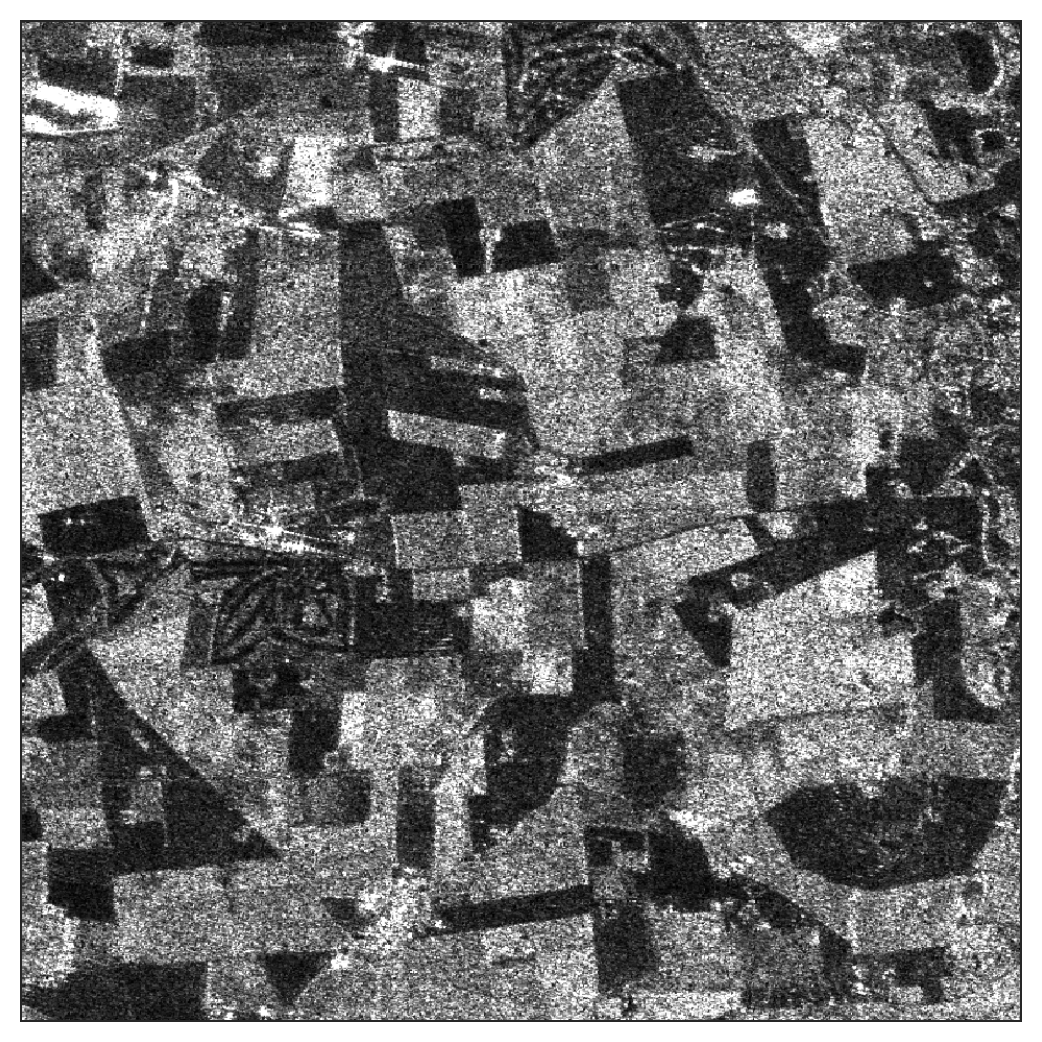}}\hfill
\centering
\subfigure[]{
\includegraphics[width=.3\linewidth]{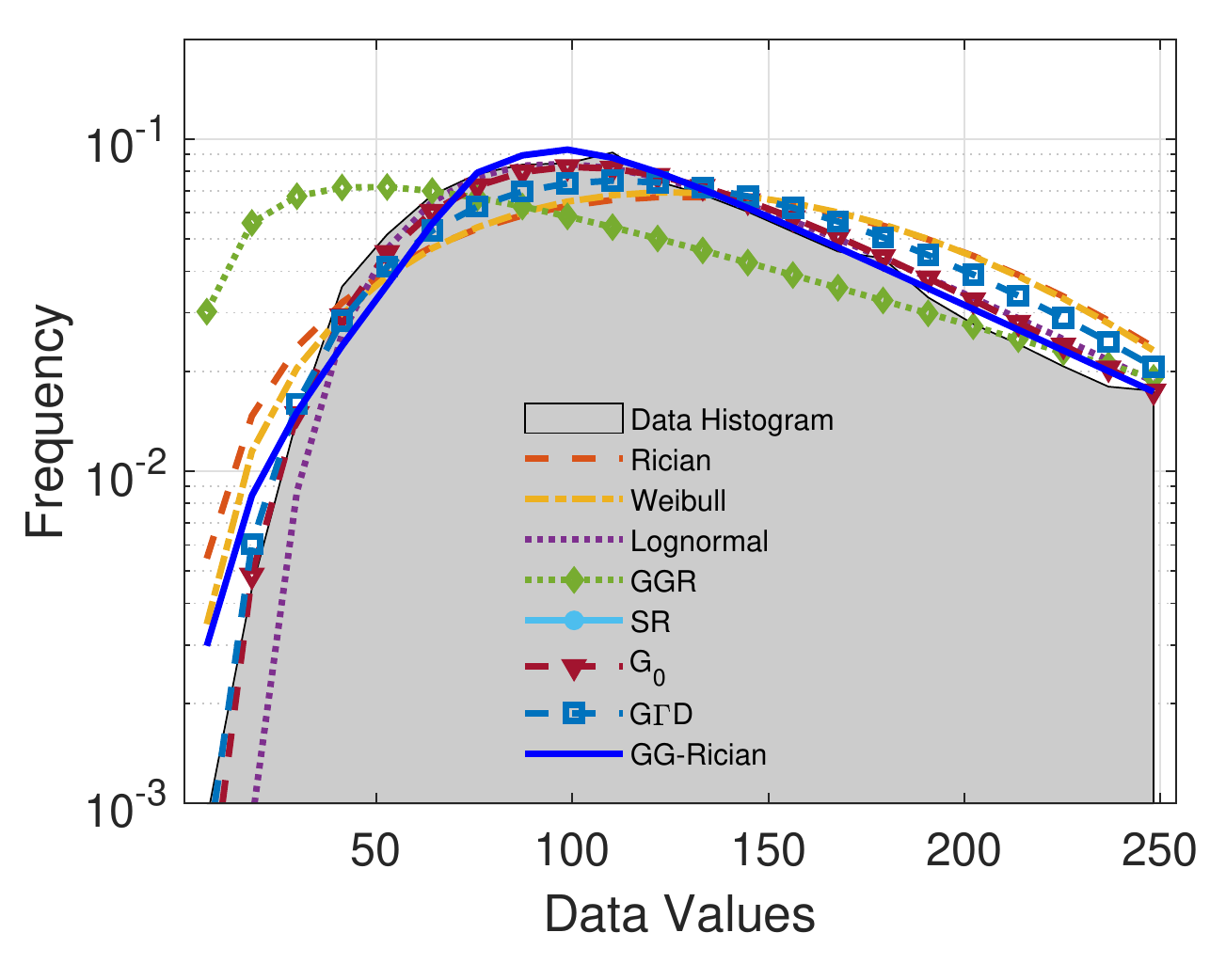}}
\centering
\subfigure[]{
\includegraphics[width=.3\linewidth]{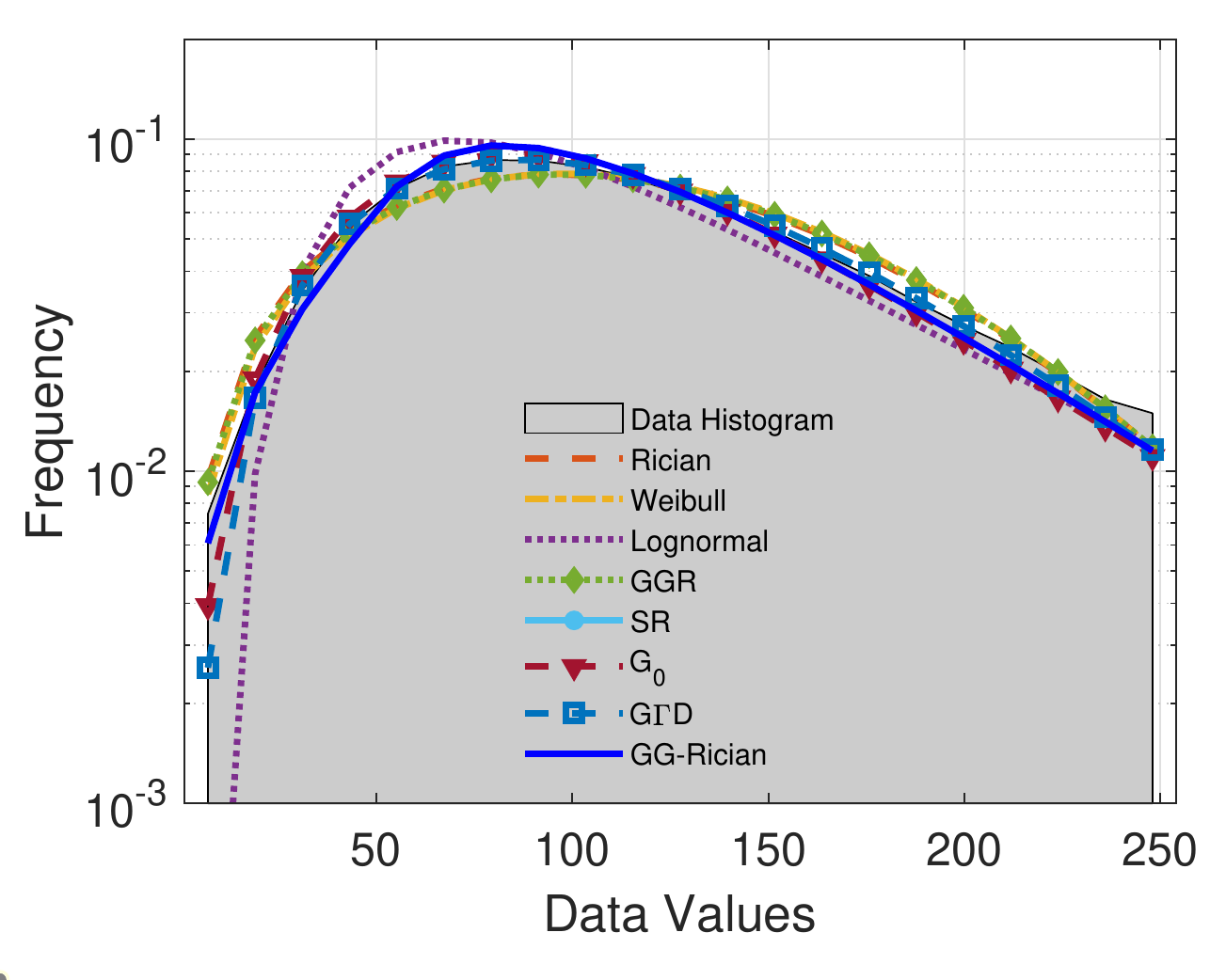}}
\centering
\subfigure[]{
\includegraphics[width=.3\linewidth]{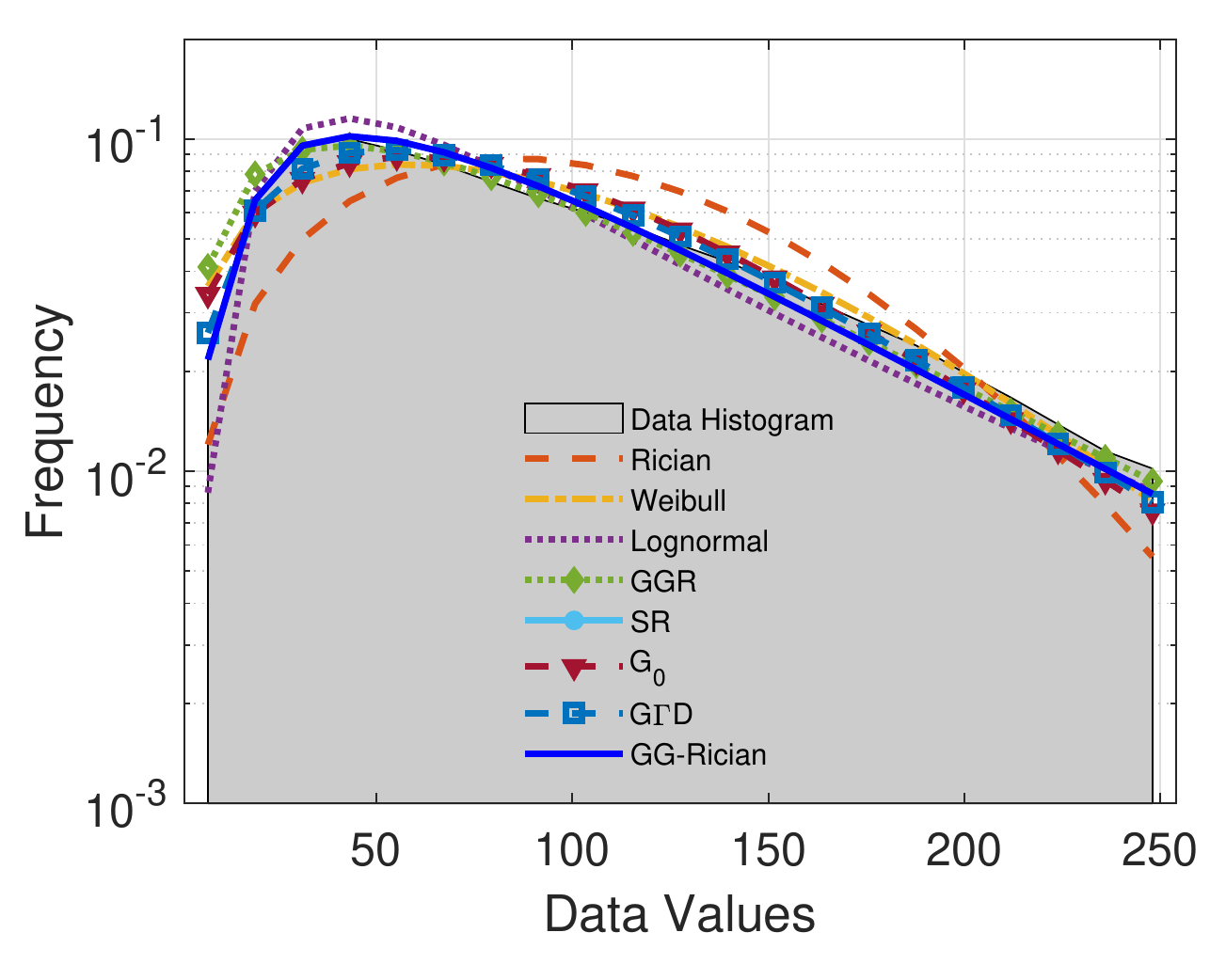}}\hfill
\centering
\subfigure[]{
\includegraphics[width=.3\linewidth]{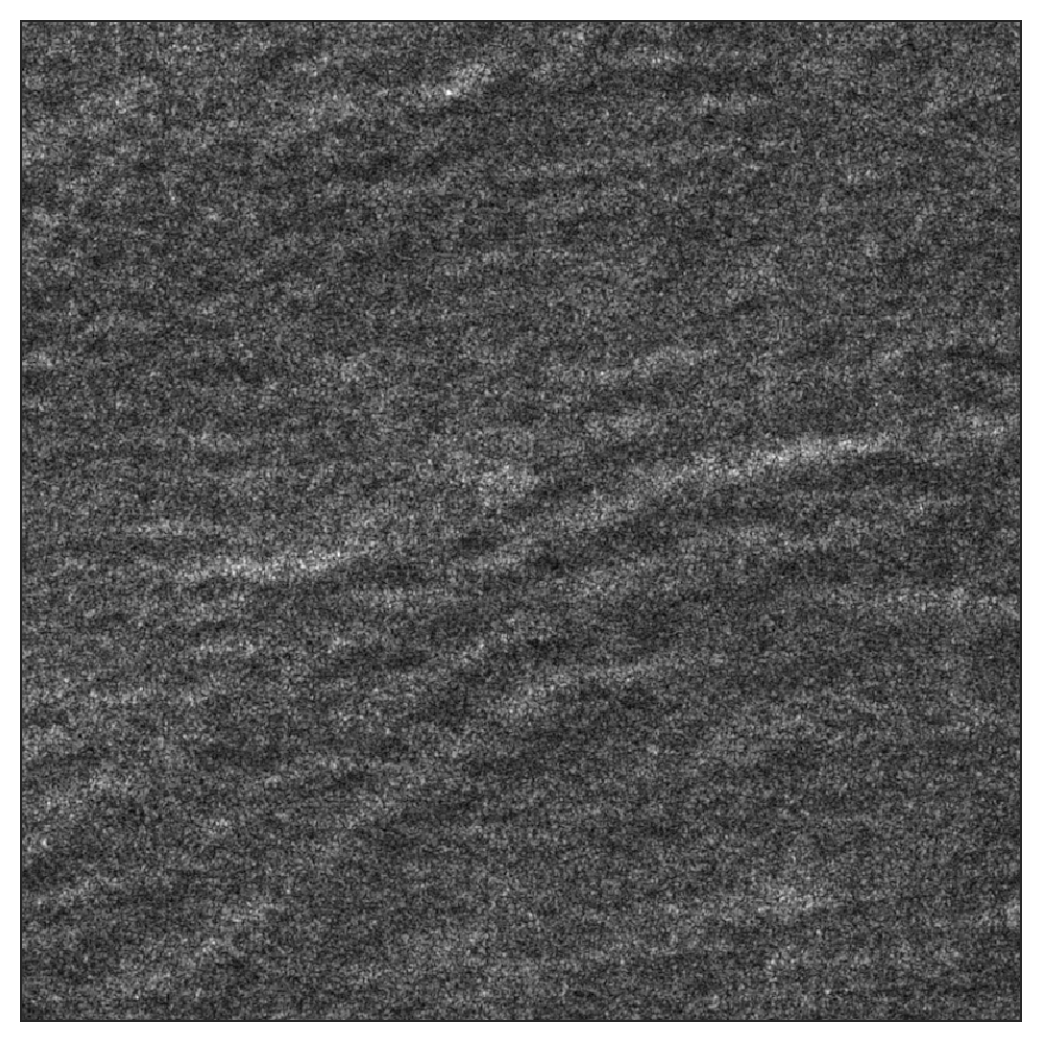}}
\centering
\subfigure[]{
\includegraphics[width=.3\linewidth]{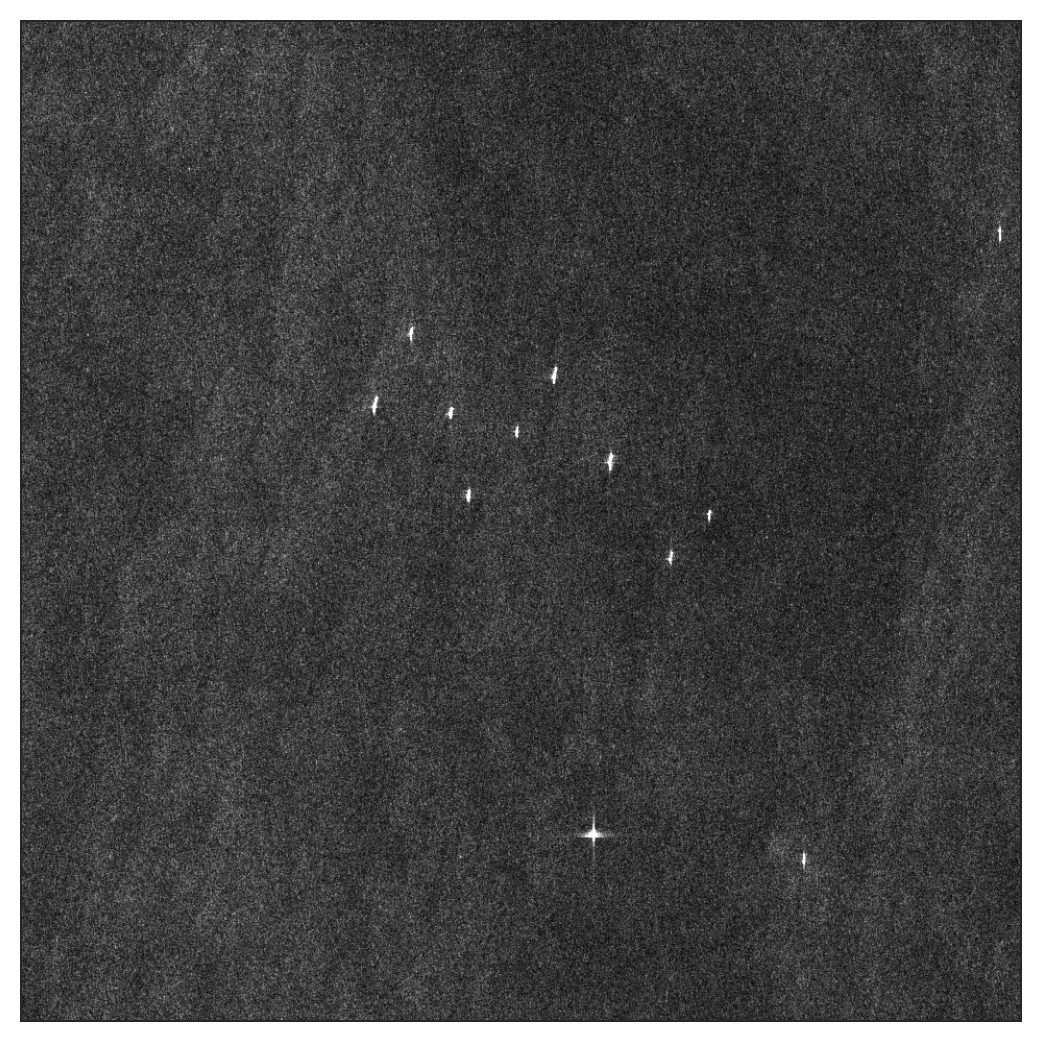}}
\centering
\subfigure[]{
\includegraphics[width=.3\linewidth]{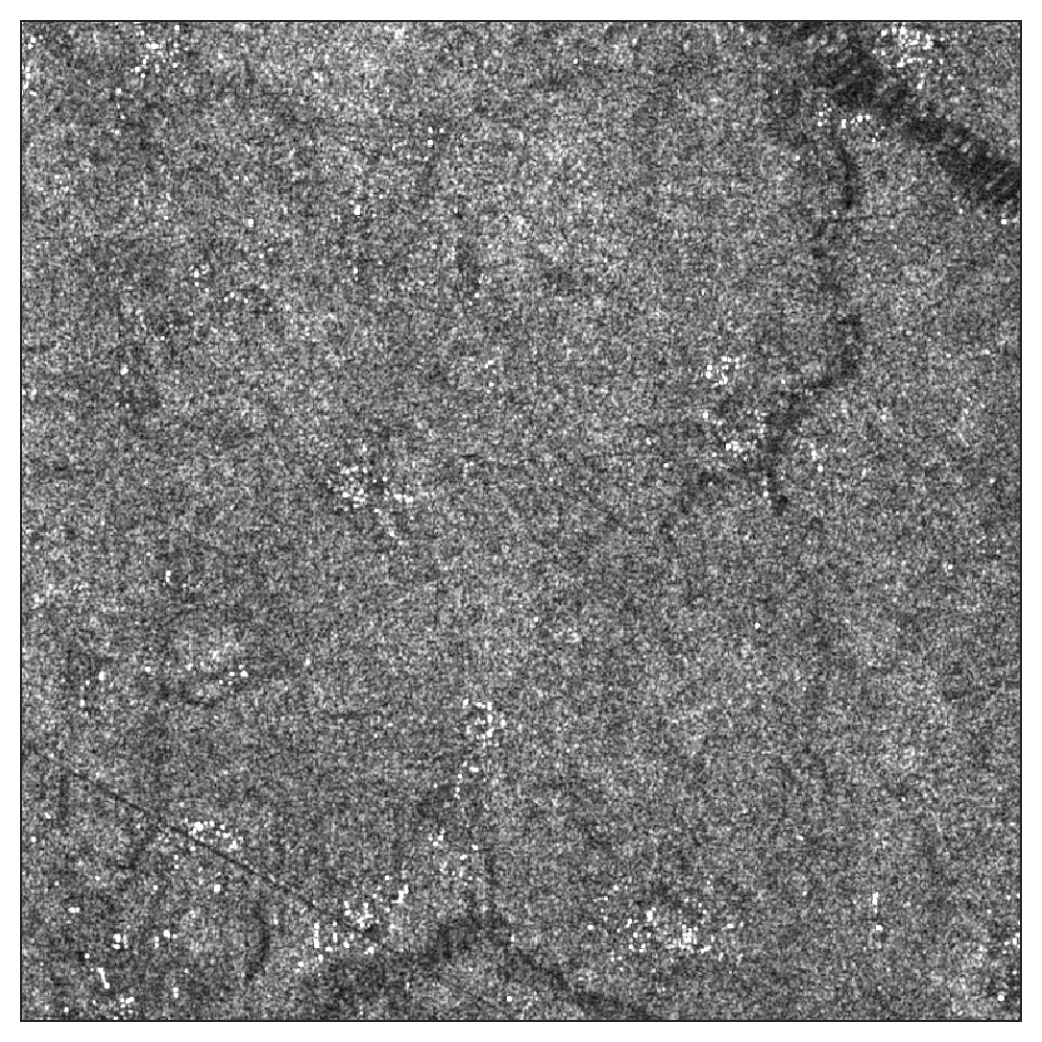}}\hfill
\centering
\subfigure[]{
\includegraphics[width=.3\linewidth]{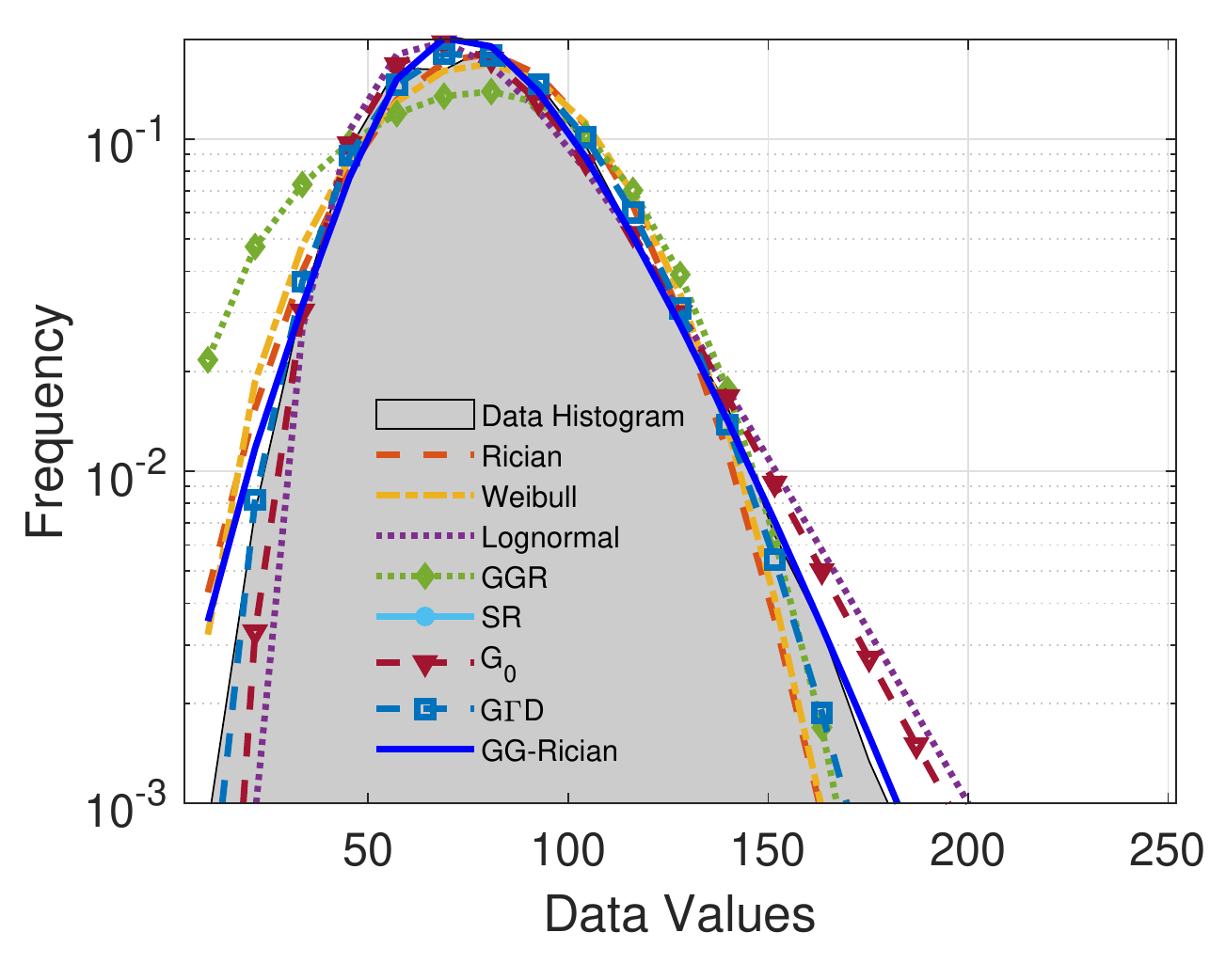}}
\centering
\subfigure[]{
\includegraphics[width=.3\linewidth]{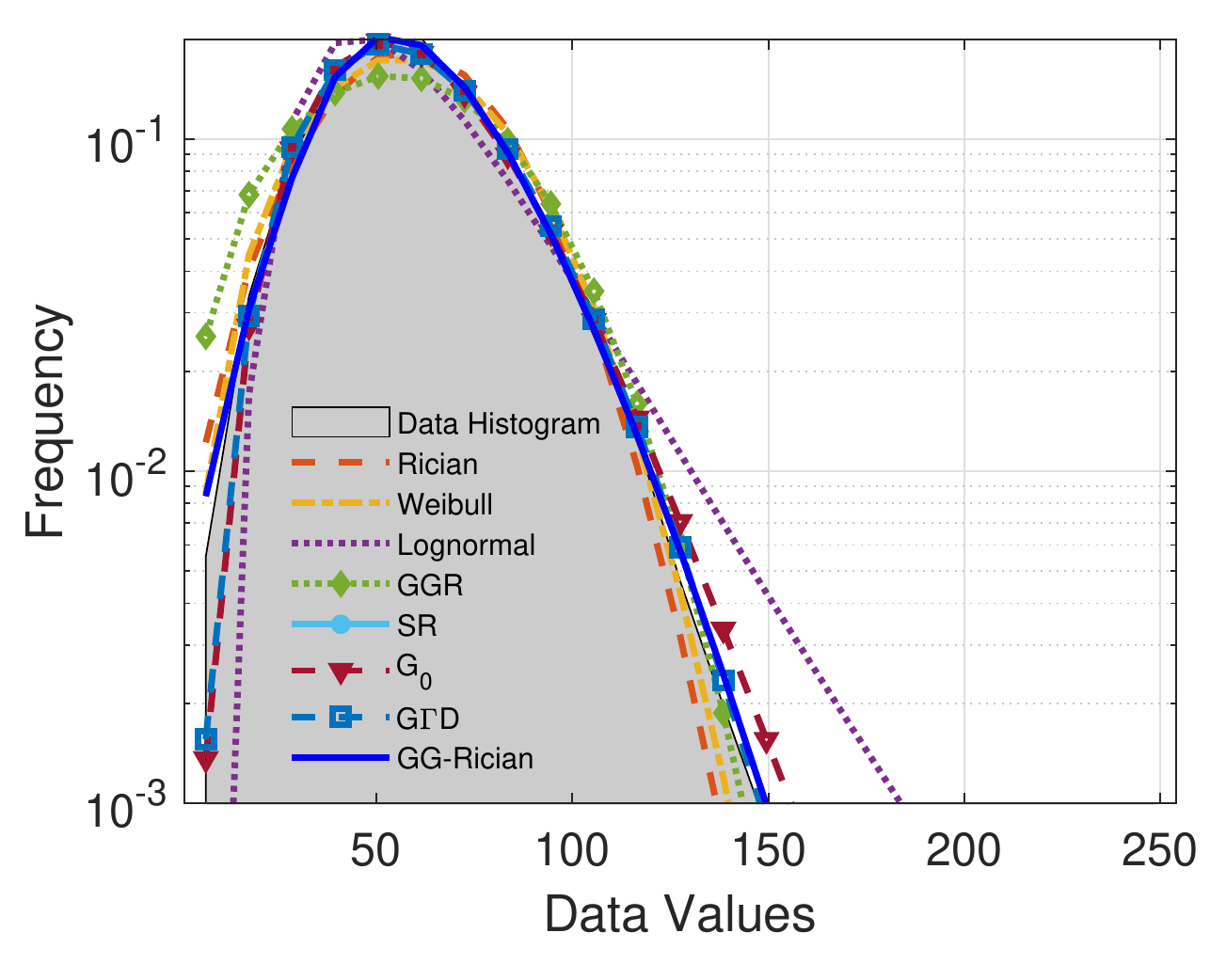}}
\centering
\subfigure[]{
\includegraphics[width=.3\linewidth]{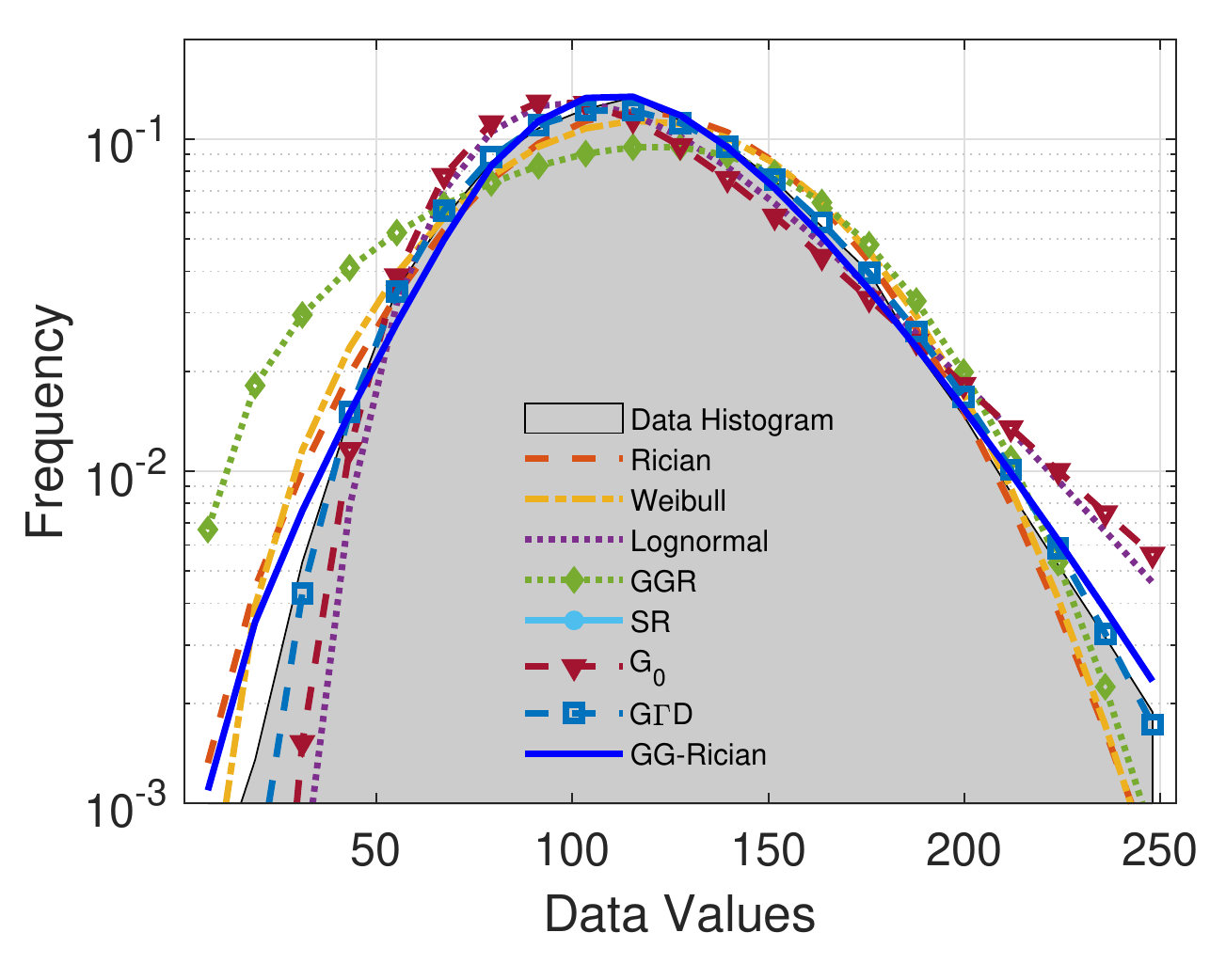}}\hfill
\caption{Visual evaluation of SAR amplitude models. Amplitude SAR images from scenes of (a) Urban (X) - COSMO/Sky-Med, (b) Mountain (X) - ICEYE, (c) Agricultural (L) - ALOS2, (g) Sea / woShips (X) - TerraSAR-X, (h) Sea / wShips (C) - Sentinel-1, and (i) Land (C) - Sentinel-1. Sub-figures in (d)-(f) and (j)-(l) refer to the corresponding modeling results in log-pdf scale for amplitude images in (a)-(c) and (g)-(i), respectively.}
\label{fig:realFigures}
\end{figure*}

The statistical significance measures for all 43 SAR images utilized in this paper are presented in the first rows of Figure \ref{fig:klks}. In terms of the KL divergence results in Figure \ref{fig:klks}-(a) for overall percentages, the $\mathcal{G}_0$ and GG-Rician models perform best. $\mathcal{G}_0$ is slightly better than the proposed method for urban and mountain scenes, but fails to model dominantly homogeneous scenes such as the sea surface without ships and the agricultural one. Moreover, the Lognormal model also shows arguably better modeling performance compared to the $\mathcal{G}_0$ and G$\Gamma$D models for land scenes, whereas G$\Gamma$D successfully models scenes of sea surface and land cover. In terms of KS Score and $p$-value results in Figure \ref{fig:klks}-(b) and (c), statistical significance analysis shows that for around 60\% of the images, the most suitable distribution is the proposed GG-Rician distribution. The GG-Rician model is the best model for all SAR scenes in terms of the KS Score and $p$-value. When we specifically examine the frequency band performance in Figure \ref{fig:klks}-(c)-(f), similar to the scene-specific results, the GG-Rician models appear as the best performing model specifically for the C and L bands, whilst $\mathcal{G}_0$ appears to be a robust model for X band SAR images.

\begin{figure}[ht]
\centering
\subfigure[]{
\includegraphics[width=.48\linewidth]{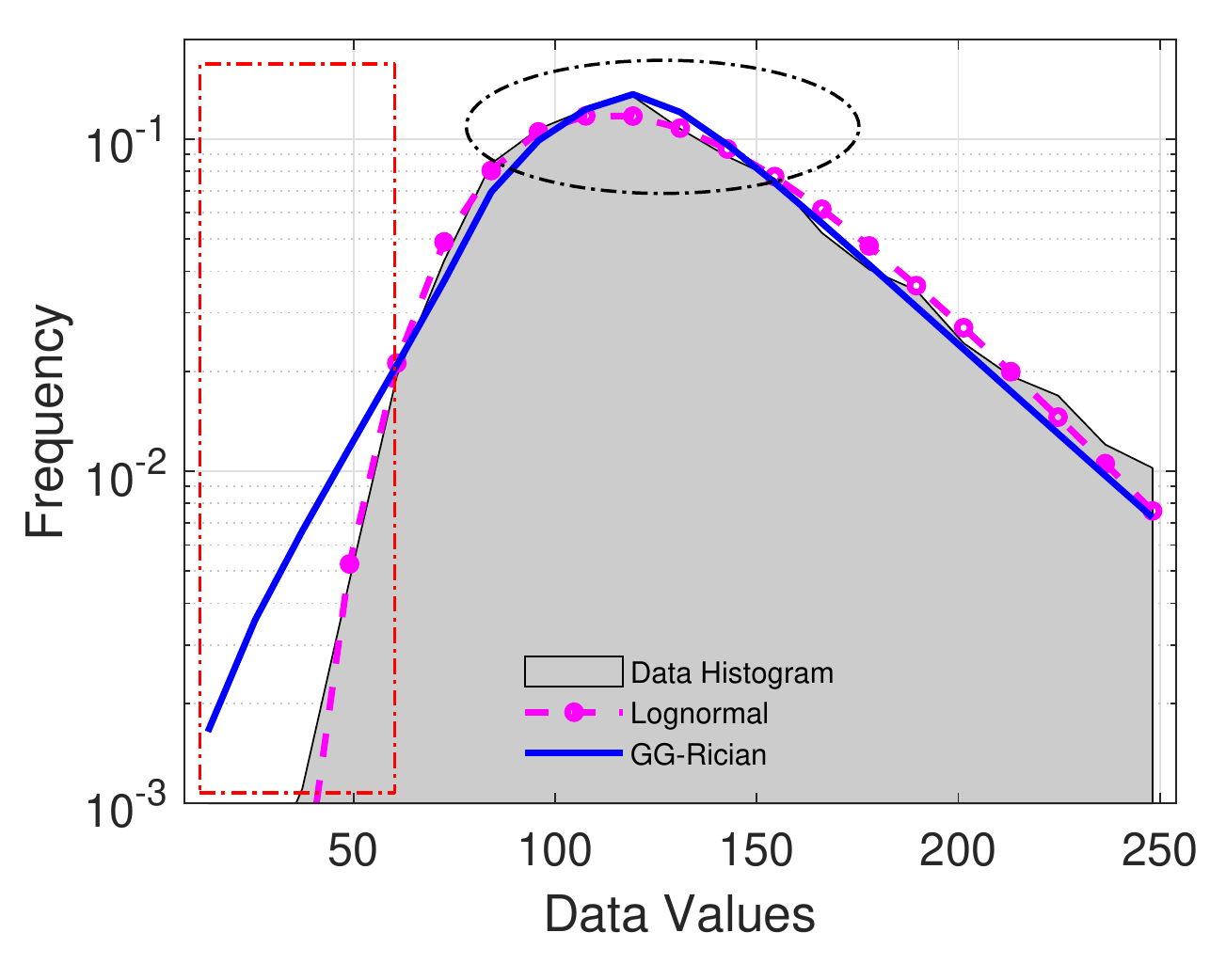}}
\centering
\subfigure[]{
\includegraphics[width=.48\linewidth]{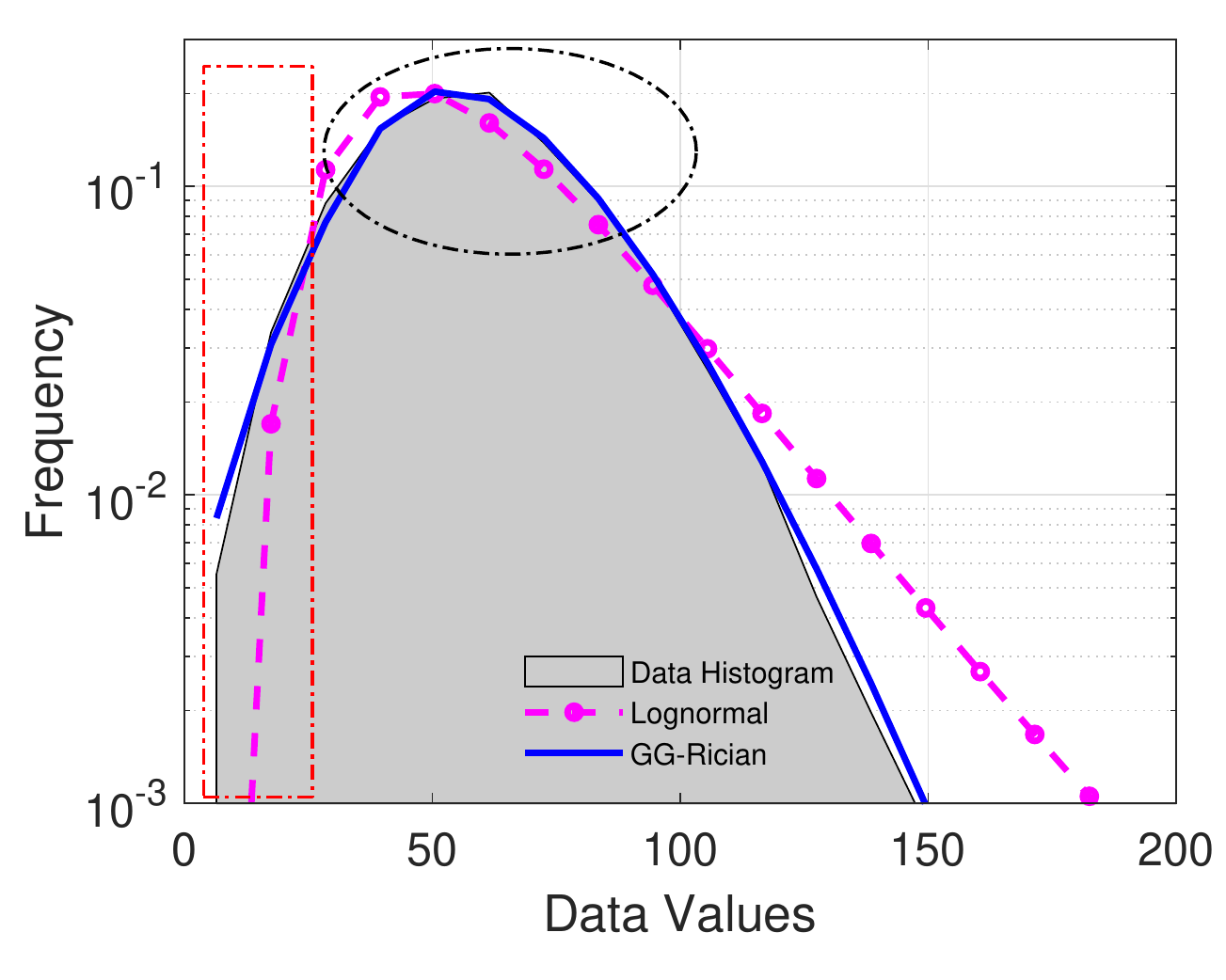}}
\caption{A direct comparison between the GG-Rician and Lognormal models for two different cases. (a) The best in terms of KL Div. is Lognormal, whilst in terms of KS, GG-Rician is the best model. (b) GG-Rician is the best for KL and KS measures.}
\label{fig:KLKS}
\end{figure}
\begin{figure}[ht]
\centering
\includegraphics[width=0.6\linewidth]{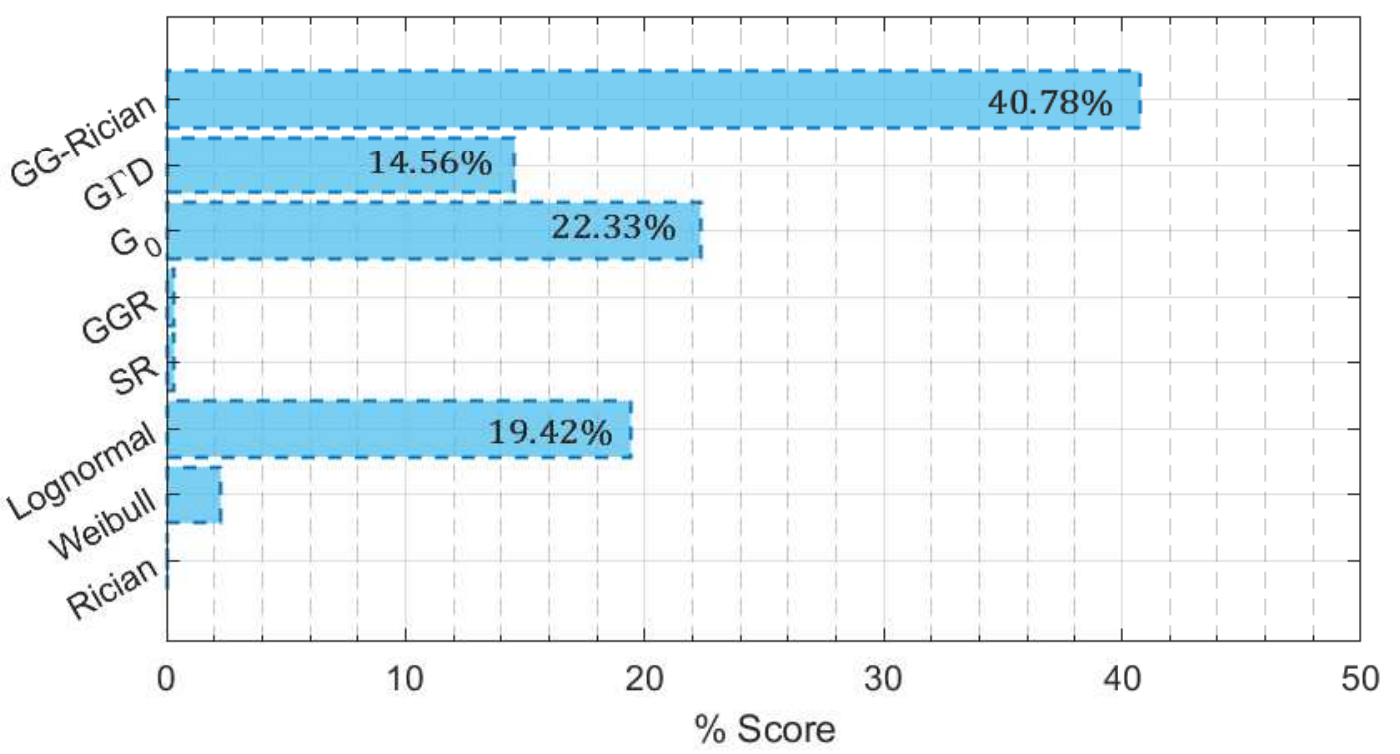}
\caption{Overall score percentages for performance analysis. All 8 models are assessed in terms of 7 measures utilized. Overall percentage of score values are calculated.}
\label{fig:score}
\end{figure}

\begin{figure*}[ht]
\centering
\subfigure[]{
\includegraphics[width=.48\linewidth]{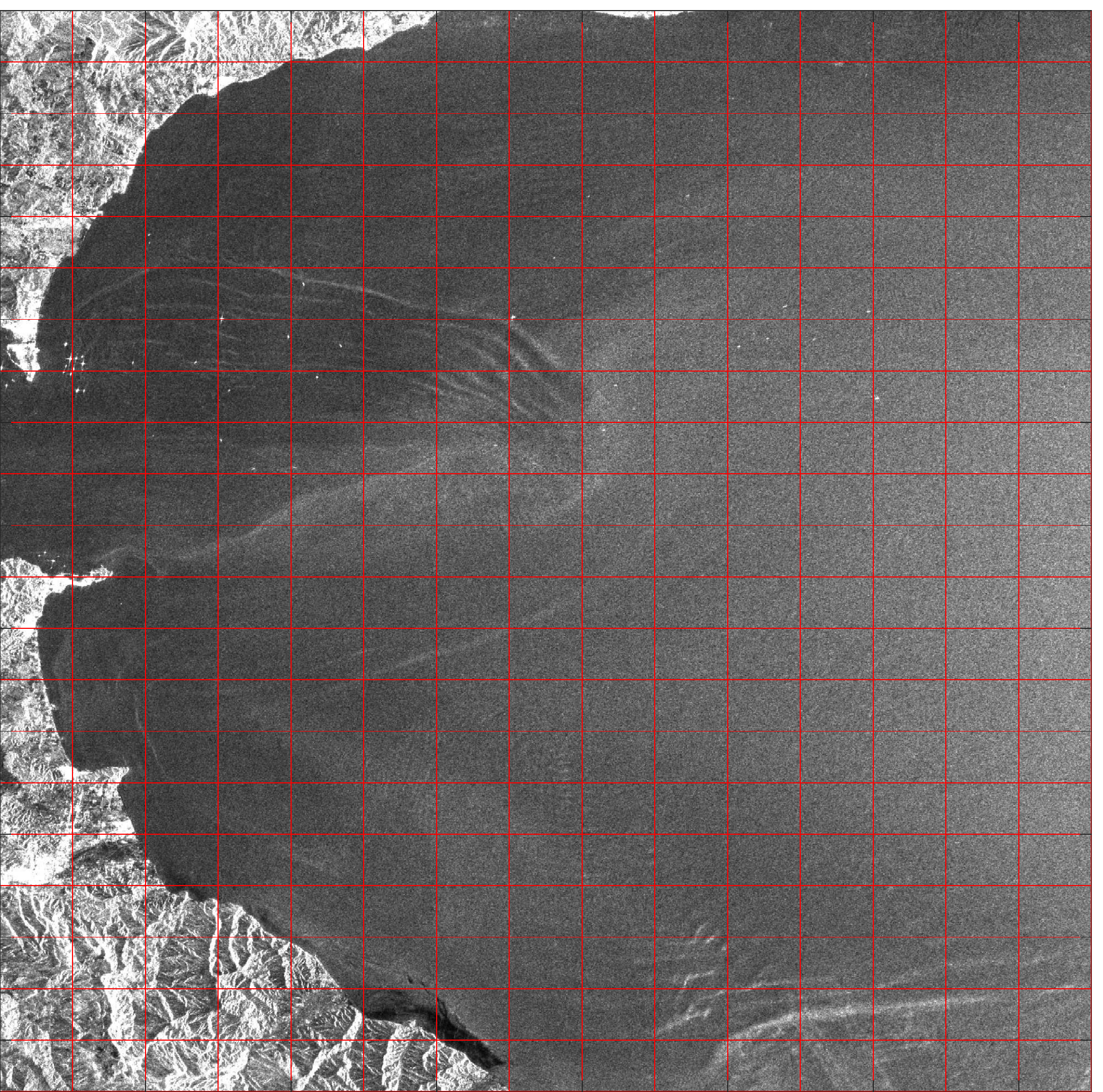}}
\centering
\subfigure[]{
\includegraphics[width=.48\linewidth]{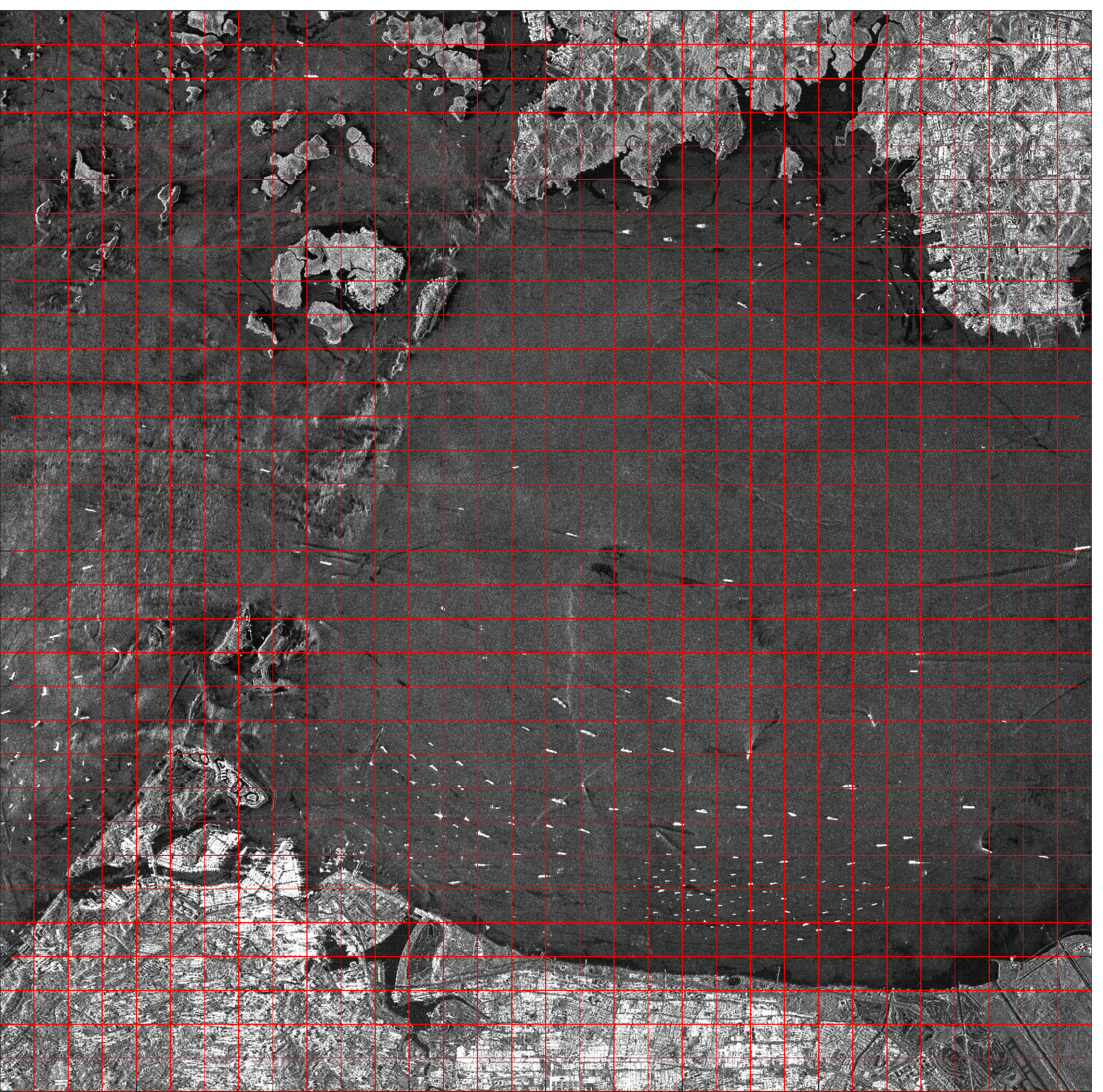}}\\
\centering
\subfigure[]{
\includegraphics[width=.15\linewidth]{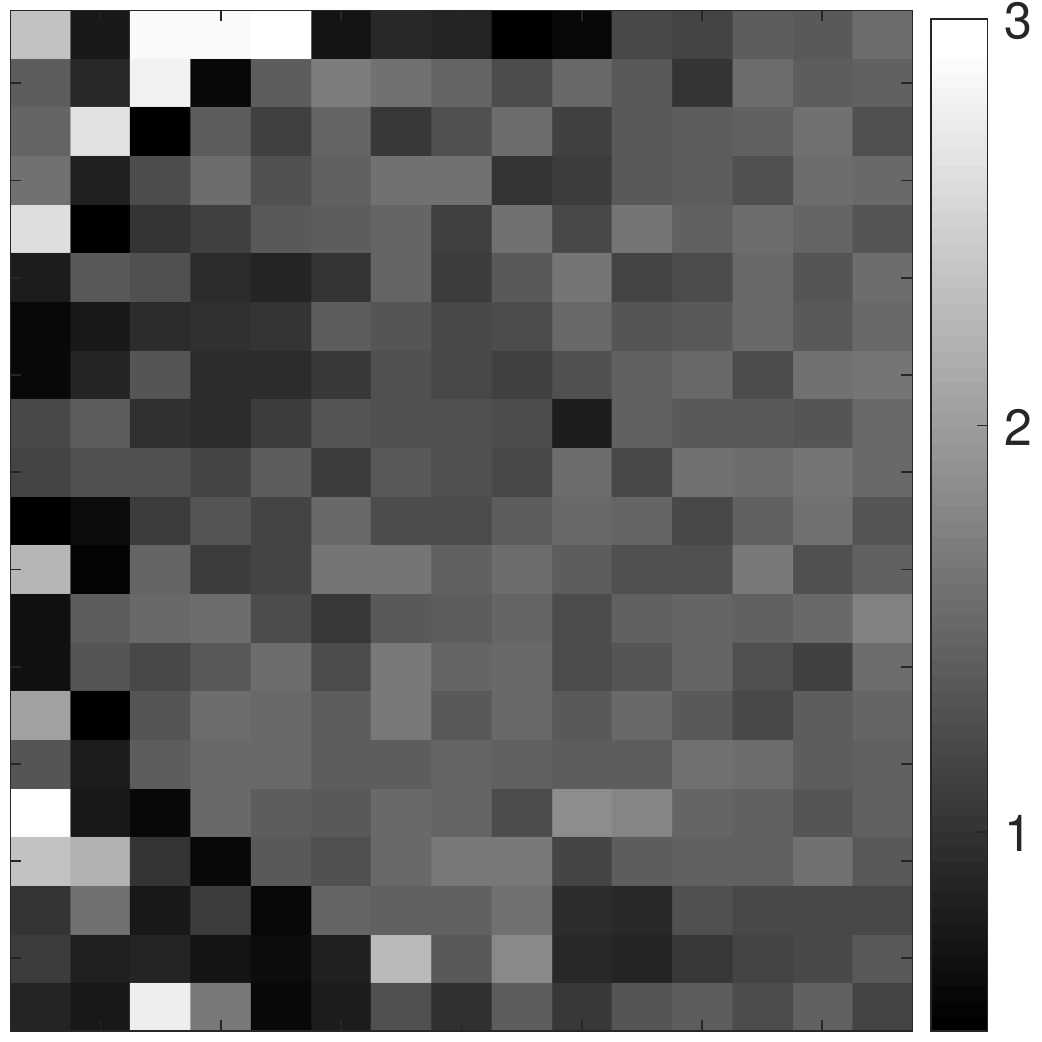}}
\centering
\subfigure[]{
\includegraphics[width=.15\linewidth]{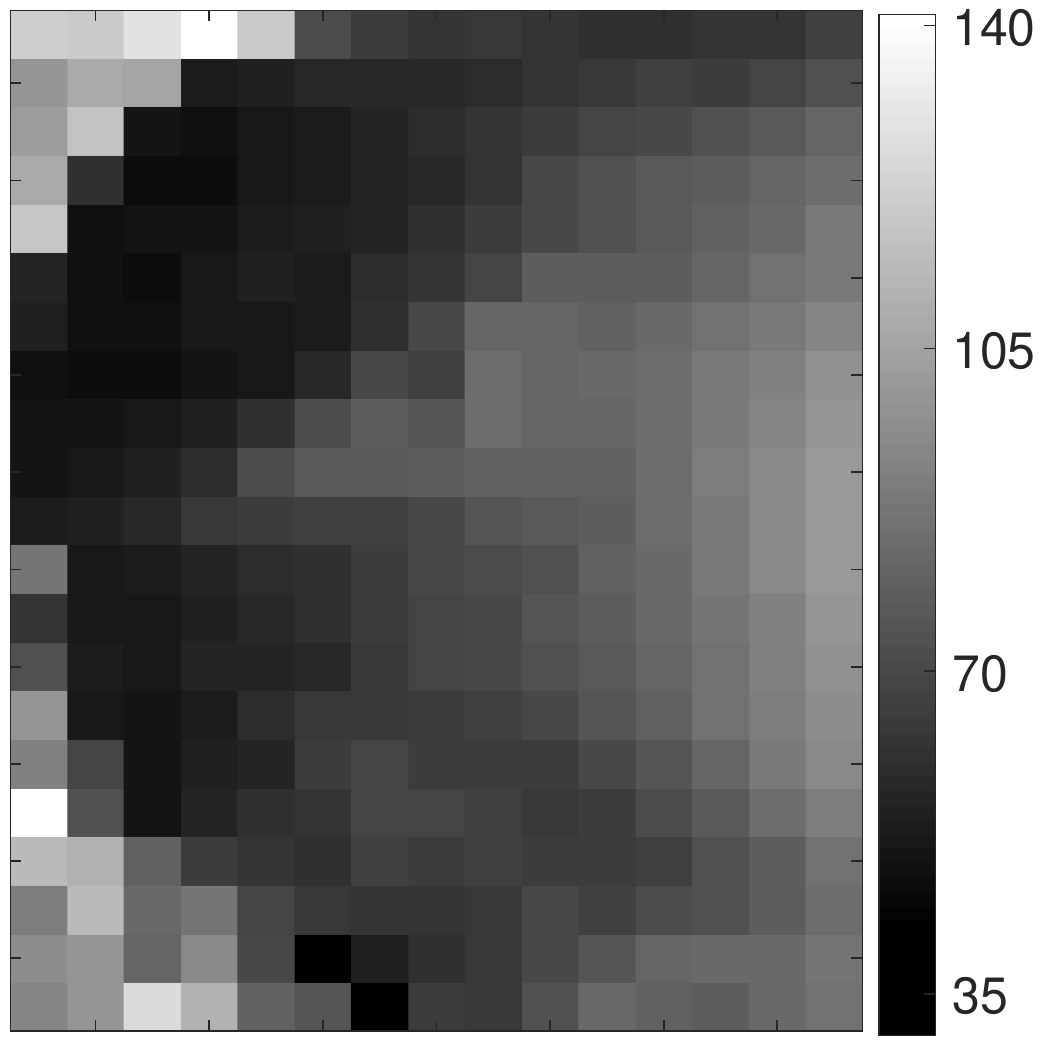}}
\subfigure[]{
\includegraphics[width=.15\linewidth]{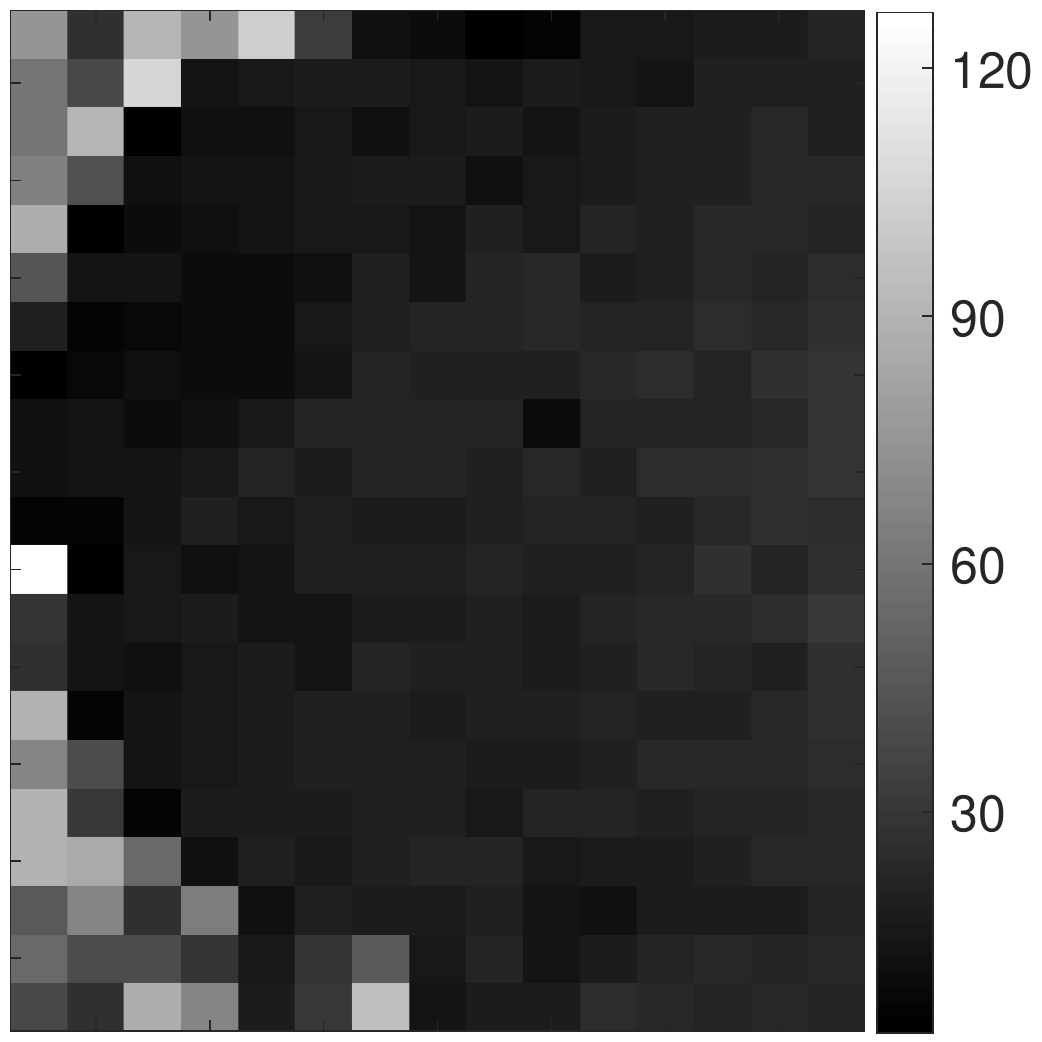}}
\centering
\subfigure[]{
\includegraphics[width=.15\linewidth]{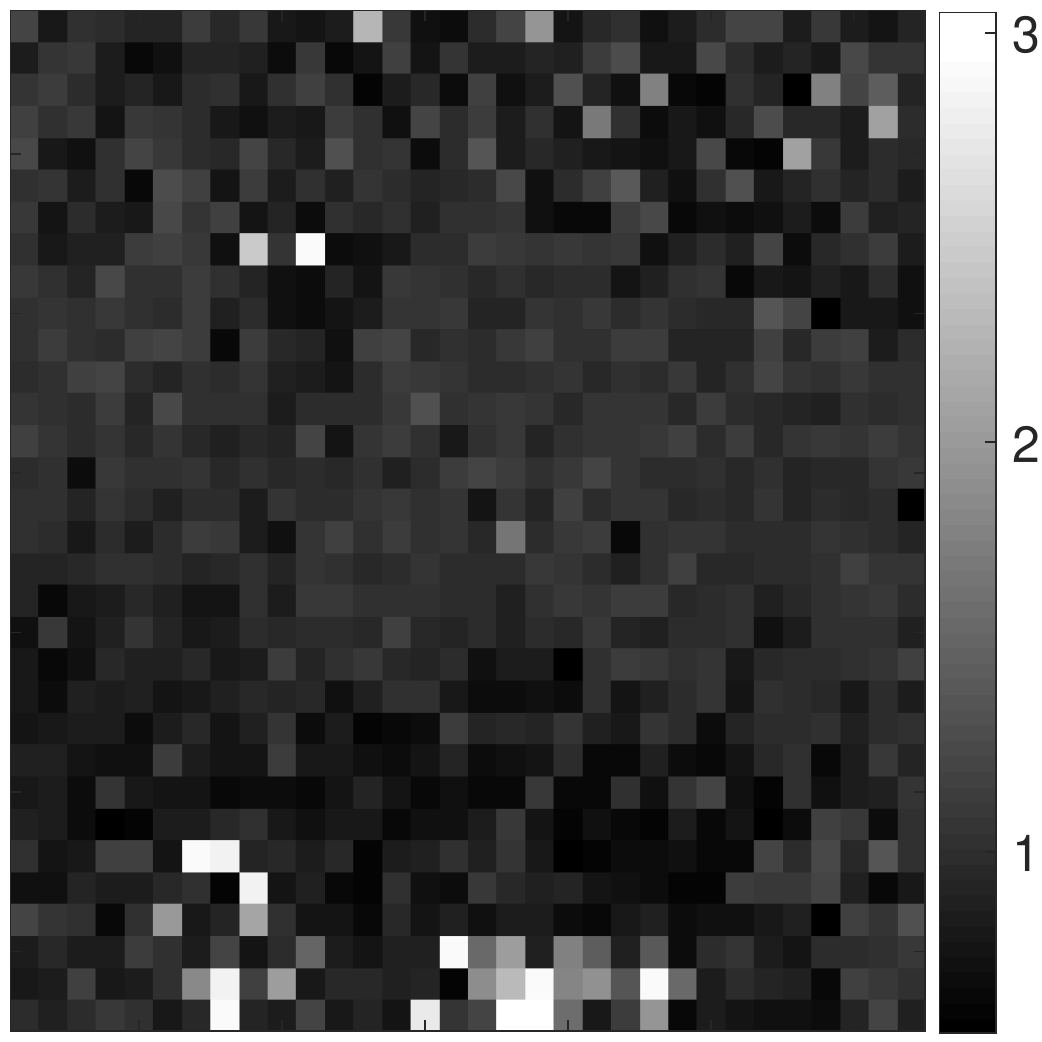}}
\centering
\subfigure[]{
\includegraphics[width=.15\linewidth]{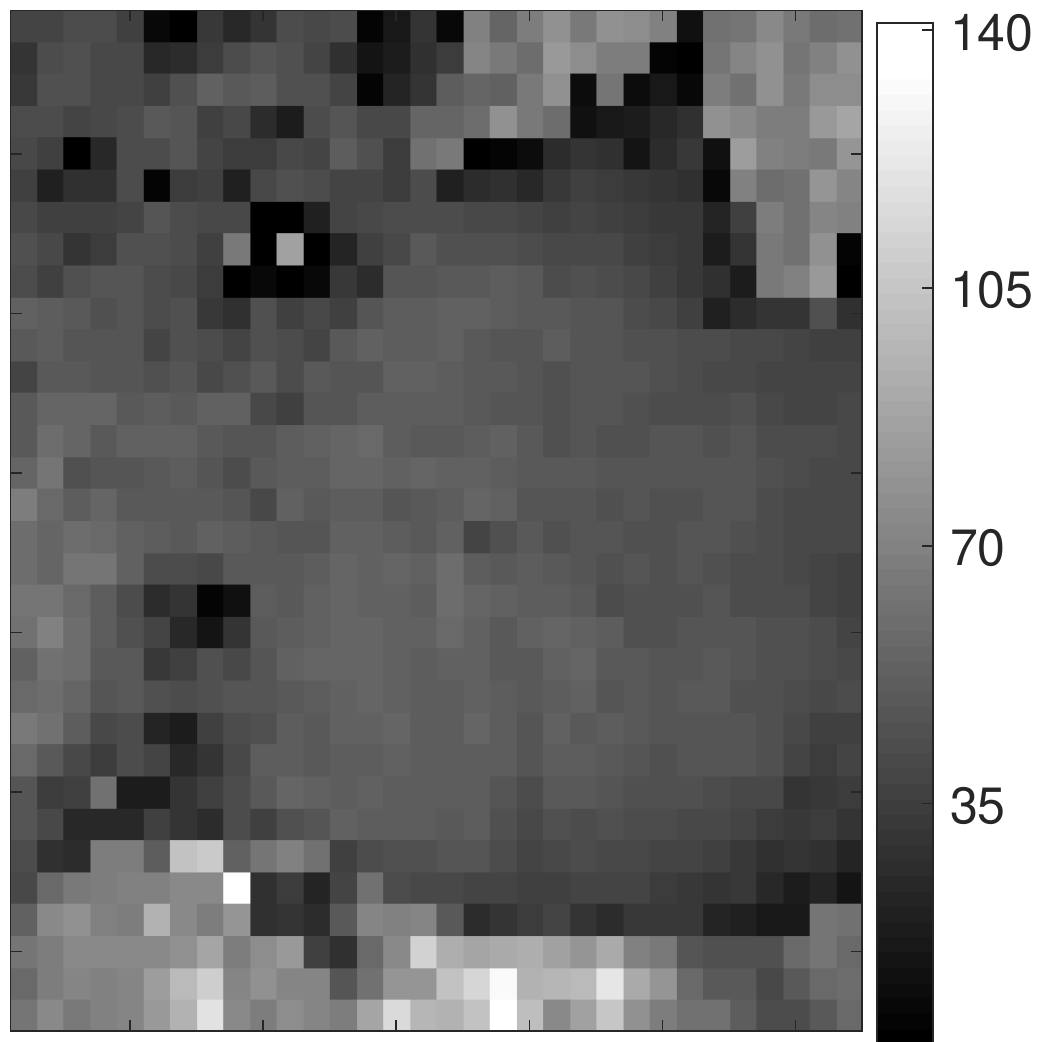}}
\subfigure[]{
\includegraphics[width=.15\linewidth]{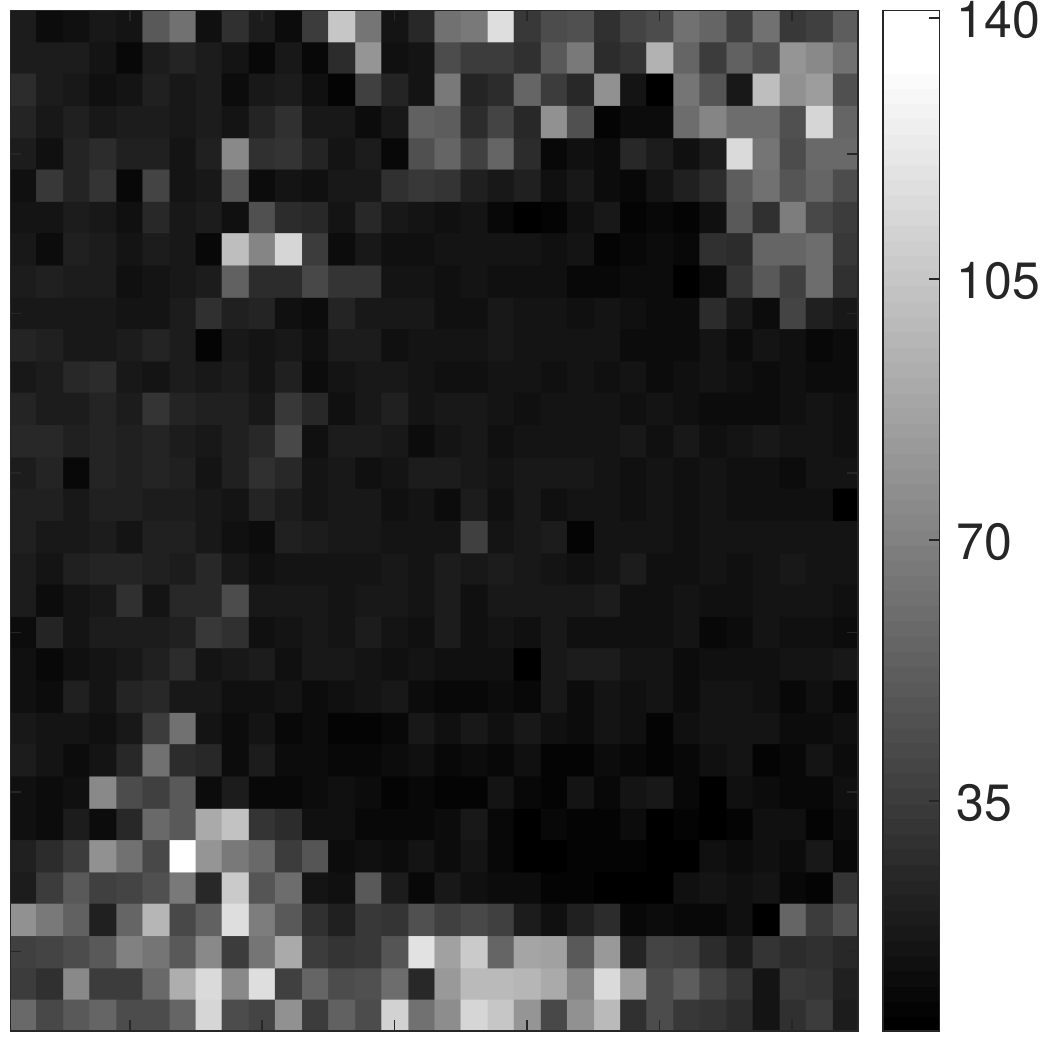}}
\caption{Estimated model parameters for two large SAR scenes for 250$\times$250 image patches in red rectangles. (a) Original Image-1 (5250$\times$3750). (b) Original Image-2 (8000$\times$8000). (c) and (f) The shape parameter $\alpha$. (d) and (g) The location parameter $\delta$. (e) and (h) The scale parameter $\gamma$.}
\label{fig:paramClass2}
\end{figure*}


Figure \ref{fig:rmse} depicts a visual representation to the modeling error metrics of RMSE, MAE, BD and AICc. Dark pixels show lower values, which also refer to better modeling results. When examining sub-figures Figure \ref{fig:rmse}-(a)-(c), we can easily state that for over all images, the proposed GG-Rician model is better than the state-of-the-art models. In Figure \ref{fig:rmse}-(d) since all AICc results are similar for all statistical models for a single image, for better visualisation, the difference AICc values are plotted
\begin{align}
    AICc_{dif}(i, j) = AICc_{best}(i) - AICc(i, j)
\end{align}
where $AICc_{dif}(i, j)$ is the AICc difference value for the location $(i,j)$ on Figure 6-(d) with the number of images is $i = 1, 2, \dots, 43$ and the number of models is $j = 1, 2, \dots, 8$. $AICc(i, j)$ is the AICc value for image $i$ and statistical model $j$. $AICc_{best}(i)$ is the highest AICc value for the image $i$. For each image, the best performing model has the darkest colour, since $AICc_{dif}$ will be 0 for that model. When we examine the $AICc_{dif}$ results, it is obvious that the Lognormal model appears to perform the best. The primary reason for this is the number of parameters of the models (Lognormal has 2 whilst GG-Rician, $\mathcal{G}_0$ and G$\Gamma$D have 3 parameters). Considering the AICc index penalizes the number of parameters, getting such results is unsurprising. However, it is worth noting that in combination with all other performance results shown, we believe that the Lognormal model will still come after the GG-Rician model for general SAR applications in which only one parameter difference do not play a crucial role.

Figure \ref{fig:realFigures} presents SAR images for six different scenes and their modeling results in logarithmic scale. The log-scale pdf modeling results in Figure \ref{fig:realFigures} confirm the numerical results presented in Figures \ref{fig:klks} and \ref{fig:rmse}, whereby the GG-Rician model outperforms most of the reference models utilized in this study.

In the results presented in Figure \ref{fig:klks}, it can be seen that KL divergence results for the GG-Rician model is somehow worse than that of KS scores. To analyze the reason behind this performance, we show two examples in Figure \ref{fig:KLKS}. The example comparison plot in Figure \ref{fig:KLKS}-(a) shows a case where the Lognormal model (it is clear from the examples in Figure \ref{fig:realFigures} that the same analysis holds for cases where $\mathcal{G}_0$ and G$\Gamma$D perform the best in a specific SAR scene) achieves the best results in terms of KL divergence values whilst Figure \ref{fig:KLKS}-(b) demonstrates a case where the GG-Rician model is the best for both KL and KS. When examining Figure \ref{fig:KLKS}-(a), for the data points in the rectangle,  Lognormal shows a closer fit than GG-Rician. Since the data has limited amount of low amplitude pixels, the GG-Rician estimate becomes rough. Despite its worse fit for the lower amplitude tail, GG-Rician provides a better model for the high probability region in ellipse and the right tail. Since KL scales the weights according to the relative entropy while KS looks at the cumulative distribution, these low amplitude areas with small number of pixels seem to have affected the results unevenly.

In order to give an example to support this reasoning, we depict the Figure \ref{fig:KLKS}-(b) which, compared to Figure \ref{fig:KLKS}-(a),  has higher probabilities  for dark (low-amplitude) pixels  than the bright (high-amplitude) ones, which is better described by the GG-Rician model in both tails and the main lobe. However, we can see that the Lognormal model fails to model both tails and the higher probability region at the same time. This characteristic of the proposed GG-Rician model provides the reason why it is having difficulties to have lower KL divergence values for the urban and the mountain scenes (see Figure \ref{fig:klks}-(a)), which are generally bright and obtain less darker radar returns. Finally, the same effects as in Figure \ref{fig:KLKS}-(a) and (b) can also be seen in Figure \ref{fig:realFigures}-(d) and (f), respectively.

To summarize the second simulation scenario, we performed a scoring mechanism to provide an overall quantitative measure on the performance of all models for 43 images. The scoring mechanism combines all seven performance results (KL, KS, $p$-value, RMSE, MAE, BD and AICc) into a single evaluation, and decides which model performs the best for a given SAR scene. When examining Figure \ref{fig:score}, the GG-Rician model's robust performance becomes clearer when compared to state-of-the-art models. 

\subsection{Analysis of Estimated GG-Rician Model Parameters}
In the third set of simulations, we analyzed the variations of the estimated GG-Rician model parameters depending on the different surface characteristics in a single large SAR scene.

We chose two example amplitude SAR images, each of which are TerraSAR-X products for sea surface with and without ships, and their corresponding wakes (Figure \ref{fig:paramClass2}-(a) and (b)). We believe that SAR images of this type include several distinct structures, such as land/mountain, urban area, sea, ships, shorelines, some islands, and even agricultural, which are suitable for the analysis in this simulation case. Each large image were decomposed into 250$\times$250 pixel patches and each patch was modeled via the proposed GG-Rician model. For each patch, we estimated model parameters and we plot them as images in Figure \ref{fig:paramClass2}-(c)-(h).

When examining shape parameter estimations for both images, we could state that areas including bright radar returns such as mountain tops, buildings, have relatively high shape parameter estimates, e.g. around $\alpha$ estimates of 2-3. For the sea surface, we can conclude that the shape parameter estimates do not directly reflect the changes of the sea surface, the estimated values of which generally lie around 1-2.

In Figures \ref{fig:paramClass2}-(d) and (g), we show the  estimated location parameter, $\delta$, for two example images. Examining these sub-figures, it can be stated that the location parameter estimates reflect a direct relation with the original image amplitude values, and provide a so-called good ``down-sampled version" of the original image. Different wave heights, shore-lines, as well as bright amplitudes such as urban areas are clearly distinguishable. As a feature, the location parameter $\delta$ of the proposed model can play an important role in classification tasks involving amplitude SAR images. Please also note that the location parameter ($\delta$) only exists for the Rician and Lognormal distributions, which are outperformed by the proposed GG-Rician model for all kinds of scenes and frequency bands, as discussed in the previous sets of simulations.

The estimated scale parameters for each patch are depicted in Figure \ref{fig:paramClass2}-(e) and (h). Both sub-figures generally show similar characteristics to the shape parameter estimates. For bright radar returns, it takes a $\gamma$ of around 100-140. The sea and land regions can be easily distinguished according to the estimated scale parameter values, whilst the sea surface changes are not distinguishable based on the $\gamma$ estimates.

Please note that, different from most of the statistical models, the GG-Rician model includes a location parameter. As a remark of this analysis, we can state that it is the most suitable parameter to reflect the radical changes on the SAR scene, and can be shown as an important advantage of the proposed statistical model, especially in applications such as segmentation and classification.

\subsection{Real Intensity SAR Data}
In the fourth and the last simulation experiment, we evaluated the performance of the intensity GG-Rician model. For this simulation, we only utilized three example intensity SAR images. We left the further analysis of the intensity model as future work.

The same procedure used for the amplitude SAR modeling was also applied here for intensity images. The same parameter estimation methodology was used to estimate the model parameters of the intensity GG-Rician model. Due to their ability to model intensity images, the Weibull, $\mathcal{G}_0$, G$\Gamma$D, and the Gamma distribution were used as reference.

\begin{figure}[htbp]
\centering
\subfigure[]{
\includegraphics[width=.3\linewidth]{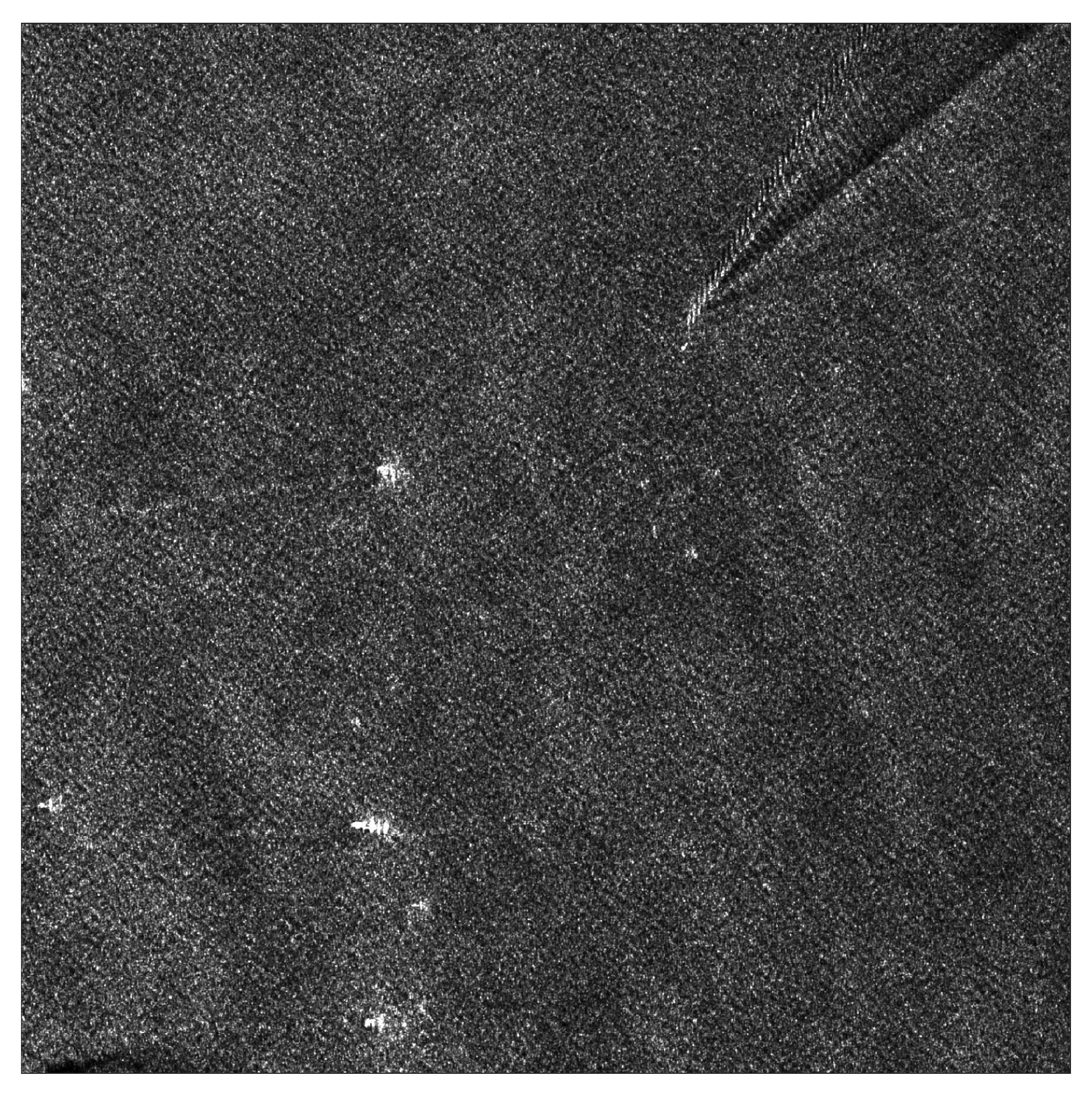}}
\subfigure[]{
\includegraphics[width=.33\linewidth]{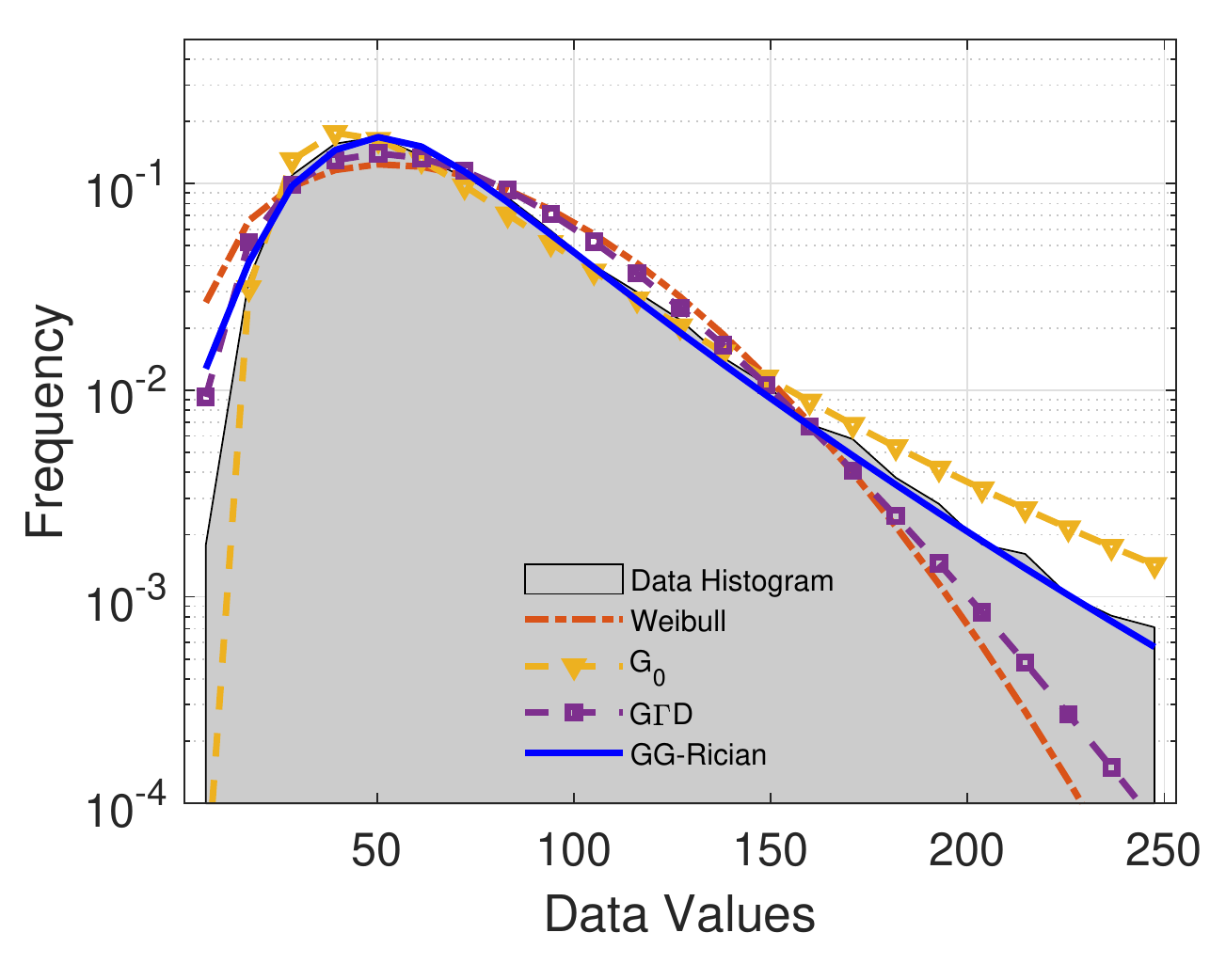}}
\centering
\subfigure[]{
\includegraphics[width=.3\linewidth]{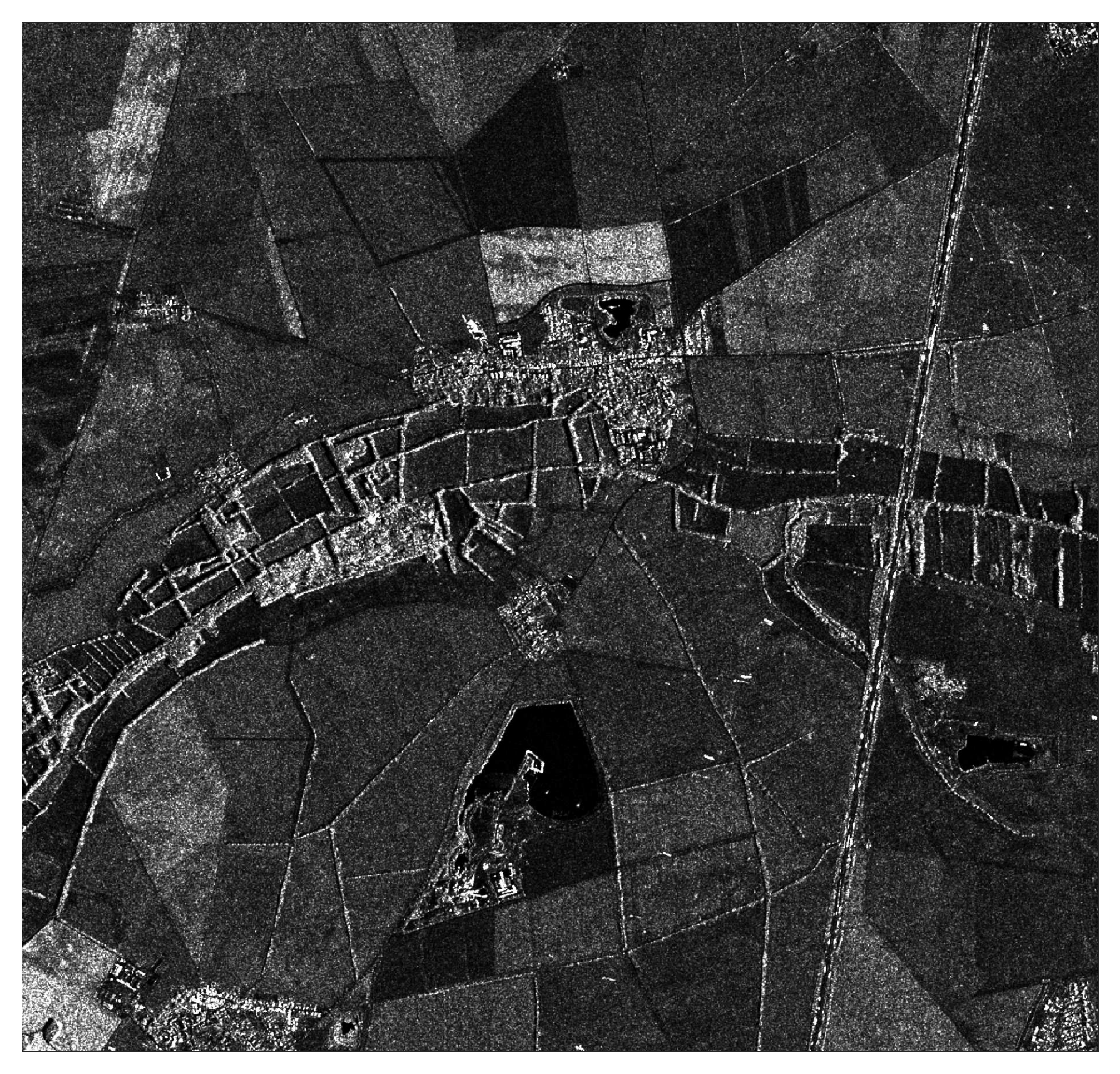}}
\subfigure[]{
\includegraphics[width=.33\linewidth]{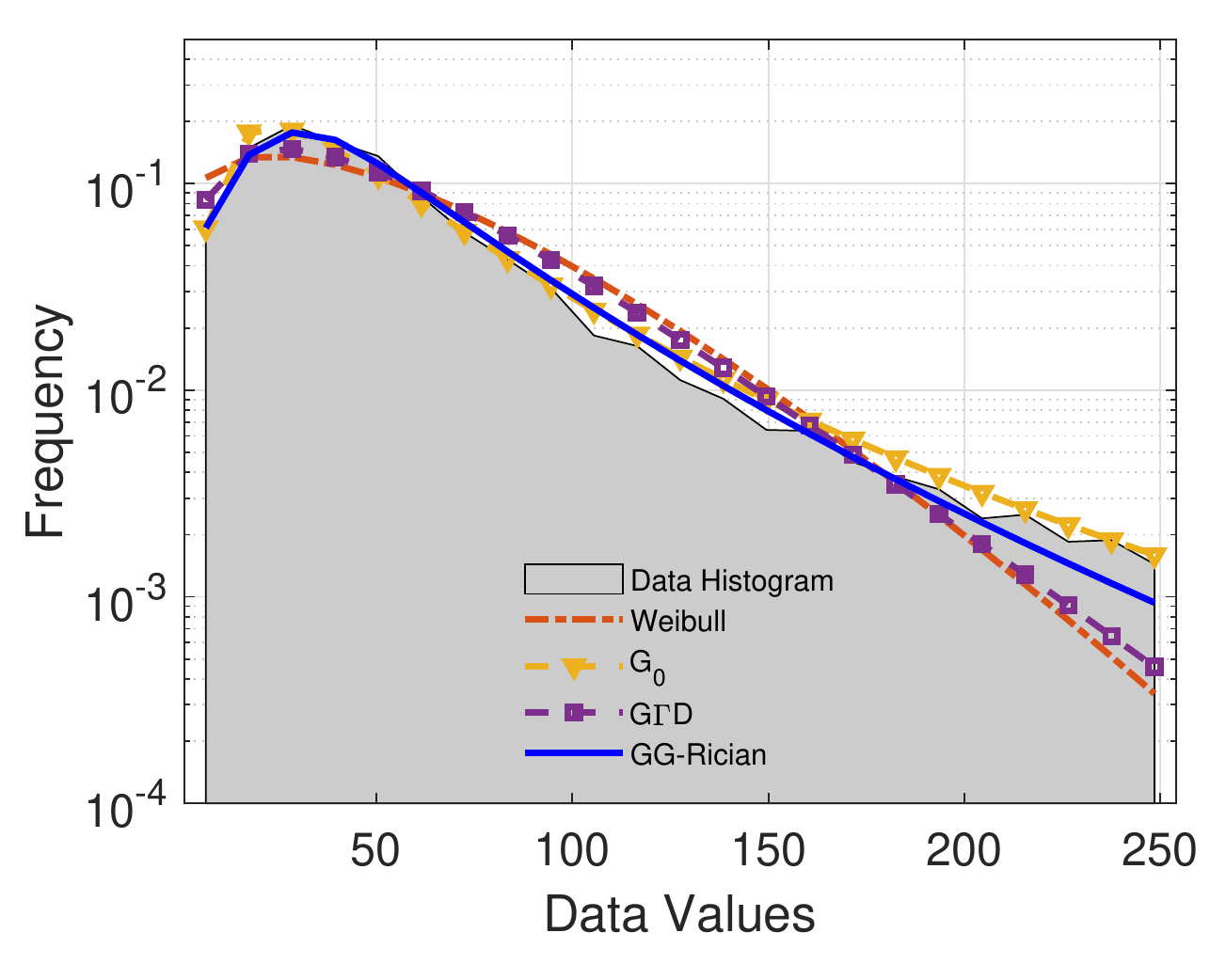}}
\centering
\subfigure[]{
\includegraphics[width=.3\linewidth]{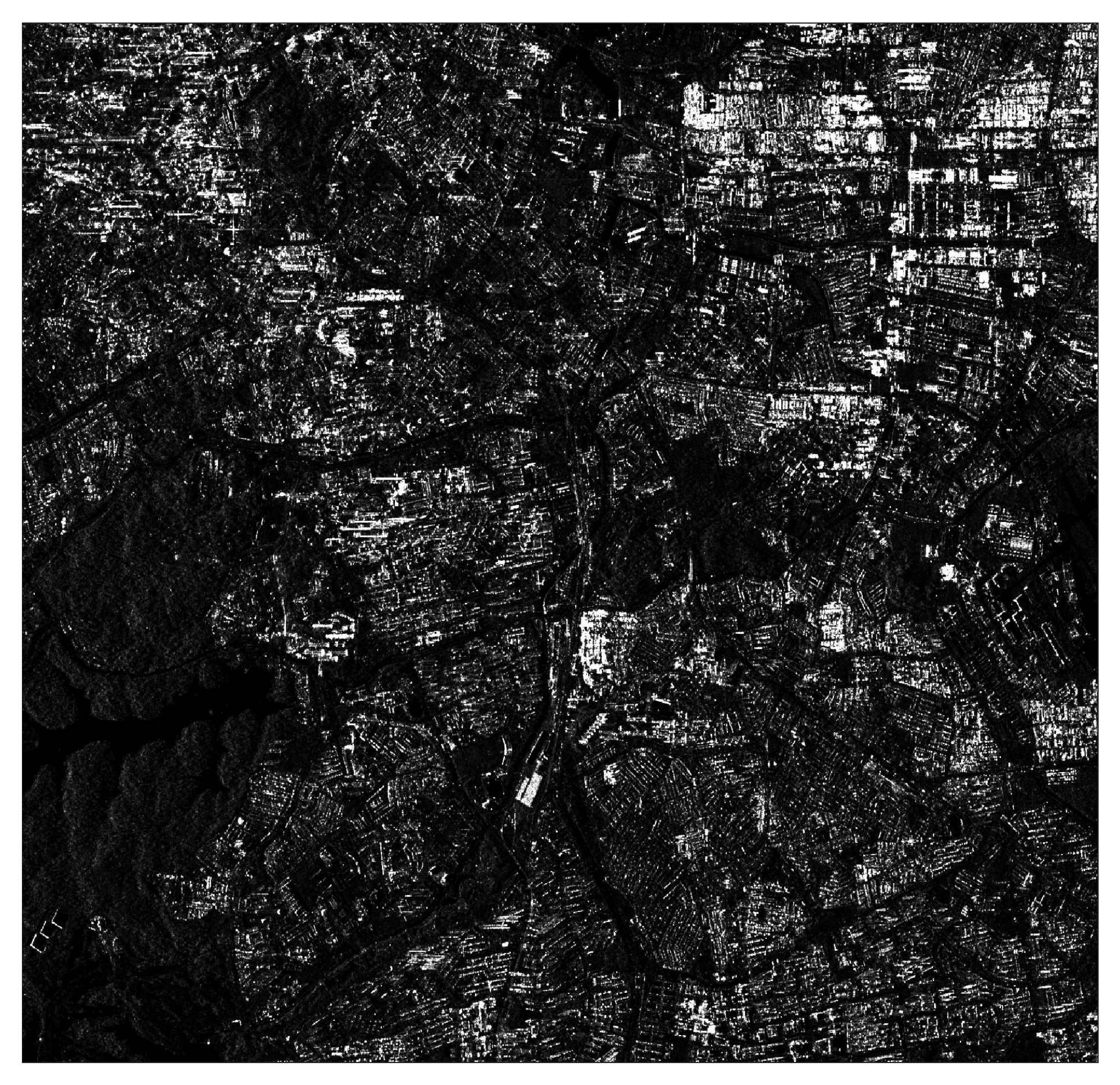}}
\subfigure[]{
\includegraphics[width=.33\linewidth]{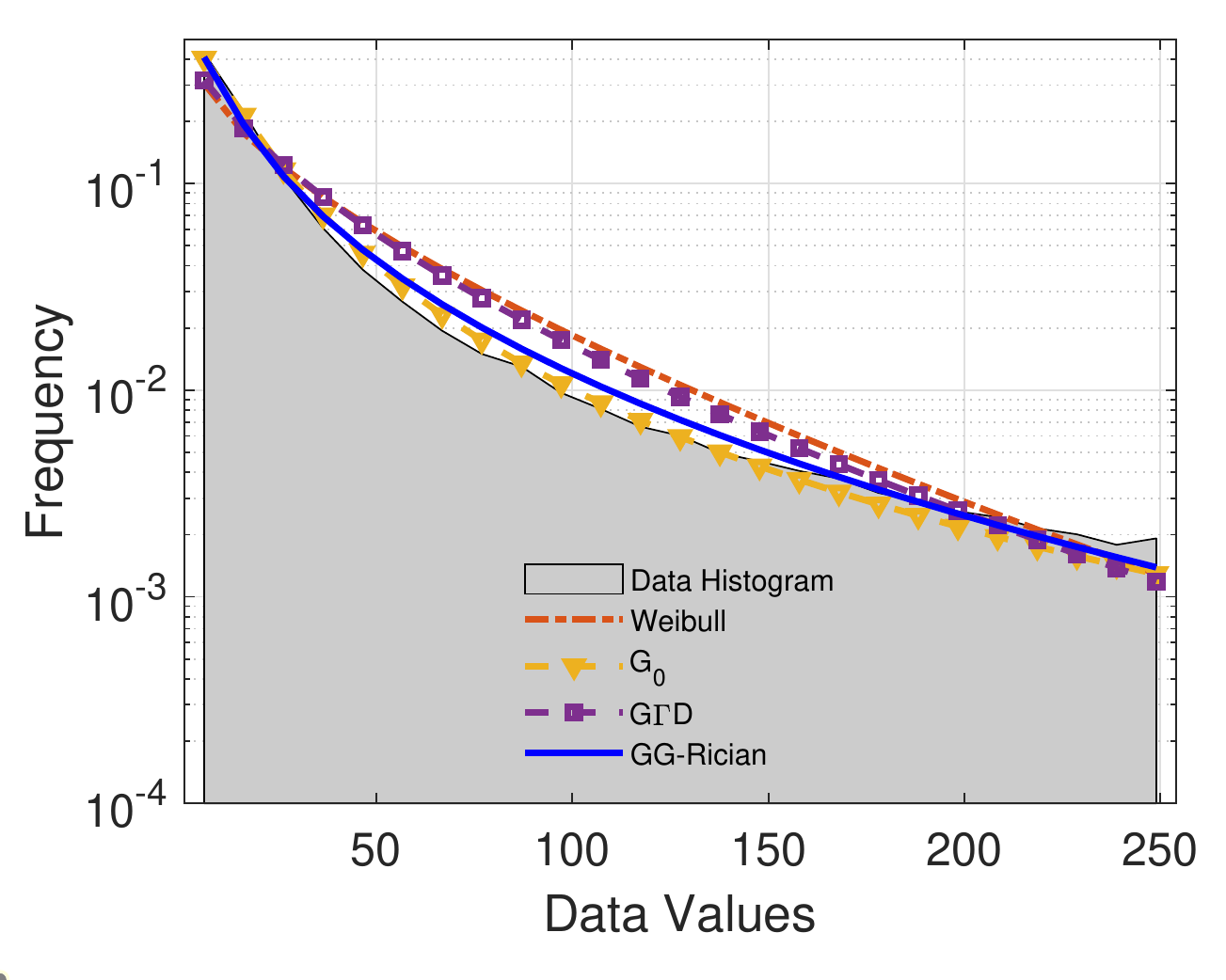}}
\caption{Intensity modeling comparison. Intensity SAR images from scenes of (a) sea with ships, (c) mixed, and (e) urban . Sub-figures in (b), (d) and (f) refer to the corresponding modeling results in log-pdf scale for intensity images in (a), (c) and (e), respectively.}
\label{fig:intensity}\vspace{-0.4cm}
\end{figure}

The corresponding results are depicted in Figure \ref{fig:intensity}, and statistical significance of modeling is given in Table \ref{tab:intensity}. There are three SAR images of scenes: Sea with ships, Mixed and Urban, respectively in Figure \ref{fig:intensity}-(a), (c) and (e). The corresponding modeling results are presented in logarithmic scale in sub-figures (b), (d) and (f).
The superior modeling performance of the proposed GG-Rician intensity model compared to the state-of-the-art models is obvious specifically for sea with ships and mixed scenes, whereas the $\mathcal{G}_0$ appears to be the best performing model for urban scene.
\begin{table}[htbp]
  \centering
  \caption{Statistical significance for intensity SAR image modeling}
     \resizebox{0.5999\linewidth}{!}{\begin{tabular}{cccccc}
    \toprule
    \multicolumn{1}{c}{Image} & Stats & \multicolumn{1}{c}{Weibull} &
    \multicolumn{1}{c}{$\mathcal{G}_0$} &
    \multicolumn{1}{c}{G$\Gamma$D} & \multicolumn{1}{c}{GG-Rician} \\
    \toprule
   \multirow{3}[0]{*}{wSea} & KL Div & 0.0756 & \textbf{0.0145} & 0.0258 & 0.0154 \\
          & KS Score & 0.0652 & 0.0364 & 0.0551 & \textbf{0.0226} \\
          & $p$-value & 0.2712 & 0.9316 & 0.4734 & \textbf{0.9998} \\
          \hline
    \multirow{3}[0]{*}{Mixed} & KL Div & 0.0755 & 0.0157 & 0.0415 & \textbf{0.0064} \\
          & KS Score & 0.0896 & 0.0461 & 0.0777 & \textbf{0.0317} \\
          & $p$-value & 0.0466 & 0.4730 & 0.1187 & \textbf{0.9727} \\
          \hline
    \multirow{3}[0]{*}{Urban} & KL Div & 0.1084 & \textbf{0.0054} & 0.0843 & 0.0144 \\
          & KS Score & 0.1592 & \textbf{0.0552} & 0.1452 & 0.0701 \\
          & $p$-value & 0.0000 & \textbf{0.4194} & 0.0000 & 0.1642 \\
          \bottomrule
    \end{tabular}}%
  \label{tab:intensity}%
\end{table}%
\section{Conclusions}\label{sec:conc}
In this paper, we proposed a novel parametric statistical model, namely the GG-Rician distribution, to characterize the amplitude and the intensity of the complex back-scattered SAR signal. Specifically, the GG-Rician model is an extension of Rician model whereby the Gaussian components of the complex SAR signal are replaced by the generalized-Gaussian distribution. An expression in integral form was derived for the pdf and a Bayesian sampling scheme for the model parameter estimation was developed. We have tested the modeling performance of the GG-Rician model both on synthetically generated and real SAR data, which are specifically coming from satellite platforms of TerraSAR-X, Sentinel-1, ICEYE, COSMO/Sky-Med and ALOS2.

The performance of the proposed statistical model was then compared to state-of-the-art statistical models including the Rician, Weibull, Lognormal, $\mathcal{G}_0$, G$\Gamma$D, SR, and GGR. The results demonstrate that the proposed method achieves the best modeling results for most of the images, and outperforms state-of-the-art models for images from various frequency bands and sources. It is interesting to note that the results show the need for combining the advantages of non-Gaussian heavy-tailed modeling and non-zero mean reflections modeling provided by the Rician model. 
Furthermore, using non-zero reflections along with the heavy-tailed modeling of the GG-Rician model shows important success in all types of SAR scenes and frequency bands.

All  experimental analyses in this paper demonstrate that the extension from Rayleigh (zero-mean components) to Rician (non-zero mean components) offers clear advantages over the Rayleigh-based GGR in \cite{moser2006sar}. 
On the other hand, since the GGR model is actually a simplified special member (for $\delta = 0$) of the proposed GG-Rician model, we conclude that we generalized GGR to a more flexible and robust model that covers various characteristics.
The flexibility of adjusting distribution tails (via a shape parameter) in conjunction with the location parameter demonstrates a remarkable gain over the Rician model for all the simulation scenarios considered in this paper. As future work, we will investigate faster and closed form parameter estimation methods.

Finally, we would like to note two limitations of the proposed GG-Rician model: (i) the fact that an analytical parameter estimation method cannot be designed, and (ii) its relatively low modeling performance for extremely heterogeneous regions such as urban scenes. For the former, we should state that the MCMC based parameter estimation in this paper provides less-sensitivity to initial parameter values, whereby the estimation methods based on the method of log-cumulants (MoLC) for $\mathcal{G}_0$ and G$\Gamma$D might suffer from poor initialisation \cite{cui2011coarse}. On the other hand, instead of solving an optimisation problem for a system of nonlinear equations, the proposed parameter estimation method provides a simple sampling-based approach with high performance. To address the latter limitation, we believe that further analysis is required by focusing on heavier tailed statistical models than the GG distribution.


\section*{Acknowledgment}
We are grateful to the UK Satellite Applications Catapult for providing us the COSMO-SkyMed data sets employed in this study.

\bibliographystyle{IEEEtran}
\bibliography{template}

\end{document}